\documentclass[12pt]{article}
\usepackage[a4paper, margin=2.6cm]{geometry}
\usepackage{graphicx}
\usepackage{amsmath,amsfonts}
\usepackage{url}
\usepackage{booktabs}
\usepackage{multirow}

\usepackage{algorithm}
\usepackage{algpseudocode}
\usepackage{amssymb}
\usepackage{amsmath}
\usepackage{enumerate}
\usepackage{amsthm}

\usepackage{lineno}

\newtheorem{theorem}{Theorem}

\newtheorem{lemma}[theorem]{Lemma}
\newtheorem{observation}[theorem]{Observation}

\providecommand{\ch}{{\text{\sc Ch}}}
\providecommand{\rank}{\mathrm{rank}}
\providecommand{\col}{\mathrm{col}}
\providecommand{\val}{\mathrm{val}}
\providecommand{\res}{\mathrm{res}}
\providecommand{\deg}{\mathrm{deg}}
\providecommand{\mul}{\mathrm{mul}}
\providecommand{\code}{\mathrm{code}}

\providecommand{\rhoG}[1]{\rho_G(#1)}

\providecommand{\copyv}{\mathrm{copy}}
\providecommand{\Ch}{{\text{\sc Ch}}}
\providecommand{\freq}{\mathrm{frq}}
\providecommand{\mm}{\mathrm{M}}

\providecommand{\pathset}{\mathcal{P}}
\providecommand{\fvset}{\mathbf{f}}
\providecommand{\parent}{\mathrm{p}}

\algrenewcommand\algorithmicindent{2mm}%

\begin{document}

\title{Enumerating Chemical Graphs with
Mono-block 2-Augmented Tree Structure 
from Given Upper and Lower Bounds on Path Frequencies}

\author{
Yuui Tamura$^1$ \and
Yuhei Nishiyama$^1$ \and
Chenxi Wang$^1$ \and
Yanming Sun$^1$ \and
Aleksandar Shurbevski$^1$ \and
Hiroshi Nagamochi$^1$ \and
Tatsuya Akutsu$^2$ \\ \\
$^1$~Department of Applied Mathematics and Physics, \\
Kyoto University, Kyoto 606-8501, Japan. \\ 
{\normalsize \{y.tamura, y.nishiyama, chenxi, 
	    sun.ym,  shurbevski, nag\}@amp.i.kyoto-u.ac.jp}  \and
$^2$~Bioinformatics Center,  
Institute for Chemical Research, \\
Kyoto University, Uji 611-0011, Japan.\\
{\normalsize takutsu@kuicr.kyoto-u.ac.jp} 
} 

\maketitle              
\thispagestyle{plain}

\begin{verse}
~~~~{\bf Abstract. }
We consider a problem of enumerating chemical graphs 
 from given constraints concerning their structures, 
 which has an important application
to a novel method for the inverse QSAR/QSPR recently proposed.
In this paper, the structure of a chemical graph is specified 
by a feature vector each of whose entries represents the frequency
of a prescribed  path. 
We call a graph a 2-augmented tree if it is obtained from a tree (an acyclic graph)
by adding edges between two pairs of non-adjacent vertices.  
Given a set of feature vectors as the interval between upper and lower bounds of feature vectors, 
we design an efficient algorithm for enumerating chemical 2-augmented trees
 that satisfy the path frequency specified by some feature vector  in the set. 
 We implemented the proposed algorithm and conducted some computational
 experiments.  

\end{verse}

\section{Introduction}\label{sec:introduce}

Development of novel drugs is one of the major goals
 in chemoinformatics and bioinformatics.
To achieve this purpose, it is important not only to investigate common
 chemical properties over chemical compounds having certain structural 
 pattern~\cite{BK99_formula,BK99_isomer,MS07},
 but also to enumerate all the chemical compounds having a specified structural pattern.
The enumeration of chemical compounds has a long history, which can be traced back to 
Cayley~\cite{cayley1875}, who addressed the enumeration of structural 
isomers of alkanes in the 19th century.

A multigraph is a graph that can have multiple edges between the same pair of vertices,
where multiple edges  represent double bonds or triple bonds
 in a chemical compound.  
Let us call a multigraph
 a {\em $k$-augmented tree} if it is connected and it becomes a tree possibly
  with multiple edges
 after removing edges between $k$ pairs of adjacent vertices,
 where a $0$-augmented tree and a $1$-augmented tree 
  are also called an {\em acyclic} graph
 and a {\em monocyclic} graph, respectively. 
In the $97,092,888$  chemical compounds in
the  PubChem database, 
the ratio of the number of chemical compounds of a $k$-augmented tree structure 
to that of  all registered   chemical compounds
is around $2.9\%$, $13.3\%$, $28.2\%$, $24.2\%$  and $16.0\%$   
for $k=0, 1, 2, 3$ and $4$, respectively.  
 
Quantitative Structure Activity/Property Relationships (QSAR/QSPR) analysis
is a major approach for computer-aided drug design.
In particular, inverse QSAR/QSPR plays an important role
\cite{Miyao16,Skvortsova93}, 
which is to infer chemical structures from given chemical
activities/properties.
As in many other fields,
Artificial Neural Network (ANN) and deep learning technologies
have recently been applied to inverse QSAR/QSPR.
In these approaches, new chemical graphs are generated
by solving a kind of inverse problems on neural networks,
where neural networks are trained using
known chemical compound/activity pairs.
However, there was no mathematical guarantee for the existence of
solutions in these approaches.
In order to solve the inverse problem mathematically, 
a novel approach has been proposed by Akutsu and Nagamochi~\cite{AN19}
for ANNs,  
using mixed integer linear programming (MILP).

Recently 
Chiewvanichakorn et~al.~\cite{CWZSNA20}, 
Azam et~al.~\cite{ACZSNA20,IWSNA20} proposed
a novel framework for the inverse  QSAR/QSPR by 
combining the MILP-based formulation of the inverse problem on ANNs~\cite{AN19}
and  
efficient enumeration of  acyclic graphs 
and monocyclic graphs.   
This combined framework for  inverse QSAR/QSPR mainly consists of two phases. 
The first phase defines a function $f$ that converts each chemical graph $G$
into a feature vector $f(G)$ that consists of several descriptors on
the structure of $G$ and then 
solves (1) {Prediction Problem},  
where 
a prediction function $\psi_{\mathcal{N}}$ on a chemical property $\pi$
is constructed with an ANN $\mathcal{N}$ 
using a data set of  chemical compounds $G$ and their values $a(G)$ of $\pi$. 
The second phase solves (2) {Inverse Problem}, 
where (2-a) given a target value $y^*$ of the chemical property $\pi$,
a feature vector $x^*$ is inferred from the trained ANN  $\mathcal{N}$
so that  $\psi_{\mathcal{N}}(x^*)$ is close to  $y^*$ 
and (2-b) then a set of chemical structures $G^*$
such that $f(G^*)= x^*$  is enumerated.  
Methods applied to the case of inferring acyclic or monocyclic chemical graphs have been
implemented as computer programs, through which 
chemical graphs $G^*$ are  inferred from given target values $y^*$ of 
actual chemical properties such as  
heat of atomization,
heat of formation, 
boiling point and 
octanol/water partition coefficient~\cite{ACZSNA20,CWZSNA20,IWSNA20}.
In this framework, an efficient algorithm for enumerating
acyclic or monocyclic graphs that satisfy 
given descriptors is an important building block to solve (2-b).
A natural next target to apply the framework for the inverse QSAR/QSPR is to construct
 a system of inferring chemical 2-augmented trees.
To attain this, we design  an efficient algorithm for enumerating
chemical 2-augmented trees. 

Some useful tools such as MOLGEN~\cite{MOLGEN5}, 
OMG~\cite{OMG}, and similar, 
have been developed and are available for enumeration of chemical graphs.
However, they are not always very efficient in enumerating chemical graphs
that satisfy a given condition on structures such 
as frequency of certain types of subgraphs, 
because they treat general graph structures.
In particular, it is known that the number of
molecules (i.e., chemical graphs) with up to 30 atoms (vertices)
{\tt C}, {\tt N}, {\tt O}, and {\tt S},
may exceed~$10^{60}$~\cite{BMG96}.

Fujiwara~{\it et al.}~\cite{Fujiwara08} and  Ishida~{\it et al.}~\cite{IZNA08} 
studied the enumeration of acyclic chemical graphs 
that satisfy a given feature vector which specifies 
the frequency of all paths of up to 
a prescribed length in a chemical compound to be constructed.
Their results
have been recognized as establishing new methodologies 
in this track of research in chemoinformatics~\cite{VB12}.
Instead of giving  a single feature vector $f$ on the frequency of prescribed paths, 
Shimizu~{\it et al.}~\cite{SNA11} treated a set $F$ of 
feature vectors on path frequency given as  the set of all vectors
between a pair of upper and lower feature vectors, and designed 
a branch-and-bound algorithm  of enumerating 
acyclic chemical graphs  each of which satisfies some feature vector $f$ in the set $F$.
Afterward Suzuki~{\it et al.}~\cite{SNA12} proposed 
an improved and more efficient algorithm.
For monocyclic graphs,
Suzuki~{\it et al.}~\cite{Suzuki14} proposed an efficient algorithm that constructs
a monocyclic chemical graph by adding an edge to an acyclic chemical graph.
Such acyclic chemical graphs in turn, can be obtained by an existing algorithm~\cite{Fujiwara08,SNA12}.
The above-mentioned algorithms for enumerating acyclic and monocyclic chemical graph
with given path frequencies now play a crucial role in the novel methods for 
 inverse QSAR/QSPR~\cite{ACZSNA20,CWZSNA20,IWSNA20}.

For 2-augmented trees, we distinguish two types: (i) those with two edge-disjoint cycles
and (ii) those with a single bi-connected component, where every two cycles share an edge.
We call a 2-augmented tree in type (ii)  a {\em mono-block 2-augmented tree}. 
In this paper, we design an algorithm for 
enumerating mono-block 2-augmented trees that satisfy 
given upper and lower bounds of path frequencies of graphs.
We implemented the proposed algorithm and conducted some computational
 experiments.

\section{Preliminaries on Graphs}\label{sec:graph}

This section reviews some basic definitions on graphs 
and introduces the notion of chemical graphs as used in this paper.

\subsection{Multigraphs}
Let $\mathbb{Z}_{+}$ denote the set of positive integers.
For two integers $a$ and $b$,
let $[a,b]$ denote the set of all integers $i$ with $a \leq i \leq b$.

A {\em graph} is defined to be an ordered pair $(V,E)$ 
of a finite set $V$ of vertices and a finite set $E$ of edges.
In this paper, we do not consider self-loops, and 
an edge in $E$ joining two vertices  $u,v\in V$ is  denoted by $uv$.

Let $G$ be a graph.
We denote the vertex set and the edge set of $G$ by $V(G)$ and~$E(G)$, respectively.
An ordered pair $(V',E')$ of subsets $V' \subseteq V(G)$ and 
$E' \subseteq E(G)$ is called a {\em subgraph} of $G$ if $(V',E')$ forms a graph, 
i.e., $\{u,v\in V\mid uv\in E'\}\subseteq V'$.
We say that a subset $X \subseteq V(G)$ {\em induces} 
a subgraph $G'$ if $V(G') = X$ and $E(G')$ 
contains every edge in $E(G)$ between two vertices in $X$.

We call a graph where each vertex and edge 
has a unique name or an index a {\em labeled graph}.
Throughout the paper,
graphs are considered to be labeled, to distinguish or enumerate vertices, 
edges or some other structures in a graph.
Figure~\ref{fig:labeled-graph}(a) shows an example of a labeled graph.
A {\em rooted} graph is a graph in which either a vertex 
or an edge is designated as a root.

\begin{figure}
\centering
\includegraphics[width=0.95\linewidth]{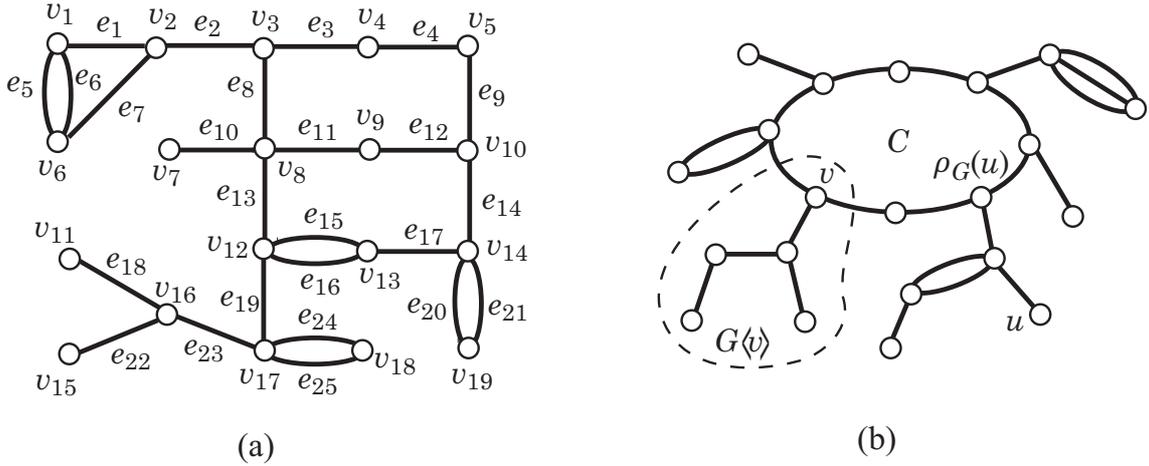}
\caption{
(a)~An example of a labeled graph, each vertex is labeled 
with a unique label $v_{i}$, $i = 1, 2, \dots, 19$, 
and each edge with a unique label $e_{j}$, $j=1, 2, \dots, 25$,
(b)~A monocyclic graph $G$ with a unique cycle $C$,
where  
the pendent tree $G\langle v \rangle$ for a vertex $v\in V(C)$
 is the subtree enclosed by a dashed line.} 
\label{fig:labeled-graph}
\end{figure}

A graph is called a {\em multigraph} when there can be more than 
one edge between the same pair of endvertices.
Let $G$ be a multigraph.
For a vertex $v \in V(G)$, we call the number of 
edges incident to $v$ the {\em degree} of $v$ and denote it by ${\rm deg}(v)$.
Let~$\{u,v\} \subseteq V(G)$.
The \emph{multiplicity}, i.e., the number of edges, 
between two vertices $u$ and $v$ is denoted by ${\rm mul}_{G}(u,v)$, 
where for two non-adjacent vertices $u$ and $v$,  it holds ${\rm mul}_{G}(u,v) = 0$.
We refer to the set of ${\rm mul}_{G}(u,v)~(\geq 0)$ simple edges 
with endvertices $u$ and $v$ as a {\em multiple edge}
with multiplicity~${\rm mul}_{G}(u,v)$.
Let $\overline{E}(G)$ denote the set of pairs $\{u,v\}\subseteq V(G)$ 
with ${\rm mul}_{G}(u,v) = 0$.
For a pair of adjacent vertices $u$ and $v$, let $G - uv$ denote the graph 
$G'$ obtained by removing~${\rm mul}_{G}(u,v)$ simple edges 
between $u$ and $v$ from $G$.
Conversely, let $G + q\cdot uv$ denote the graph~$G'$ obtained by 
adding $q \in \mathbb{Z}_{+}$ simple edges between $u$ and $v$, 
i.e., ${\rm mul}_{G'}(u,v) = {\rm mul}_{G}(u,v) + q$.
In particular, we denote~$G + 1\cdot uv$ by $G + uv$.
If $G$ is clear from the context, then we denote ${\rm mul}_{G}(u,v)$ 
by ${\rm mul}(u,v)$.

For a nonnegative integer $k$, a graph $P$ which consists of $k + 1$
distinct vertices~$v_{0},v_{1},$ $\dots,$ $v_{k}$ and $k$ 
multiple edges $v_{i}v_{i+1},~i \in [1,k-1]$, 
is called a {\em path} (or a path of length $k$), 
and is denoted by $P = (v_{0}, v_{1}, \ldots, v_{k})$.
The length, i.e., the number of edges in a path $P$ is also denoted by~$|P|$.
A graph $C$ which consists of a path~$(v_{0}, v_{1}, \ldots, v_{\ell})$ 
of length $\ell \geq 2$ 
and a multiple edge between $v_{\ell}$ and $v_{0}$ is called a {\em cycle} 
and is denoted by $C = (v_{0}, v_{1}, \ldots, v_{\ell}, v_{0})$.
In this paper, a graph which consists of two vertices and
two edges between them is not considered as a cycle.

A {\em block} in $G$ is defined to be a maximal vertex subset $X\subseteq V(G)$
such that for any two vertices $u,v\in X$,
there is a cycle of $G$ that passes through $u$ and $v$. 
We call $G$ a {\em mono-block graph} if $G$ has exactly 
one block $X$ with $|X|\geq 2$.

A connected multigraph with no cycles is called a {\em multitree}.
Every multitree $T$ has either a vertex $v$ or
an adjacent vertex pair $\{v,v'\}$ 
removal of which  leaves no connected component 
with more than $\lfloor|V(T)|/2 \rfloor$ vertices~\cite{J69}. 
Such a vertex or an adjacent vertex pair is called a {\em centroid}, 
where  a centroid $v$ is called a {\em unicentroid} and
a centroid $\{v,v'\}$  is called a {\em bicentroid}.

Let $T$ be a rooted multitree and $v \in V(T)$.
Let $u \in V(T)$ be a vertex such that $v$ is on the 
unique path between $u$ and the root.
We call $u$ a {\em descendant} of $v$ and call $v$ an {\em ancestor} of~$u$.
In particular, if $u$ and $v$ are adjacent, 
we call $u$ a {\em child} of $v$ and call $v$ the {\em parent} of $u$.
The parent of $v$ is denoted by ${\rm p}(v)$.
The set of children of a vertex $v$ is denoted by $\Ch(v)$.
The {\em depth} of $v$ represents the length of 
the unique path between~$v$ and the root, 
and is denoted by ${\rm d}(v)$.
If~$v$ is the root vertex or an endvertex of the root edge, 
then $v$ has no parent, and ${\rm d}(v) = 0$.
We denote by $T_{v}$ the subtree of $T$ induced 
by $v$ and the set of descendants of~$v$.
For an edge $vw \in E(T)$ such that $w = {\rm p}(v)$, 
we denote by $T_{wv}$ the subtree of $T$ induced by~$w$, $v$, 
and the set of descendants of~$v$.
That is, $T_{wv}$ consists of the subtree $T_{v}$ and the vertex $w = {\rm p}(v)$
joined by a multiple edge between $v$ and $w$ with multiplicity~${\rm mul}_{T}(v,w)$.
We regard $T_{wv}$ to be rooted at $w$.

For a connected multigraph $G$ with at least one cycle
and a vertex~$v \in V(G)$ such that $v$ is included in some of the cycles in $G$, 
 the {\em pendent tree} $G\langle v \rangle$ of vertex~$v$
is defined to be the subgraph $T$ of $G$ induced by $v$
and the set of vertices reachable from $v$
without passing through any edge in a cycle of $G$,
where $T$ becomes a tree, possibly only consisting of vertex~$v$. 
We treat~$G\langle v\rangle$ as a tree rooted at $v$.
For a vertex $u \in V(G\langle v\rangle)$ we define
$\rhoG{u} = v$.
For convenience, for a pendent tree $T = G\langle v\rangle$
and a vertex $u \in V(G\langle v\rangle)$,
we denote the subtree $T_u$ of  $G\langle v\rangle$ rooted at 
$u$ by $G\langle u \rangle$.
In addition, for the parent $w = \parent(u)$ of $u$ in $G\langle v\rangle$,
we denote by $G \langle w, u \rangle$ the rooted tree $T_{wu}$.

An example of a pendent tree   
is illustrated in
Fig.~\ref{fig:labeled-graph}(b).
For a multigraph $G$ with $n$ vertices,
we say that a pendent tree of $G$ is \emph{exceeding} if it has at least~$n/3$ vertices.

\subsection{$k$-Augmented Trees}
A $k$-augmented tree with $n$ vertices is a connected 
multigraph such that the number of pairs of adjacent vertices is 
$(n-1)+k$, i.e., it is constructed from a multiree with $n$ vertices 
by adding edges between $k$ pairs of non-adjacent vertices. 
 
\subsubsection{1-Augmented Trees}
Let $G$ be a monocyclic graph, which has a unique cycle $C$. 
Throughout this draft we will also call 1-augmented trees 
\emph{monocyclic graphs}.
For a vertex $u \in V(G)$, let $\rhoG{u}$ denote 
the vertex $v \in V(C)$ such that $u \in V(G\langle v\rangle)$.
Note that in the case when $u \in V(C)$, it holds that $u = \rhoG{u}$.
See Fig.~\ref{fig:labeled-graph}(b) for an example of
a monocyclic graph, where~$\rhoG{u}$ is denoted for a vertex~$u$.

For a monocyclic graph $G$ and two vertices
$u, v \in V(G)$ such that $\rhoG{u} = \rhoG{v}$ 
(i.e., $u$ and $v$ are contained in the same pendent tree of $G$),
let $P(u, v)$ denote the unique path in $G$ between $u$ and~$v$.

\subsubsection{Mono-block 2-Augmented Trees}
\label{subsubsec:mono-block}
Let $H$ be a mono-block $2$-augmented tree, 
and $C_{1}$, $C_{2}$, and $C_{3}$ be the three distinct cycles of $H$.
We observe that there are exactly two vertices contained in all of the cycles of $H$.
We call these vertices {\em junctions}.
We call a pair of a junction $u$ and a neighbor of $u$ 
on some of the cycles of $H$ a {\em junction pair}.
Note that the number of junction pairs in one 
mono-block 2-augmented tree is either six or five
(see Fig.~\ref{fig:2-block}\,(a) and~(b), respectively) .
Let $u$ and $u'$ be the junctions of $H$.
For~$i \in [1,3]$, we define $P(u,u'; C_{i})$ to be the 
unique path from $u$ to $u'$ that is not in~$C_{i}$.
Note that either $P(u,u'; C_{i})$ is itself a junction pair, 
or it contains exactly two junction pairs.
See Fig.~\ref{fig:2-block}  
for an example of mono-block $2$-augmented trees and junctions. 

\begin{figure}[hbt]
\centering
\includegraphics[width=0.80\linewidth]{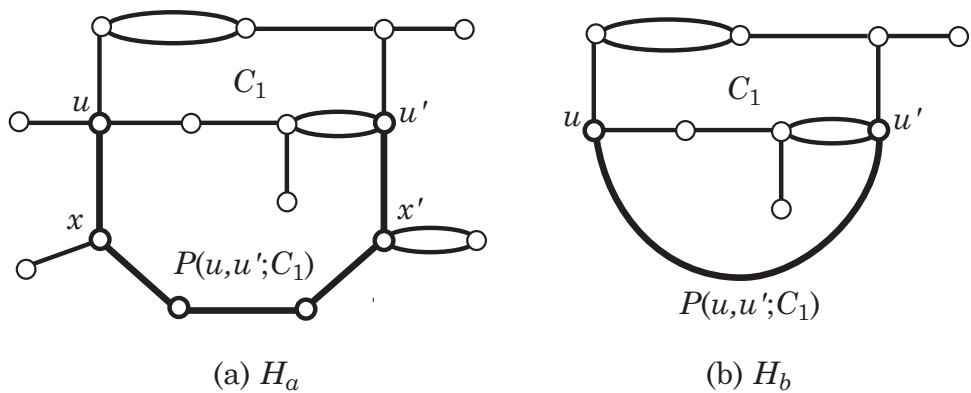}
\caption{
Mono-block $2$-augmented trees $H_{a}$ and $H_{b}$ with junctions  $u$ and $u'$.
(a)~The pair $\{u,x\}$ is one of the junction pairs of $H_{a}$.
The path $P(u,u';C_{1})$ is drawn in bold.
(b)~The path $P(u, u'; C_{1})$ drawn in bold consists of a single edge $uu'$.
}
\label{fig:2-block}
\end{figure}

\section{Problem Formulation}
\label{sec:formulate}

In this section, we formalize the problem to be addressed in this paper.

\subsection{Chemical Graphs}
To treat chemical compounds as multigraphs, 
we introduce a {\em color} for every vertex in a graph.
Colors represent chemical elements, such as {\tt O}, {\tt N}, or {\tt C}.
Let us denote the set of colors by $\Sigma$, 
and the color of a vertex $v$ by ${\rm col}(v)$.
The valence of a chemical element $c \in \Sigma$ is denoted 
by an integer function ${\rm val}(c) \in \mathbb{Z}_{+}$.
The size of a bond between two adjacent atoms is
indicated by the edge multiplicity between the two vertices 
that correspond to those atoms.
A multigraph $G$ is called a {\em $\Sigma$-colored} graph 
if each vertex $v \in V(G)$ is assigned a color ${\rm col}(v) \in \Sigma$.
A chemical compound can be viewed as a $\Sigma$-colored multigraph 
without self-loops, where vertices and colors represent 
atoms and elements, respectively.
Throughout this paper, we call $\Sigma$-colored multigraphs 
without self-loops {\em chemical graphs}.

Let $G$ be a chemical graph.
Chemical compounds, especially organic compounds, are rich with hydrogen atoms.
Since the valence of hydrogen is $1$, we can determine
the structure of chemical graphs without considering hydrogen atoms.
On this ground, we suppress hydrogen in order to 
enumerate chemical graphs more quickly.
As a result, for some vertex $v$ in a hydrogen-suppressed
chemical graph, the degree~${\rm deg}(v)$ 
may be smaller than the valence ${\rm val}({\rm col}(v))$.
For a vertex $v \in V(G)$, we define the {\em residual degree} 
of $v$ to be ${\rm val}({\rm col(v)}) -{\rm deg}(v)$ 
and denote it by $\mathrm{res}(v)$.
In a hydrogen-suppressed chemical graph, 
for a vertex $v$, $\mathrm{res}(v)$ represents the number 
of hydrogen atoms adjacent to $v$ in the corresponding chemical compound.

\subsection{Isomorphism on Chemical Multigraphs}
In enumerating chemical graphs, we must avoid duplication of equivalent graphs.
For example, two chemical graphs $G$ and $G'$ may have 
the same graph structure, and imply the same chemical 
compound even if they are different as labeled graphs.
This case is formalized by the notion of \emph{isomorphism} as follows.
Let $G = (V,E)$ and $G' = (V',bE')$ be two chemical multigraphs.
The following bijection $\psi$ from $V(G)$ to $V(G')$
is called an {\em isomorphism} from $G$ to $G'$:
\begin{itemize}
\item[(i)] for each vertex $x \in V(G)$, it holds that ${\rm col}(x) = {\rm col}(\psi(x))$; and
\item[(ii)] for each pair $\{x,y\} \subseteq V(G)$,
    it holds that ${\rm mul}_{G}(x, y) = {\rm mul}_{G'}(\psi(x),\psi(y))$. 
\end{itemize}

If there exists an isomorphism from $G$ to $G'$, 
then we say that $G$ and $G'$ are {\em isomorphic}.
We write $G \approx G'$ if $G$ and $G'$ are isomorphic, 
and write $G \not\approx G'$ otherwise.
For two sets $\mathcal{G}$ and  $\mathcal{G}'$ of chemical graphs, 
we say that  $\mathcal{G}'$  
{\em represents} $\mathcal{G}$ if \\
~-~for each  chemical graph  $G\in \mathcal{G}$, 
there is a chemical graph $G'\in \mathcal{G}'$ such that 
$G\approx G'$; and \\
~-~for any two   chemical graphs $G'_1, G'_2\in \mathcal{G}'$,
it holds that $G'_1\not\approx G'_2$. 

An {\em automorphism} of a chemical graph $G$ is defined to be 
an isomorphism $\psi$ from $V(G)$ to $V(G)$ itself. 

In addition, for two graphs $G$ and $G'$ and vertex subsets $X \subseteq V(G)$ and 
$Y \subseteq V(G')$ such that $|X| = |Y|$, 
if there exists an isomorphism $\psi$ from $V(G)$ to $V(G')$
such that for each vertex $x \in X$ it holds that $\psi(x) \in Y$,
we say that $G$ and $G'$ are \emph{$(X, Y)$-isomorphic}.
In the case when the sets $X$ and $Y$ are singletons,
i.e. $X = \{x\}$ and $Y = \{y\}$ and $G$ and $G'$ are
$(X, Y)$-isomorphic, we may write that they are
$(x, y)$-isomorphic.
Finally, if $G$ is a graph rooted at vertex $v_r$ and $G'$ is rooted at $v'_r$
and they are $(v_r, v'_r)$-isomorphic, then we say that they
are {\em rooted isomorphic} and denote this by $G \underset{r}{\approx} G'$.

\subsection{Feature Vectors}
\label{sec:feature_vectors}

In this paper, ``feature vectors'' represent occurrences of paths in a graph.
To specify a feature vector space, we fix three parameters: a set $\Sigma$ of colors;
the maximum multiplicity $d \geq 1$ among all pairs of vertices 
in multigraphs to be enumerated;
and the maximum length $K \geq 0$ of the path structures to be specified.

Let $c_{0},c_{1},\dots,c_{K} \in \Sigma$ be $K+1$ colors and 
$m_{1},m_{2},\dots,m_{K} \in [1,d]$ be $K$ integers, 
where it may hold that $c_{i} = c_{j}$ or $m_{i} = m_{j}$ for some $i$ and $j$.
We call an alternating sequence $t = (c_{0},m_{1},c_{1},\dots,m_{K},c_{K})$ 
a {\em colored sequence} of length $|t| = K$.
We denote the set of all colored sequences with length $K$ by $\Sigma^{K,d}$,
and the union of 
$\Sigma^{0, d},\Sigma^{1,d},\dots,\Sigma^{K, d}$ by $\Sigma^{\leq K, d}$.
For a colored sequence $t = (c_{0}, m_{1}, c_{1}, \dots, m_{K}, c_{K}) \in \Sigma^{K, d}$, 
let  ${\rm rev}(t)$ denote  its reverse sequence 
$(c_{K}, m_{K}, c_{K-1}, \dots, m_{1}, c_{0}) \in \Sigma^{K,d}$.

Let a chemical graph $P$ be 
a path $P = (v_{0}, v_{1}, \ldots, v_{K})$ of length $K$ with root $v_0$. 
We define the colored sequence $\gamma(P) \in \Sigma^{K, d}$ of $P$ to be
\begin{equation*}
\gamma(P) \triangleq ({\rm col}(v_{0}), {\rm mul}_{P}(v_{0},v_{1}), {\rm col}(v_{1}),
\ldots, {\rm mul}_{P}(v_{K-1},v_{K}), {\rm col}(v_{K})).
\end{equation*}

Let $G$ be a chemical graph.
For a colored sequence $t \in \Sigma^{\leq K,d}$, the {\em frequency} ${\rm frq}(t,G)$ 
of $t$ in $G$ is defined to be the number of vertex-rooted subgraphs $G'$ of $G$ 
such that 
$G' \underset{r}{\approx} P$ for a rooted path $P$ with   $\gamma(P) = t$.
We define the {\em feature vector} ${\bf f}(G)$ of {\em level} $K$ of $G$ 
to be the $|\Sigma^{\leq K,d}|$-dimensional vector
such that ${\bf f}(G)[t] = {\rm frq}(t,G)$
 for each colored sequence $t \in \Sigma^{\leq K,d}$.

Given a color set $\Sigma$ and integers $d$ and $K$, 
the set of $|\Sigma^{\leq K,d}|$-dimensional vectors whose entries are
nonnegative integers is  
called a  {\em feature vector space} and is
denoted by ${\bf f}(\Sigma, K, d)$.
Equivalently, each vector $\mathbf{g} \in \fvset(\Sigma, K, d)$ 
is a mapping $\mathbf{g} : \Sigma^{\leq K, d} \to \mathbb{Z}_{+}$.
For two vectors $\mathbf{g}, \mathbf{g}' \in \fvset(\Sigma, K, d)$,
we write $\mathbf{g} \leq \mathbf{g}'$ if for each entry $t \in \Sigma^{\leq K,d}$,
it holds that $\mathbf{g}[t] \leq \mathbf{g}'[t]$. 
For two given vectors ${\bf g}, {\bf g'} \in {\bf f}(\Sigma, K, d)$, 
a chemical graph $G$ is called {\em feasible}
if ${\bf g} \leq {\bf f}(G) \leq {\bf g'}$ and $\mathrm{res}(v) \geq 0$
holds for all vertices~$v \in V(G)$.
Let $\mathcal{G}({\bf g}, {\bf g'})$
denote the set of all chemical graphs feasible to  $({\bf g}, {\bf g'})$.

\section{Problem of Enumerating Mono-block 
$2$-Augmented Trees}\label{sec:EULF2aug}

Let $\mathcal{G}_1$ denote the set of $\Sigma$-colored labeled monocyclic graphs,
and $\mathcal{G}_2$ denote the set of 
$\Sigma$-colored labeled mono-block $2$-augmented trees.
Note that each of $\mathcal{G}_1$ and $\mathcal{G}_2$ conceptually contains 
infinitely many labeled graphs unless a way of expressing labeled graphs is restricted.
 
For two given vectors ${\bf g}_{\ell}$ and ${\bf g}_{u}$,
 the set of feasible graphs $H\in \mathcal{G}_{2}$ 
 is denoted by $\mathcal{G}_{2}({\bf g}_{\ell}, {\bf g}_{u})$.
Our goal in this paper is  to construct a set $\mathcal{G}'_{2}$ that represents 
the set $\mathcal{G}_{2}({\bf g}_{\ell},{\bf g}_{u})$. 
 Since we can use an existing algorithm to obtain a set 
 $\mathcal{G}'_1$ that represents  $\mathcal{G}_1$
 (for example, the algorithm due to Suzuki~{\it et al.}~\cite{Suzuki14}), 
we design an algorithm to our goal 
so that  our target set $\mathcal{G}'_{2}$ is constructed \\
(I)~by  adding  some number $p\in [1, \min\{\mathrm{res}(x),\mathrm{res}(y) \}]$ 
of edges 
between some non-adjacent vertices $x$ and $y$ 
in each $\Sigma$-colored labeled  monocyclic graph $H$ in $\mathcal{G}'_1$; and\\
(II)~by discarding from the graphs $H$ those that are infeasible,
i.e.,   ${\bf g}_{\ell} > {\bf f}(H)$ or ${\bf f}(H) > {\bf g}_{u}$.
\medskip

We design procedures for tasks (I) and (II) separately.
In the rest of this subsection, we examine the task (I). 
Testing if a graph $H\in \mathcal{G}_2$ 
satisfies ${\bf g}_{\ell} \leq  {\bf f}(H) \leq {\bf g}_{u}$
can be handled in a relatively easy way.
We remark that all non-isomorphic $\Sigma$-colored labeled monocyclic graphs
might not be needed for us to  generate our 
target set $\mathcal{G}'_{2}$ which is restricted 
by bounds ${\bf g}_{\ell}$ and ${\bf g}_{u}$.
In other words, we may be able to introduce
a pair $({\bf g}'_{\ell}, {\bf g}'_{u})$ of lower and upper feature vectors 
on $\Sigma$-colored labeled monocyclic graphs 
so that a smaller set of non-isomorphic  $\Sigma$-colored labeled monocyclic graphs
can produce all necessary graphs in our target set $\mathcal{G}'_{2}$.
 
For two vectors ${\bf g}_1$ and ${\bf g}_2$, the set of feasible
  graphs in $\mathcal{G}_{1}$ is denoted by $\mathcal{G}_{1}({\bf g}_1,{\bf g}_2)$.
Following the idea of Suzuki~{\it et al.}~\cite{Suzuki14}, 
we modify a given lower vector ${\bf g}_{\ell}$ into a vector 
${\bf g}^{\dagger}_{\ell}$ so that any graph in
$\mathcal{G}_{2}({\bf g}_{\ell},{\bf g}_{u})$ can be constructed 
from some graph in $\mathcal{G}_{1}({\bf g}^{\dagger}_{\ell}, {\bf g}_{u})$.
For a given lower vector ${\bf g}_{\ell} \in {\bf f}(\Sigma,K,d)$, 
we define~${\bf g}^{\dagger}_{\ell} \in {\bf f}(\Sigma, K, d)$ as follows:
\begin{itemize}
\item[(i)] for each $t \in \Sigma^{0,d}$, let ${\bf g}^{\dagger}_{\ell}[t] = {\bf g}_{\ell}[t]$;
\item[(ii)] for each $t = (c,m,c') \in \Sigma^{1,d}$, let
\begin{equation*}
\mathbf{g}_{\ell}^{\dagger}[t] = 
\left\{
\begin{array}{ll}
\max\{0 , \mathbf{g}_{\ell}[t] - 1\}, & \text{if}~c \neq c'\\
\max\{0 , \mathbf{g}_{\ell}[t] - 2\}, & \text{if}~c = c', \text{ and}
\end{array}
\right.
\end{equation*}
\item[(iii)] for each $t \in \Sigma^{\leq K,d} \setminus \Sigma^{\leq 1,d}$, 
let ${\bf g}^{\dagger}_{\ell}[t] = 0$.
\end{itemize}

\begin{lemma}~\label{lem:gpl}
For a set $\Sigma$ of colors and integers $K \geq 0$ and $d \geq 1$, 
let the vectors~${\bf g}_{u},{\bf g}_{\ell} \in {\bf f}(\Sigma,K,d)$ be such that 
${\bf g}_{\ell} \leq {\bf g}_{u}$, and $H \in \mathcal{G}_{2}({\bf g}_{\ell},{\bf g}_{u})$.
Then, for each junction pair $\{x,y\} \subseteq V(H)$, 
it holds that $H - xy \in \mathcal{G}_{1}({\bf g}^{\dagger}_{\ell},{\bf g}_{u})$.
\end{lemma}

\begin{proof}
Let $\{x,y\}$ be a junction pair in $H$. 
Let $G$ be the graph $H - xy$.
We show that $G \in \mathcal{G}_{1}({\bf g}^{\dagger}_{\ell},{\bf g}_{u})$
 i.e., $G$ is a graph in $\mathcal{G}_{1}$ with
 ${\bf g}^{\dagger}_{\ell} \leq {\bf f}(G) \leq {\bf g}_{u}$
and $\mathrm{res}(v) \geq 0$ for each vertex $v \in V(G)$.
Obviously, $G = H - xy$ is
 a graph in $\mathcal{G}_{1}$ because removing from $H$ all edges 
 between a junction pair leaves only one cycle in $G$.

Since $H \in \mathcal{G}_{2}({\bf g}_{\ell},{\bf g}_{u})$, 
we have ${\bf g}_{\ell} \leq {\bf f}(H) \leq {\bf g}_{u}$.
First we check the upper bound of ${\bf f}(G)$.
Since removing all edges between $x$ and $y$ does not increase
 the frequency of any colored sequence, 
 it holds that ${\rm frq}(t,G) \leq {\rm frq}(t,H)$ for any colored sequence~$t$.
Hence it holds that ${\bf f}(G) \leq {\bf f}(H)\leq {\bf g}_{u}$.
Next we check the lower bound of~${\bf f}(G)$.
For a colored sequence $t \in \Sigma^{0,d}$, we consider the entry~${\bf f}(G)[t]$.
We have~${\bf f}(G)[t] = {\bf f}(H)[t] = {\bf g}_{\ell}[t] = {\bf g}^{\dagger}_{\ell}[t]$
 for all colored sequences $t \in \Sigma^{0,d}$, since removing edges does not 
 change the frequency of colored sequences of length~$0$.
A colored sequence~$t \in \Sigma^{1,d}$ is given by $(c,m,c')$, and
 we have the following observations: If~$t = \gamma(xy)$ or~$t = \gamma(yx)$ 
 for some path~$xy$ in~$H$, it holds that
\begin{equation*}
{\bf f}(G)[t] = 
\left\{
\begin{array}{ll}
\max\{0 , {\bf f}(H)[t] - 1\} & \text{if}~t \neq {\rm rev}(t) \\
\max\{0 , {\bf f}(H)[t] - 2\} & \text{if}~t = {\rm rev}(t);
\end{array}
\right.
\end{equation*}
and otherwise, we have ${\bf f}(G)[t] = {\bf f}(H)[t]$.
For all colored sequences $t \in \Sigma^{1,d}$, 
since ${\bf g}_{\ell} \leq {\bf f}(H)[t]$, it holds that
\begin{equation*}
{\bf f}(G)[t] \geq
\left\{
\begin{array}{ll}
\max\{0 , {\bf g}_{\ell}[t] - 1\} & \text{if}~c \neq c' \\
\max\{0 , {\bf g}_{\ell}[t] - 2\} & \text{if}~c = c'. 
\end{array}
\right.
\end{equation*}
According to the definition of ${\bf g}^{\dagger}_{\ell}$, 
we have ${\bf f}(G)[t] \geq {\bf g}^{\dagger}_{\ell}[t]$
for each colored sequence~$t \in \Sigma^{1,d}$.
For a colored sequence $t \in \Sigma^{\leq K, d}  \setminus \Sigma^{\leq 1,d}$, 
it is clear from the definition of~${\bf g}^{\dagger}_{\ell}$ 
that ${\bf f}(G)[t] \geq {\bf g}^{\dagger}_{\ell}[t] = 0$.
From above, we have ${\bf g}^{\dagger}_{\ell} \leq {\bf f}(G) \leq {\bf g}_{u}$.

Finally, we prove that $\mathrm{res}(v) \geq 0$ for each vertex $v \in V(G)$.
Since $H$ is feasible, it holds that $\mathrm{res}(u) \geq 0$ for each vertex $u \in V(H)$.
Removing edges from $H$ does not decrease the residual degree of any vertex.
Hence it holds that $\mathrm{res}(v) \geq 0$ for each vertex $v \in V(G)$.
\end{proof}

Lemma~\ref{lem:gpl} states that for each graph 
$H \in \mathcal{G}_{2}({\bf g}_{\ell}, {\bf g}_{u})$,
 there is at least one monocyclic graph 
 $G \in \mathcal{G}_{1}({\bf g}^{\dagger}_{\ell},{\bf g}_{u})$, 
 a pair $\{x,y\} \in \overline{E}(G)$, and an integer 
 $p$ giving $H$ as $G + p \cdot xy$.
For this, 
we first generate all monocyclic graphs
in~$\mathcal{G}_{1}({\bf g}^{\dagger}_{\ell},{\bf g}_{u})$.
Then we enumerate all mono-block $2$-augmented trees
in $\mathcal{G}_{2}({\bf g}_{\ell}, {\bf g}_{u})$
by adding edges between pairs of non-adjacent vertices 
in each of the given monocyclic graphs, 
where the pair of endvertices of the newly added edges becomes 
a junction pair of the newly 
created mono-block $2$-augmented tree.
The problem to deal with in this paper is formalized as follows.\\

\begin{samepage}
 \noindent{\bf Enumerating mono-block $2$-augmented trees 
from given monocyclic graphs
with given Upper and Lower Path Frequency}
\begin{algorithmic}
\Require
 A set $\Sigma$ of colors, integers $K \geq 0$ and $d \geq 1$, two vectors 
${\bf g}_{u},{\bf g}_{\ell} \in {\bf f}(\Sigma, K, d)$
such that ${\bf g}_{\ell} \leq {\bf g}_{u}$ 
and for all colored sequences $t \in \Sigma^{0,d}$
it holds that ${\bf g}_{\ell}[t] = {\bf g}_{u}[t]$, 
and a set $\mathcal{G}'_{1}$ that represents the set 
$\mathcal{G}_{1}({\bf g}^{\dagger}_{\ell}, {\bf g}_{u})$.
\Ensure A   set $\mathcal{G}'_{2}$  that represents 
the set $\mathcal{G}_{2}({\bf g}_{\ell}, {\bf g}_{u})$. 
\end{algorithmic}
\bigskip
\end{samepage}

In the rest of the paper, we focus on designing an enumerating procedure for the task~(I).

\subsection{Sketch of the Enumerating Procedure}
\label{subsec:IdeasEnu}
Suppose that 
a set $\mathcal{G}'_{1}=\{G_1, G_2,\ldots, G_q\}$ that represents the set 
$\mathcal{G}_{1}({\bf g}^{\dagger}_{\ell}, {\bf g}_{u})$ is given. 
When $\Sigma$-colored labeled mono-block $2$-augmented trees are generated
by adding a multiple edge $xy$ 
between a pair $\{x, y\} \in \overline{E}(G)$
to each graph $G_i\in  \mathcal{G}'_{1}$, we have to check 
whether the same graph has already been generated in the process or not, i.e., 
whether a newly generated graph
$H = G_i + p\cdot xy \in \mathcal{G}_{2}$
is isomorphic to another graph
$H' = G_j + q\cdot x'y' \in \mathcal{G}_{2}$
that has already been obtained from some graph $G_j \in  \mathcal{G}'_{1}$.
We call the duplication arising when 
$H = G_i + p\cdot xy$ and $H' = G_j + q\cdot x'y'$ are isomorphic
with $i \neq j $ \emph{inter-duplication}, 
and that when $H = G_i + p\cdot xy \approx H' = G_i + q\cdot x'y'$
\emph{intra-duplication}.
This section provides some ideas on how to avoid such duplications 
without storing all generated graphs and 
explicitly comparing the new one with each of them.

In enumeration algorithms, 
the concept of a {\em family tree}
 is widely employed~\cite{NU03_rooted,NU05_colored,Suzuki14} 
 in order to efficiently cope with inter-duplication.
In order to define a family tree for graphs, 
we need to define a {\em parent-child relationship} 
between graph structures so that the parent structure of a given 
chemical graph $H$ is uniquely determined 
 from the topological structure of $H$.
A {\em parent-child relationship} over   
classes $\mathcal{H}$ and  $\mathcal{G}$ of chemical graphs 
is defined by an injective mapping  $\pi: \mathcal{H} \to \mathcal{G}$ as follows.
Given a chemical graph $H\in \mathcal{H}$,
we define a labeled graph $G=\pi(H)\in \mathcal{G}$ so that 
for any two chemical graphs $H_1$ and $H_2$ with $H_1 \approx H_2$,
it holds 
$G_1 \approx G_2$ 
with chemical graphs $G_i=\pi(H_i)$, $i=1,2$,
implying that $\pi(H)$ is determined only from the topological structure of $H$.
For two chemical graphs $G\in \mathcal{G}$ and $H\in \mathcal{H}$ 
with $G = \pi(H)$, the graph $G$ is called the {\em parent} of $H$
(the parent of $H$ as an unlabeled graph which is unique up to automorphism), 
and the graph $H$ is called a {\em child} of $G$.
Let  $\Ch_\pi(G)$ denote the set of chemical graphs $H\in \mathcal{H}$
 such that $G = \pi(H)$.
Therefore, we easily observe the next.

\begin{lemma}
Let  $\pi: \mathcal{H} \to \mathcal{G}$ be a parent-child relationship over
classes $\mathcal{H}$ and  $\mathcal{G}$ of chemical graphs.  
For two chemical  graphs $G,G'\in \mathcal{G}$, if $G \not\approx G'$, then
 $\Ch_\pi(G)\cap \Ch_\pi(G')=\emptyset$.
 \label{lem:no_duplication_from_different_parents}
\end{lemma}

Then, Lemma~\ref{lem:no_duplication_from_different_parents} ensures that 
we do not need to explicitly check 
for inter-duplication, that is, 
if a  chemical  graph $H\in \mathcal{G}_{2}$ that is a child of
a chemical  graph $G_i \in \mathcal{G}_{1}$ is isomorphic 
to some  chemical  graph $H'\in \mathcal{G}_{2}$ 
generated from
a chemical  graph $G_j\in \mathcal{G}_{1}$ with~$i\neq j$.

This section closes by presenting a high level description 
of a procedure that, given a chemical  graph
$G\in \mathcal{G}_1$ and a maximum multiplicity $d$, 
constructs 
a set $\mathcal{G}_2(G,d)$ that represents 
the set of chemical graphs in $\Ch_\pi(G)~(\subseteq \mathcal{G}_2)$ 
with multiplicity at most $d$.
For this, we will introduce \\
(I-a)~an appropriate choice of  a parent-child relationship
$\pi: \mathcal{G}_{2} \to \mathcal{G}_{1}$;    and \\
(I-b)~a way of avoiding intra-duplication from the same 
chemical graph $G\in \mathcal{G}_{1}$.\\
The technical details on (I-a) and (I-b) will be discussed in 
Sections~\ref{sec:Parent-Child}
and \ref{sec:avoid_duplicate}, respectively.
Now the task (I)~is divided into two subtasks,  (I-a) and~(I-b).

We here explain ideas on (I-a) and (I-b).
For (I-a), we define a parent-child relationship 
$\pi: \mathcal{G}_{2} \to \mathcal{G}_{1}$  
such that 
the parent $\pi(H)$ of $H$ is a $\Sigma$-colored labeled monocyclic graph 
obtained from $H$ by removing all 
$\mathrm{mul}_H\{x, y\}$ edges between a junction pair $\{x,y\}$
which meets a certain condition derived in Section~\ref{sec:Parent-Child}.
The reason why such a junction pair  $\{x, y\}$ is carefully chosen is that
 the resulting function $\pi$ must satisfy the definition of parent-child relationships
over classes $\mathcal{G}_{2}$ and $\mathcal{G}_{1}$.

To handle~(I-b) for a chemical graph $G\in \mathcal{G}_{1}$, 
 we call a subset $F\subseteq \overline{E}(G)$ {\em proper} if\\ 
~-~for any  $H \in \Ch_\pi(G)$, there is a pair $\{u,v\} \in F$
 and an integer $p \in [1, \min \{ \res(u), \res(v) \}]$,
 such that $H  \approx G + p \cdot uv$; and\\
~-~for two distinct pairs $\{x,y\}, \{u,v\} \in F$, 
and integers  $p \in [1, \min \{ \res(x), \res(y)\}]$
and $q \in [1, \min \{ \res(u), \res(v)\} ]$
it holds that $G + p \cdot xy \not\approx G + q \cdot uv$.\\
Section~\ref{sec:avoid_duplicate} provides a procedure for testing
if a given pair $\{x,y\}\in \overline{E}(G)$ belongs to a proper set $F$
after executing a preprocessing to construct~$F$. 

With the above ideas, a procedure for constructing 
a set $\mathcal{G}_2(G, d)$ is given as follows. 

\pagebreak

\medskip
\noindent{\bf Procedure~1}
\begin{algorithmic}[1]
\Require An integer~$d$ 
and a chemical graph $G\in \mathcal{G}_1$ with a unique cycle $C$
and multiplicity at most~$d$.
\Ensure A set $\mathcal{G}_2(G, d)$ that represents 
the set of chemical graphs in $\Ch_\pi(G)~(\subseteq \mathcal{G}_2)$
with multiplicity at most $d$. 
\State Execute a preprocessing to construct a proper set $F \subseteq \overline{E}(G)$;
  /* by Procedure~{5}~in Section~\ref{sec:avoid_duplicate} */
\For{{\bf each} pair $\{x, y\} \in F$ }
    \For{{\bf each} integer $p \in [1, \min \{ d, \res(x), \res(y) \} ]$}
      \If {$G + p \cdot xy$ is a child of $G$ 
	    /* tested by Procedure~{4} in Section~\ref{sec:Parent-Child} */}
	\State Generate the mono-block $2$-augmented tree 
	  $G + p \cdot xy$ as part of the output
      \EndIf
    \EndFor
\EndFor.
\end{algorithmic}
\vspace{1ex}

%

\section{Signature and Code}
\label{sec:signature}

In order to define our parent-child relationship 
$\pi: \mathcal{G}_2\to  \mathcal{G}_1$ 
for the task~(I-a), 
we introduce the notion of ``signature'' of graphs. 

For a class $\mathcal{G}$ of graphs, if we have a way 
of choosing a labeling of each graph 
$G \in \mathcal{G}$ which is unique up to the graph's automorphism, then we can test 
the isomorphism of the two graphs directly by comparing their labels.
Such a labeling for a graph~$G$ is called a {\em canonical form} of~$G$.
Once such a canonical form for a class $\mathcal{G}$ of graphs is obtained, 
we can easily 
encode each graph $G \in \mathcal{G}$ into a sequence $\sigma(G)$, 
called the {\em signature} of $G$, 
such that two graphs $G,G' \in \mathcal{G}$ 
are isomorphic if and only if~$\sigma(G) = \sigma(G')$.

\subsection{Lexicographical Order}
We fix a total order of the colors in $\Sigma$ arbitrarily, 
e.g., ${\tt O} < {\tt N} < {\tt C}$.
We introduce a {\em lexicographical order} among sequences with elements in
 $\Sigma \cup \mathbb{Z}_{+}$ as follows.
A sequence $A = (a_{1},a_{2},\dots,a_{p})$ is {\em lexicographically smaller} 
than a sequence $B=(b_{1},b_{2},\dots,b_{q})$ 
if there is an index $k \in [1, \min\{p,q\}]$ such that
\begin{itemize}
\item[(i)] $a_{i} = b_{i}$ for all $i \in [1,k]$; and
\item[(ii)] $k = p < q$ or $a_{k+1} < b_{k+1}$ with $k < \min\{p,q\}$.
\end{itemize}
If $A$ is lexicographically smaller than $B$, then we denote $A \prec B$.
If $p = q$ and $a_{i} = b_{i}$ for all $i \in [1,p]$, then we denote $A = B$.
Let $A \preceq B$ mean that $A \prec B$ or $A = B$.

We often rely on lexicographically sorting a collection $S = (s_1, s_2, \ldots, s_k)$ 
of $k$ sequences.
We will represent a lexicographically ascending (resp., descending) 
order on collection $S$ by a permutation
$\pi : [1, k] \to [1, k]$ such that for $1 \leq i < j \leq k$
it holds that $s_{\pi(i)} \preceq s_{\pi(j)}$ 
(resp., $s_{\pi(j)} \preceq s_{\pi(i)}$).

For a collection $S$ of sequences, let us denote by $||S|| = \sum_{s \in S} |s|$
the total length of the sequences
in the collection~$S$.
A known algorithm due to Aho {\em et al.}~\cite{AHU74}
can be used to lexicographically sort a collection $S$ of sequences
over an alphabet of size $n$
in $O(||S|| + n)$ computation time.

\subsection{Canonical Form and Signature of Trees}
\label{subsec:SignatureOfTrees}
In this subsection, we review the concept of a canonical form 
of rooted trees~\cite{NU03_rooted,NU05_colored}.

\subsubsection{Ordered Trees}
\label{sec:ordered_trees}

An {\em ordered tree} is a rooted tree with a fixed total 
order among the children of each vertex.
By convention, we assume that the order of children 
in an ordered tree is from left to right.

Let $T$ be a multitree with $n$ vertices, rooted at a vertex $r$.
There may be many different ordered trees on $T$.
A canonical form of $T$ is given by an adequately chosen ordered tree on~$T$.
When we conduct a depth-first-search, 
we assume that we visit children from left to right.
We denote the vertices of $T$ by $v_{1}, v_{2}, \dots, v_{n}$, 
indexed in the order visited
 by a depth-first-search starting from the root.
Let $\tau$ be an ordered tree on $T$ and let~$\delta(\tau)$ 
denote the alternating \emph{color-depth sequence} 
$(c_{1}, d_{1}, \dots, c_{n} ,d_{n})$ that consists of the color 
$c_{i}$ and the depth $d_{i}$ 
of the $i$-th vertex $v_{i}$ for $i \in [1, n]$.
Let ${\rm M}(\tau)$ denote the sequence 
$(m_{2}, m_{3}, \dots, m_{m})$ of the multiplicity 
$m_{i} = {\rm mul}(v_{i},{\rm p}(v_{i}))$ between the $i$-th vertex $v_i$ and 
its parent ${\rm p}(v_{i})$ for $i \in [2,n]$.

%

\subsubsection{Left-heavy Trees}
\label{sec:left-heavy}

A {\em left-heavy tree} of a rooted tree $T$ is an ordered tree $\tau$ that 
has the lexicographically maximum code $\delta(\tau)$
among all ordered trees of $T$.
Note that a left-heavy tree has the following recursive structure: 
for every vertex $v \in V(T)$, 
the subtree $T_{v}$ is also a left-heavy tree, 
and $\delta(T_{v})$ and ${\rm M}(T_{v})$ are 
continuous subsequences of $\delta(T)$ and ${\rm M}(T)$, respectively.
We define the canonical form of a rooted tree $T$ 
to be the left-heavy tree $\tau$ that has 
the lexicographically maximum sequence ${\rm M}(\tau)$ 
among all left-heavy trees of $T$, and 
define the \emph{signature} of $T$ 
to be $\sigma(T) \triangleq (\delta(\tau), {\rm M}(\tau))$.
We give a procedure to calculate the signature $\sigma$ of
a rooted tree as Procedure~{2} in Section~\ref{sec:calculating_sigma}.


\subsubsection{Calculating the Signature of Rooted Multi-Trees}
\label{sec:calculating_sigma}
For any two sequences $S_1$ and $S_2$, 
let $S_1 \oplus S_2$ denote the concatenation of~$S_1$ and~$S_2$.
Given an ordered multi-tree $T$ 
on $n$ vertices indexed $v_1, v_2, \ldots, v_n$ as visited in a depth-first
traversal,
let $\delta(T) = (c_1, d_1, c_2, d_2, \ldots, c_n, d_n)$
be its color-depth sequence as defined in Section~\ref{sec:left-heavy}.
For an integer $k \geq 1$ we define the \emph{$k$-shift} 
$\delta^k(T)$ of the sequence $\delta(T)$ to be the sequence 
$(c_1, d_1 + k, c_2, d_2 + k, \ldots, c_n, d_n + k)$
obtained by adding $k$ 
to each of the depth entries of~$\delta(T)$.

Let $T$ be an ordered multi-tree rooted at a vertex $r$,
and let $u_1, u_2, \ldots, u_{\deg(r)}$ denote
the children of $r$ indexed according to their left-to-right ordering.
Let $\mm(u_i)$ denote $\mul(u_i, \parent(u_i)) \oplus \mm(T(u_i))$ for each $i \in [1,\deg(r)]$.
Given the signatures $\sigma(T(u_i)) = (\delta(T(u_i)), \mm(T(u_i)))$ for all $i \in [1,\deg(r)]$, 
we devise a way to represent $\sigma(T)$ by $\sigma(T(u_i))$ via the following observation.

\begin{observation}
\label{obs:representation}
Let $T$ be an ordered multi-tree rooted at a vertex $r$, and let 
$\Ch(r) = \{u_1, u_2, \ldots, u_{\deg(r)} \}$ 
denote the set of children of $r$ indexed  according to their left-to-right ordering.
Given the sequences $\delta(T_{u})$ and $\mm(T_{u})$ 
for all $u \in \Ch(r)$,
for the sequences $\delta(T)$ and $\mm(T)$ it holds that:
\begin{align*}
\delta(T) &= (\col(r), 0) \oplus \delta^1(T(u_1)) \oplus  
\delta^1(T(u_2)) \oplus  \cdots \oplus  \delta^1(T(u_{\deg(r)})) ; \\
\mm(T) &= \mm(u_1) \oplus \mm(u_2) \oplus  \cdots \oplus \mm(u_{\deg(r)}). 
\end{align*}
\end{observation}

\begin{lemma}
\label{lem:canonical}
Given a rooted multi-tree $T$, let $\tau$ denote 
the ordered multi-tree such that $\sigma(T) = (\delta(\tau), \mm(\tau))$,
and let $\prec_{\tau}$ denote the left-to right ordering among siblings in~$\tau$.
For any two siblings $u$ and $v$ in $\tau$,
if $u \prec_{\tau} v$, then it holds that 
$(\delta(\tau_v),  \mul(v, \parent(v)),  \mm(\tau_v)) \prec 
(\delta(\tau_u), \mul(u, \parent(u)), \mm(\tau_u))$.
\end{lemma}
\begin{proof}
To derive a contradiction, suppose that there exist siblings 
$u$ and $v$ in $\tau$ such that $u \prec_{\tau} v$ and 
$(\delta(\tau_u), \mul(u, \parent(u)), \mm(\tau_u)) \prec 
(\delta(\tau_v), \mul(v, \parent(v)), \mm(\tau_v))$ holds.
Let $\tau'$ denote the ordered multi-tree obtained by switching the places of $u$ and $v$ in $\tau$.
Clearly, $\tau'$ is isomorphic to $T$ and $(\delta(\tau), \mm(\tau)) \prec (\delta(\tau'), \mm(\tau'))$.
This contradicts the assumption that $ \sigma(T) = (\delta(\tau), \mm(\tau))$, 
i.e., that $\sigma(\tau)$ is lexicographically maximum among 
all of the ordered multi-trees isomorphic to~$T$.
\end{proof}

By Observation~\ref{obs:representation} and Lemma~\ref{lem:canonical}, 
we show an algorithm to calculate the signature of a given rooted 
multi-tree in Procedure~{2}.
As an added benefit, the procedure in fact calculates the signatures of all rooted subtrees
of a given tree.

\bigskip
{\bf Procedure~2} {\sc SubTreeSignature}
\label{alg:signature_rmt}
\begin{algorithmic}[1]
\Require A $\Sigma$-colored multi-tree $T$
		    with multiplicity at most $d$ rooted at a vertex $r \in V(T)$.
\Ensure The signatures $\sigma(T_v)$ of each rooted tree $T_v$, $v \in V(T)$.
\State {$s := \emptyset$};
\For {{\bf each} $v \in V(T)$ in DFS-post order} \label{line:for1-begin}
	\If{ $v$ is a leaf}
		\State {$\delta[v] := (\col(v), 0); \mm[v] := \emptyset$}
	\Else
		\Statex {~~~/* The signatures $s[u]$ of all children of $v$ are already obtained */}
		\For{{\bf each} $u \in \ch(v)$} \label{line:for2-begin}
			\State {$\delta' :=$ 1-shift of $\delta[u]$}; \label{line:get-delta}
			\State {$\mm':= \mul(u, \parent(u)) \oplus \mm[u]$}; \label{line:get-mul}
			\State {$s'[u] := (\delta', \mm')$}
		\EndFor; \label{line:for2-end}
		\State {$S := (s'[u] \mid u \in \ch(v) )$};
		\State {Let $k :=  |\ch(v)|$};
		\Statex {~~~/* Represent $S$ as $( s_i = (\delta_i, \mm_i) \mid i \in [1, k] )$ */}
		\State {Sort $S$ in lexicographically descending order $\pi$}; \label{line:lex-sort}
		\State {$\delta[v] := (\col(v), 0) \oplus \delta_{\pi(1)} \oplus \delta_{\pi(2)} \oplus 
			      \cdots \oplus \delta_{\pi(k)}$};
		\State {$\mm[v] := \mm_{\pi(1)} \oplus \mm_{\pi(2)} \oplus \cdots \oplus \mm_{\pi(k)}$}
	\EndIf
	\State {$s[v] := (\delta[v], \mm[v])$}
\EndFor; \label{line:for1-end}
\State{{\bf output} $s[v]$ as $\sigma(T_v)$ {\bf for each} $v \in V(T)$}.
\end{algorithmic}
\bigskip

\begin{lemma}
 \label{lem:complexity_of_signature}
 Given a $\Sigma$-colored rooted tree $T$ on $n$ vertices
 and multiplicity at most~$d$, 
 Procedure {2} computes the signatures $\sigma(T_v)$
 of all rooted subtrees $T_v$, $v \in V(T)$, of $T$
 in $O(n \cdot(n + |\Sigma| + d))$ time.
\end{lemma}
\begin{proof}
Let $n_v$ denote the number of vertices in the subtree $T_v$ rooted at vertex~$v$,
and ${\rm d}_v$ the maximum depth of a leaf in the rooted tree $T_v$,
where~$v$ is taken to have depth~$0$, and it holds that ${\rm d}_v \leq n_v$.

The for-loop of lines~\ref{line:for1-begin} to~\ref{line:for1-end}
is executed for each vertex $v$ in $T$.
Since vertices are iterated in an DFS-post order,
in each iteration, the signatures $\sigma(T_u)$ are already computed for 
each child $u$ of $v$ in $T$.
Then, in the for-loop of lines~\ref{line:for2-begin} to~\ref{line:for2-end},
their signatures are gathered and the depth entries are offset by~1 in line~\ref{line:get-delta}.
This obviously takes at most $O(n_v)$ time.
Then, in line~\ref{line:lex-sort} the gathered sequences are sorted lexicographically.
The total length of the sequences is $O(n_v)$, and they are over the 
alphabets $\Sigma$ for the color of each vertex, $[1, {\rm d}_v]$ for the
depth, and $[1, d]$ for the multiplicity, thus the total alphabet size is
$|\Sigma| + d +  {\rm d}_v$.
By the algorithm for lexicographical sorting due to Aho {\em et al.}~\cite{AHU74},
the lexicographical sorting in line~\ref{line:lex-sort} takes~$O(n_v + |\Sigma| + d +  {\rm d}_v)$ time.
Finally, summing over all vertices $v$ in $T$,
for the computational complexity we get
\begin{align*}
 \sum_{v \in V(T)} O(n_v + |\Sigma| + d +  {\rm d}_v) 
 &= O(\sum_{v \in V(T)} (n_v + {\rm d}_v) + n \cdot (|\Sigma| + d) ) \\
 &= O(n^2 + n \cdot(|\Sigma| + d)),
\end{align*}
as required. 
\end{proof}


\subsubsection{Ranking of Rooted Trees}
\label{sec:ranking_of_trees}

Let $\mathcal{T}$ be a  finite set of rooted multi-trees,
and let $Z = \{\sigma(T) \mid T \in \mathcal{T} \}$
denote the set of signatures of the trees in $\mathcal{T}$.
We define a lexicographical order over $Z$
in the usual sense, i.e., 
for $\sigma_1 = (\delta_1, {\rm M}_1), \sigma_2=(\delta_2, {\rm M}_2) \in Z$
we write $\sigma_1 \prec \sigma_2$ if 
``$\delta_1 \prec \delta_2$'' or 
``$\delta_1 = \delta_2$ and ${\rm M}_1 \prec {\rm M}_2$.''
Then, we use the lexicographical order over the set $Z$ to
define a ranking $\rank_{\mathcal{T}} : \mathcal{T} \to [1, |Z|]$,
such that for two trees $T_1, T_2 \in \mathcal{T}$,
$\rank_{\mathcal{T}} (T_1) < \rank_{\mathcal{T}}(T_2)$ if $\sigma(T_1) \prec \sigma(T_2)$,
and $\rank_{\mathcal{T}}(T_1) = \rank_{\mathcal{T}}(T_2)$ means that $\sigma(T_1) = \sigma(T_2)$,
i.e. $T_1$ and $T_2$ are isomorphic.
It  follows that having a rank function over a set of multi-trees,
we can check whether two trees in the set 
are isomorphic to each other by comparing their ranks.

There exist algorithms in the literature
that can calculate the rank of each subtree of a given tree~\cite{DIR99}
and rooted subgraph of an outerplanar graph~\cite{IN12}
in time linear in the number of vertices in the graph.
In our implementation we use simpler algorithms for this purpose
at the cost of a higher time complexity.

For a set $\mathcal{T}$ of rooted trees, let $\mathcal{T}^*$
denote the set of all rooted subtrees of trees in $\mathcal{T}$.
We give a procedure to calculate a ranking of a given set $\mathcal{T}$  of rooted trees
in Procedure~{3}.
By  Procedure~{2},
we in fact obtain a ranking in the set $\mathcal{T}^*$ at no additional cost.

\bigskip
{\bf Procedure~3} {\sc TreeRanking}
\begin{algorithmic}[1]
\Require A set $\mathcal{T}$ 
		of $\Sigma$-colored rooted multi-trees with 
		multiplicity at most~$d$.
\Ensure A ranking function $\rank_{\mathcal{T}^*}$ of~$\mathcal{T}^*$.
  \State{$R:=\emptyset$; $h := |\mathcal{T}^*|$};
  \State{$S := (\sigma(T_i) \mid  i \in [1, h])$;  \label{line:get_sigma}
		/* Calculate $\sigma(T)$ by Procedure~{2} in Sec.~\ref{sec:calculating_sigma} */}
  \Statex{/* Treat $S = (s_1, s_2, \ldots, s_h)$ as an ordered set */}
  \State{Sort $S$ in lexicographically ascending order $\pi$;} \label{line:sort_all}
  \State{$R[T_{\pi(1)}] := 1$;} {$r := 1$};
  \For{ {\bf each} $i \in [2, h]$} \label{line:for3-begin}
	  \State{{\bf if} $s_{\pi(i-1)} \prec s_{\pi(i)}$ {\bf then} $r := r+1$ {\bf endif}}; \label{line:lexcompare}
	  \State{$R[T_{\pi(i)}] := r$;}	
  \EndFor; \label{line:for3-end}
  \State{{\bf output} $R$ as  $\rank_{\mathcal{T}^*}$}.
\end{algorithmic}
\bigskip


\begin{lemma}
 \label{lem:ranking_all_subtrees}
 Let $\mathcal{T}$ be a given set of $\Sigma$-colored
 rooted multi-trees with multiplicity at most~$d$,
and let $n$ denote the total number of vertices
over trees in $\mathcal{T}$.
Then, the rank of each rooted subtree of all trees in $\mathcal T$ can be computed in
$O(n (n + |\Sigma| + d ))$ time in total.
\end{lemma}
\begin{proof}
Let $\mathcal{T} = \{ T_1, T_2, \ldots, T_k \}$, and
let $n_i$, $i \in [1, k]$ denote the number of vertices in tree~$T_i$,
where $n = \sum_{i \in [1, k]} n_i$.
The signature $\sigma(T_i) = (\delta(T_i), \mm(T_i))$ of each
tree $T_i$ is a sequence with $O(n_i)$ entries
with $|\Sigma| + n_i + d$ possible values,
for the color of vertices and depth in a tree in $\delta(T_i)$,
and multiplicity with the parent in $\mm(T_i)$, respectively.

 By Lemma~\ref{lem:complexity_of_signature},
 computing the signatures of all rooted subtrees
 in line~\ref{line:get_sigma} takes 
 $O(n_i \cdot (n_i + |\Sigma| + d ))$ time for each tree $T_i$,
 and therefore $O(n \cdot ( n + |\Sigma| + d ))$ time in total.
 Now, each tree $T_i$ has $n_i$ rooted subtrees, 
 and the total number of vertices over these subtrees is $O(n_i^2)$.
 Therefore, the collection of signatures for the rooted subtrees
 of tree $T_i$ has in total $O(n_i^2)$ elements taking at
 most $|\Sigma| + n_i + d$ different values (alphabet size).
 Over all trees $T_i$, the elements of the subtree signatures
 take at most 
 $|\Sigma| + \max_{i \in [1, k]} \{n_i\} + d \leq |\Sigma| + n+ d$
 different values.
 Summing over all trees $T_i \in \mathcal{T}$,
 we get that the total length of the signatures over all subtrees is
 $\sum_{i \in [1, k]} O(n_i^2) = O(n^2)$.
Then, all these signatures can be lexicographically sorted in 
$O(n^2 + |\Sigma| + d)$ time~\cite{AHU74}, which is dominated by the time
to calculate the signatures.

Finally, having the lexicographically sorted signatures,
we assign rank to trees in a straightforward manner
by iterating over the sorted signatures as in lines~\ref{line:for3-begin} to~\ref{line:for3-end} in 
Procedure~{3},
and the claim follows. 
\end{proof}

\subsection{Code on Substructures of Mono-block $2$-Augmented Trees}

We can use the signature defined for rooted trees in
Section~\ref{subsec:SignatureOfTrees} to
devise codes for substructures of $k$-augmented trees, $k = 1, 2$.
Let $H$ be a mono-block $k$-augmented tree with a block~$B$,
and let $\mathcal{T}(H) = \{H \langle v \rangle \mid v \in B\}$
denote the set of all pendent trees in~$H$.
Recall that for a multi-tree $T \in \mathcal{T}(H)$,
$\rank_{\mathcal{T}(H)}(T)$ as defined in Section~\ref{sec:ranking_of_trees}
gives the rank of $T$ according to the lexicographical order
of the signature~$\sigma(T)$.
To simplify our notation, when the graph $H$ and hence 
the set $\mathcal{T}(H)$ is clear, we write $\rank()$
for the rank function $\rank_{\mathcal{T}(H)}()$.

Let $H$ be a mono-block 2-augmented tree.
Recall that there are exactly two vertices called junctions, 
which are contained in all of the cycles in~$H$,
as defined in Section~\ref{subsubsec:mono-block}.
Also recall that for a junction $u$ in $H$ and a neighbor $u'$ of $u$ on a cycle of $H$,
we call the pair $(u, u')$ a junction pair.
Recall that there are at most six junction pairs in a mono-block 2-augmented tree.
Let $u$ and $v$ denote the junctions of~$H$.
For a junction $u$, we define $\code(u)$ to be the sequence 
$(|V(H \langle u\rangle)|, \col(u), \deg(u), \rank(H\langle u\rangle))$.
For junctions $u$ and $v$ and a cycle $C_i$ of $H$, 
recall that $P(u, v; C_i)$
denotes the $uv$-path $(v_1, v_2, \ldots, v_p)$ such that
$v_j \not\in V(C_i)$ for $j \in [2, p-1]$, 
let $n(P(u, v; C_i))$ denote the number of vertices 
$\sum_{j=2}^{p-1} |V(H \langle v_j \rangle)|$, 
and we define the code $\code(P(u, v; C_i))$ 
of the path $P(u, v; C_i)$ to be the following sequence:
\begin{eqnarray}
\code(P(u, v; C_i)) &\triangleq & (n - n(P(u,v;C_i)), |P(u, v; C_i)|, \nonumber \\
&& \rank(H \langle v_1\rangle), \mul(v_1 v_2), \rank(H \langle v_2 \rangle), 
\ldots, \mul(v_{p-1 v_p}), \rank(H \langle v_p\rangle)). \nonumber
\end{eqnarray}
For a junction $u$, we define $\code^*(u)$ to be the sequence 
obtained by arranging the codes of $P(u, v; C_i)$, $i \in [1, 3]$, 
in lexicographically non-ascending order.

\section{Parent of a Mono-block 2-Augmented  Tree}
\label{sec:parent_of_2a1b}
Let $H$ be a mono-block 2-augmented tree.
Let $C_i$, $i \in [1, 3]$, denote the cycles of $H$, 
and let $u$ and $v$ denote the junctions of $H$.
Without loss of generality, we assume that 
$(\code(u), \code^*(u)) \preceq (\code(v), \code^*(v))$.
For the cycle $C^*$ such that 
$ \code(P(u, v; C^*)) \preceq \code(P(u, v; C_i))$, $i \in [1, 3]$, 
let $e^*$ denote the edge between $u$ and the neighbor of $u$ in $P(u, v; C^*)$.
We define the parent of $H$ to be the graph $H - e^*$.
%
%
\begin{lemma}
\label{lem:parent:path_len}
Let $H$ be a mono-block 2-augmented tree with $n$ vertices.
Let $C_i$, $i \in [1, 3]$ denote the cycles of $H$, 
and let $u$ and $v$ denote the junctions of $H$ 
such that $(\code(u), \code^*(u)) \prec (\code(v), \code^*(v))$.
For the cycle $C^*$ such that 
$\code(P(u, v; C^*)) \preceq \code(P(u,v;C_i))$, $i \in [1, 3]$, 
the number of edges in $P(u, v; C^*)$ is at least~2.
\end{lemma}
\begin{proof}
From the definition, a cycle has at least three vertices.
Since $H$ has three distinct cycles, 
there exists an integer $i \in [1, 3]$ such that $|P(u, v; C_i)| \geq 2$.
Suppose that $P(u, v; C^*) = uv$.
We have $\code(P(u,v;C_i)) \prec \code(P(u,v;C^*))$ 
since the first entry of $\code(P(u,v;C_i))$ 
(resp., $\code(P(u, v; C^*))$) is $n - n(P(u, v; C_i)) < n-1$ 
(resp., $n - n(P(u, v; C^*)) = n-1$).
This, however, contradicts that $\code(P(u,v;C^*))$ 
is the lexicographically minimum among 
$\code(P(u, v; C_i))$, $i \in [1, 3]$.
As a result, we see that $P(u, v; C^*) \neq uv$, 
and the number of edges in $P(u, v; C^*)$ is at least 2.
\end{proof}

\section{Necessary and Sufficient Conditions for Generating Children}
\label{sec:Parent-Child}

Let $G$ be a monocyclic graph with a cycle~$C$.
For a non-adjacent vertex pair $\{x, y\}$ in $G$, 
and an integer $p \in [1, \min \{ \res(x), \res(y) \}]$,
if $\rhoG{x} = \rhoG{y}$ holds, 
then the graph $G + p \cdot xy$ will have two blocks, 
and will not be a mono-block 2-augmented tree.
Hence, in order to generate a mono-block 2-augmented tree, 
we must choose a vertex pair $\{x, y\}$ satisfying $\rhoG{x} \neq \rhoG{y}$.
Observe that then $\rhoG{x}$ and $\rhoG{y}$ will be the junctions
in the mono-block 2-augmented tree $G + p \cdot xy$.
Now, we derive necessary and sufficient conditions to determine 
whether $G + p \cdot xy$ is a child of $G$ or not.

\begin{lemma}
\label{lem:child_condition}
Let $G$ be a monocyclic graph, and let $C$ denote the cycle of $G$.
Let $x$ and $y$ be non-adjacent vertices in $V(G)$ with $\rhoG{x} \neq \rhoG{y}$,
and let $p \in [1, \min\{\res(x), \res(y)\}]$.
Let $H$ denote the graph $G + p \cdot xy$, 
where $\rhoG{x}$ and $\rhoG{y}$ become the junctions in $H$,
and let $C_i$, 
$i \in [1, 3]$, denote the three cycles in $H$.
Then $H$ is a child of $G$ if and only if the following conditions are satisfied:
\begin{description}
\item[{\rm (i)}] $x = \rhoG{x}$ and $y \neq \rhoG{y}$ (i.e., $x \in V(C)$ and $y \not\in V(C)$);
\item[{\rm (ii)}] $(\code(\rhoG{x}), \code^*(\rhoG{x})) \preceq
    (\code(\rhoG{y}), \code^*(\rhoG{y}))$; and
\item[{\rm (iii)}] $\code(P(\rhoG{x}, \rhoG{y}; C)) \preceq 
      \code(P(\rhoG{x}, \rhoG{y}; C_i))$, for all $i \in [1, 3]$.
\end{description}
\end{lemma}
\begin{proof}
{\bf Necessity}.
Let $H$ denote the mono-block 2-augmented tree constructed as $G + p \cdot xy$.
Suppose that both $x$ and $y$ are in $V(C)$, then the junctions of $H$ are $x$ and $y$. 
From Lemma~\ref{lem:parent:path_len}, we see that $H$ is not a child of $G$.
Suppose that neither of $x$ and $y$ is in $V(C)$, 
then the edge $xy$ is not incident to a junction of $H$, 
and we see that $H$ is not a child of $G$,
and hence Condition~(i) must hold.

Next, we consider the case when vertices~$x$ and~$y$ 
satisfy Condition~(i) but do not satisfy Condition~(ii).
In this case, the parent of $H$ is obtained by deleting an edge incident to $\rhoG{y}$,
and we see that $H$ is not a child of $G$.

Finally, we assume that $x$ and $y$ satisfy
Conditions~(i) and~(ii) but do not satisfy Condition~(iii).
In this case, there exists a cycle $C^* \neq C$ 
such that $\code(P(\rhoG{x}, \rhoG{y}; C^*)) \prec \code(P(\rhoG{x}, \rhoG{y}; C))$.
The parent of $H$ is obtained by deleting the edge between $x$
and the neighbor of $x$ in $P(x,\rhoG{y};C^*)$ but not the edge $xy$.
Hence we see that $H$ is not a child of $G$.

{\bf Sufficiency}.
By choosing $x = \rhoG{x} \neq \rhoG{y} \neq y$ we see
that $H = G + p \cdot xy$ is a mono-block
2-augmented tree with junctions 
$x = \rhoG{x}$ and $\rhoG{y}$
and that $xy$ is a junction pair in~$H$.
The requirements of (ii) and (iii) in the lemma
follow the definition of a parent in Section~\ref{sec:parent_of_2a1b}.
\end{proof}
%

Let $G$ be a monocyclic graph with a cycle $C$.
Let $x$ and $y$ be non-adjacent vertices in $G$ such that 
$x = \rhoG{x} \neq \rhoG{y}$, and $y \not\in V(C)$,
and let $p \in [1, \min \{ \res(x), \res(y) \}]$.
Let $H$ be the mono-block 2-augmented tree $G + p \cdot xy$.
The pendent tree $H\langle x\rangle$ is equivalent to $G\langle x\rangle$.
Let $v_y$ be the child of $\rhoG{y}$ such that $G\langle v_y \rangle$ contains $y$.
Then the pendent tree $H\langle \rhoG{y} \rangle$ 
is equivalent to $G\langle \rhoG{y}\rangle - G\langle v_y\rangle$.
From Condition~(ii) of Lemma~\ref{lem:child_condition}, 
if $|V(G\langle \rhoG{y}\rangle)| - |V(G\langle v_y\rangle)| < |V(G\langle x\rangle)|$ holds, 
then we have 
\[
 (\code(\rhoG{y}), \code^*(\rhoG{y})) \prec (\code(\rhoG{x}), \code^*(x)), 
\]
and we see that $H$ is not a child of $G$.
Let $C_i$ denote the cycles of $H$ for $i \in [1, 3]$, where $C_1 = C$.
We have $n(P(u,v;C)) = |V(G\langle v_y\rangle)|$.
From Condition~(iii) of Lemma~\ref{lem:child_condition}, 
if $|V(G\langle v_y\rangle)| < \max\{n(P(u,v;C_2)), n(P(u,v;C_3))\}$ holds, 
then $\code(P(\rhoG{x},\rhoG{y};C))$ is not the lexicographically minimum 
among $\code(P(\rhoG{x},\rhoG{y};C_i))$, $i \in [1, 3]$, and $H$ is not a child of $G$.
From these observations, we have the following lemma.

\begin{lemma}
\label{lem:children:pendent_tree_size}
Let $G$ be a monocyclic graph with $n$ vertices, 
and let $C$ denote the cycle of~$G$.
Let $x$ and $y$ be two non-adjacent vertices in $V(G)$ 
such that $x = \rhoG{x} \neq \rhoG{y}$ and $y \not\in V(C)$,
and let $p \in [1, \min \{ \res(x), \res(y) \}]$. 
If $G + p \cdot xy$ is a child of $G$, 
then it holds that $|V(G\langle \rhoG{y}\rangle)| \geq n/3$.
\end{lemma}
\begin{proof}
Let $H$ be the mono-block 2-augmented tree $G + p \cdot xy$,
and let $C_i$, $i \in [1, 3]$,
denote the cycles of $H$, where $C_1 = C$.
The junctions in $H$ are $\rhoG{x}$ and $\rhoG{y}$.
Let $n_y$ and $n'_y$ denote the numbers $n(P(\rhoG{x}, \rhoG{y}; C_2))$
and $n(P(\rhoG{x}, \rhoG{y}; C_3))$, respectively.
Without loss of generality, we assume that $n_y \leq n'_y$.
Since the number of vertices in $G$ is $n$, 
we have $n = |V(G \langle x\rangle)| + |V(G\langle \rhoG{y} \rangle)| + n_y + n'_y$.

Let $v_y$ be the child of $\rhoG{y}$ such that $G\langle v_y\rangle$ contains $y$.
If $H$ is a child of $G$, then from 
Conditions~(ii) and~(iii) of Lemma~\ref{lem:child_condition}, 
we have 
$|V(G\langle \rhoG{y}\rangle)| - |V(G\langle v_y\rangle)| \geq |V(G\langle x\rangle)|$ 
and $|V(G\langle v_y\rangle)| \geq n'_y$.

Therefore, we have
\begin{align*}
 |V(G\langle \rhoG{y}\rangle)| 
& \geq |V(G\langle x\rangle)| + |V(G\langle v_y\rangle)| \\
& \geq |V(G\langle x\rangle)| + n'_y \\
& \geq n - |V(G\langle \rhoG{y} \rangle)| - n_y - n'_y + n'_y  \\
& \geq (n - n_y) / 2.
\end{align*}
Since we have $n_y \leq n'_y \leq |V(G\langle v_y\rangle)|$ 
and $n = |V(G \langle x \rangle)| + |V(G\langle \rhoG{y} \rangle)| + n_y + n'_y$, 
we obtain $n_y \leq n/3$ and $|V(G\langle \rhoG{y}\rangle)| \geq (n - n_y) / 2 \geq n/3$.
\end{proof}

As a consequence of Lemma~\ref{lem:children:pendent_tree_size},
for a monocyclic graph $G$ with $n$ vertices and cycle~$C$, 
a  pair $\{x, y\}$ of non-adjacent vertices in $G$
with $x = \rhoG{x} \neq \rhoG{y}$
and $y \not\in V(C)$, 
and an integer $p \in [1, \min \{ \res(x), \res(y) \}]$,
if $V(G\langle \rhoG{y} \rangle)| < n/3$ holds, 
then it holds that $G + p \cdot xy$ is not a child of $G$.

\begin{lemma}
\label{lem:two_pendent_tree}
Let $G$ be a monocyclic graph with a cycle $C$
and let $r \in V(C)$.
If $G$ has a pendent tree $G\langle r^* \rangle$ 
such that $r^* \neq r$ and $|V(G\langle r \rangle)| \leq |V(G\langle r^* \rangle)|$,
then 
for any pair $\{x, y\}$
of non-adjacent vertices such that
$x \in V(C)$ and $y \in V(G\langle r \rangle)$
and an integer $p \in [1, \min \{ \res(x), \res(y) \}]$
it holds that $G + p \cdot xy$ is not a child of~$G$.
\end{lemma}
\begin{proof}
Let $r^*$ denote the root of a pendent tree 
such that $|V(G\langle r \rangle)| \leq |V(G\langle r^* \rangle)|$.
Let $H$ be the graph $G + p \cdot xy$, 
where $\rhoG{x}$ and $\rhoG{y}$ are the junctions in $H$,
and let $C_i$, $i \in [1, 3]$,
denote the cycles of $H$, where $C_1 = C$.
Suppose that $x \neq r^*$.
Then, vertex $r^*$ is included in some path
$P(\rhoG{x}, \rhoG{y}; C_j)$, $j \in \{1, 2, 3\}$, 
and it holds that $n(P(\rhoG{x}, \rhoG{y}; C_j)) \geq |V(G\langle r^* \rangle)|$.
Let $v$ be the child of $r$ such that 
$V(G\langle v \rangle)$ contains~$y$.
For the path $P(\rhoG{x}, \rhoG{y}; C)$, we have 
$n(P(\rhoG{x}, \rhoG{y}; C)) = |V(G\langle v \rangle)| < |V(G\langle r \rangle)|$.
Hence we have $\code(P(\rhoG{x}, \rhoG{y}; C^*)) \prec \code(P(\rhoG{x}, \rhoG{y}; C))$, 
which contradicts Condition~(iii) of Lemma~\ref{lem:child_condition}.

Next, suppose that $x = r^*$.
The junctions of $H$ are $r = \rhoG{y}$ and $r^* = x$.
Let $v$ be the child of $r$ such that 
$V(G\langle v \rangle)$ contains $y$,
and we have $|V(H\langle r \rangle)| = |V(G\langle r \rangle)| - |V(G\langle v \rangle)|$
and $|V(H\langle r^* \rangle)| = |V(G\langle r^* \rangle)|$.
Therefore, since we have $|V(H\langle r^* \rangle)| > |V(H\langle r\rangle)|$
it holds that $(\code(\rhoG{x}), \code^*(x)) \succ (\code(\rhoG{y}), \code^*(\rhoG{y}))$, 
and therefore $G + p \cdot xy$ is not a child of $G$,
as required.
\end{proof}

Finally, from Lemmas~\ref{lem:child_condition}, \ref{lem:children:pendent_tree_size}, 
and \ref{lem:two_pendent_tree}, we have the following necessary 
and sufficient conditions that a pair of non-adjacent
vertices in a monocyclic graph $G$ must satisfy in order
to obtain a child mono-block 2-augmented tree by adding multiple edges
between them.

\begin{lemma}\label{lem:improve:parent_child}
Let $G$ be a monocyclic graph with $n$ vertices and a cycle $C$.
Let $x$ and $y$ be non-adjacent vertices in $V(G)$ with $\rhoG{x} \neq \rhoG{y}$,
and $p \in [1, \min \{ \res(x), \res(y) \}]$. 
Let $H$ denote the graph $G + p \cdot xy$, 
where $\rhoG{x}$ and  $\rhoG{y}$ are the junctions in~$H$, and
let $C_i$, $i \in [1, 3]$ denote the three cycles in $H$.
Then $H$ is a child of $G$ if and only if the following conditions are satisfied:
\begin{description}
\item[{\rm (i)}] $|V(G\langle \rhoG{y} \rangle)| \geq n/3$;
\item[{\rm(ii)}] $|V(G\langle \rhoG{y} \rangle)| > |V(G\langle r \rangle)|$ 
      for each $r \in V(C) \setminus \{\rhoG{y}\}$;
\item[{\rm (iii)}] $x = \rhoG{x}$ and $y \neq \rhoG{y}$, (i.e., $x \in V(C)$ and $y \not\in V(C)$);
\item[{\rm (iv)}] $(\code(\rhoG{x}), \code^*(\rhoG{x})) 
	\preceq (\code(\rhoG{y}), \code^*(\rhoG{y}))$; and
\item[{\rm (v)}] $\code(P(\rhoG{x}, \rhoG{y}; C)) 
	\preceq \code(P(\rhoG{x}, \rhoG{y}; C_i))$, $i \in [1, 3]$.
\end{description}
\end{lemma}

\subsection{Preprocessing for Efficient Computation}
\label{sec:preprocess}

Notice that in Lemma~\ref{lem:improve:parent_child}\,(iv) and~(v),
in order to check whether a mono-block 2-augmented tree $H$
obtained by adding an edge to a pair of nonadjacent vertices in
a monocyclic graph~$G$ is indeed a child of~$G$ or not,
requires our knowledge of the rank of pendent trees
of $H$ in the set $\mathcal{T}(H)$ of all pendent trees in~$H$.
This computation might seem wasteful,
as a single monocyclic graph $G$ may have many
candidates for children  mono-block 2-augmented trees.
We here give an observation that
there exists a set $\mathbb{T}$ of selected pendent trees of $G$
and their subtrees, such that this set will contain as a subset
the set of pendent trees of any graph $H$ obtained by adding an edge
between a pair of non-adjacent vertices in~$G$.
Then, to save on computation effort,
we calculate the rank of rooted trees in this set $\mathbb{T}$
only once per monocyclic graph~$G$.

Let $G$ be a monocyclic graph with a cycle $C = (v_0, v_1, \ldots, v_{n-1}, v_0)$.
Then, in addition to the set 
$\mathcal{T}(G) = \{G \langle v_i \rangle \mid i \in [0, n-1]\}$ 
of pendent trees, we define the following sets of rooted trees\\
~-~ $\widehat{\mathcal {T}}(G) =
	  \{ G \langle u \rangle \mid u \in V(G \langle v_0 \rangle) \setminus \{v_0\} \}$ \\
~-~ $\widetilde{\mathcal {T}}(G) =
	\{G \langle \parent(u) \rangle - G \langle u \rangle \mid
	      u \in V(G \langle v_0 \rangle) \setminus \{v_0\} \}$. \\
Finally, we define the union 
$\mathbb{T}(G) = \mathcal{T}(G) \cup \widehat{\mathcal {T}}(G) \cup \widetilde{\mathcal {T}}(G)$.

\begin{lemma}
 \label{lem:subset_pendent_trees}
 Given a monocyclic graph $G$ with a cycle $C = (v_0, v_1, \ldots, v_{n-1}, v_0)$,
 let $\{x, y\} \in \overline{E}(G)$ be a pair of non-adjacent vertices
 such that $x \in V(C) \setminus \{v_0\}$ and 
 $y \in V(G \langle v_0 \rangle ) \setminus \{v_0\}$,
 and let $H = G + xy$ denote the  graph obtained from $G$
 by adding an edge $xy$.
 Then it holds that
 \[
  \mathcal{T}(H) \subseteq \mathcal{T}(G) \cup \widetilde{\mathcal {T}}(G).
 \]
\end{lemma}
\begin{proof}
By the choice of $x$ and $y$ the graph $H$ is a mono-block
 2-augmented tree and
 the junctions of $H$ are the vertices $x$ and $v_0$.
 Let $P_i$, $i = 1, 2, 3$, denote the three $x,v_0$-paths in $H$,
 such that $xy \in E(P_3)$, i.e., $P_3 = P(x, y; C)$.
 The pendent trees $G \langle v_i \rangle$, $i \in [1, n-1]$
 are preserved between $G$ and $H$.
 Therefore, we focus on the pendent tree $G \langle v_0 \rangle$.
 Denote by $Q = (u_0 = v_0, u_2, \ldots, u_{k-1} = y)$ the $v_0, y$-path
 in the tree $G \langle v_0 \rangle$.
 Then,  adding the edge $xy$ to path $Q$ we obtain $P_3$ in $H$,
 and for each rooted tree $G \langle u_i \rangle$, $i \in [0, k-2]$,
 the tree $G  \langle u_i \rangle - G \langle u_{i+1} \rangle$
 becomes the rooted tree $H \langle u_i \rangle$ in the
 mono-block 2-augmented tree~$H$ (see Fig.~\ref{fig:deleted_trees}).
 Equivalently, for $i \in [1, k-1]$, the tree
 $G  \langle u_{i-1} \rangle - G \langle u_{i} \rangle$
 becomes the pendent tree 
 $H \langle u_i \rangle$ in $H$.
 Since for $i \in [1, k-1]$ it holds that $\parent(u_i) = u_{i-1}$
 in $G \langle v_0 \rangle$,
 the claim follows. %
\begin{figure}[!ht]
\centering
 \includegraphics[width = 0.88 \textwidth]{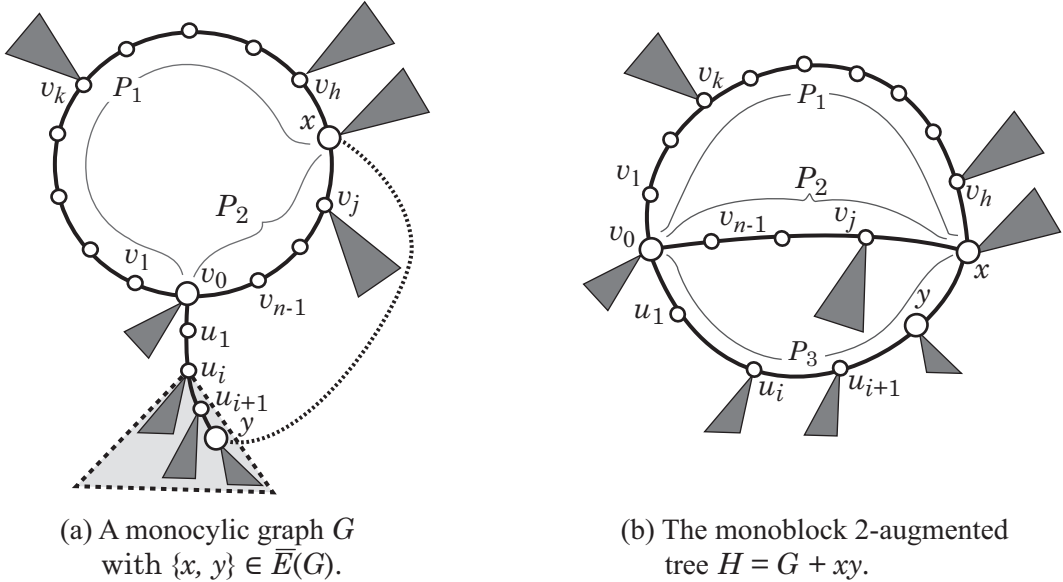}
 \vspace{-1cm}
 \caption {
 (a)~A monocyclic graph~$G$ where the unique cycle is denoted by
	$C = (v_0, v_1, \ldots, v_{n-1}, v_0)$,  
	with a pair $\{x, y\}$ of non-adjacent vertices such that
	$x \in V(C)$ and $y \in V(G \langle v_0 \rangle) \setminus \{v_0\}$.
(b)~The mono-block 2-augmented tree $H = G + xy$ obtained
    by adding an edge $xy$ to the monocyclic graph~$G$ in~(a).
    The subtrees denoted by dark gray are preserved from~$G$ in~$H$.	
   }
 \label{fig:deleted_trees}
\end{figure}
\end{proof}

\begin{lemma}
 Given an $n$-vertex $\Sigma$ -colored monocyclic graph $G$ with
 multiplicity at most $d$ and a unique cycle $C = (v_0, v_1, \ldots, v_{m-1}, v_0)$,
 the rank of all trees in $\mathbb{T}(G)$ can be computed
 in $O(n \cdot (n + |\Sigma| + d))$ time in total.
\end{lemma}
\begin{proof}
The set 
$\mathbb{T}(G)$ is composed of the
set $\mathcal{T}(G)$ of pendent trees of the graph $G$, 
the set $\widehat{\mathcal{T}}(G)$ of rooted subtrees
of pendent trees of~$G$ and trees in the set 
$\widetilde{\mathcal{T}}(G)$ that are obtained 
as a difference between rooted subtrees of $G$.
Then, by the observation 
made in Section~\ref{sec:left-heavy}
that for a left-heavy tree $T$
and any rooted subtree $T'$ of $T$
the sequences $\delta(T')$ and $\mm(T')$
are continuous subsequences of 
$\delta(T)$ and $\mm(T)$,
it is not difficult to obtain the 
signatures of all trees in the set $\widetilde{\mathcal{T}}(G)$.
Finally, since the total number of vertices of trees in the 
set $\widetilde{\mathcal{T}}(G)$ is not more than
that of the set $\widehat{\mathcal{T}}(G)$,
and Lemma~\ref{lem:ranking_all_subtrees}, the claim follows.
\end{proof}

\subsection{A Procedure to Verify Child Conditions}
\label{sec:children-check}

We show an algorithm 
that for a given monocyclic graph $G$, 
a pair $\{x, y\}$ of non-adjacent vertices in $G$,
and an integer $p \in [1, \min \{ \res(x), \res(y) \}]$,
based on Lemma~\ref{lem:improve:parent_child}
determines whether
$G + p \cdot xy$ is a child of $G$ or not in Procedure~{4}~{\sc ChildCheck}.

\bigskip
\noindent
{\bf Procedure~{4}} {\sc ChildCheck}$(G, \{x, y\}, p)$
\begin{algorithmic}[1]
\Require 
    A monocyclic graph $G$ with $n$ vertices and a cycle $C$, 
    a pair $\{x, y\} \in \overline{E}(G)$ of non-adjacent vertices,
    and an integer~$p \in [1, \min \{ \res(x), \res(y) \} ]$.
\Ensure {\tt True} if $G + p \cdot xy$ is a child of $G$, and {\tt False} otherwise.
\State{Answer := {\tt False};}
\If {$|V(G \langle \rhoG{y} \rangle)| \geq n/3$}
  \If {$|V(G \langle \rhoG{y} \rangle)| > |V(G \langle r \rangle )|$ 
	for each $r \in V(C) \setminus \{\rhoG{y}\}$}
    \If {$x = \rhoG{x} \neq \rhoG{y} \neq y$}
      \State{Construct $H := G + p \cdot xy$;} /* $\rhoG{x}$ and $\rhoG{y}$ are the junctions in $H$ */
      \State{Let $C_i$, $i \in [1, 3]$ denote the cycles of $H$;}
      \If{$\code(\rhoG{x}) \prec \code(\rhoG{y})$}
	\If{$\code(P(\rhoG{x}, \rhoG{y}; C)) \preceq \code(P(\rhoG{x}, \rhoG{y}; C_i)), i \in [1, 3]$}
	  \State{Answer :=  {\tt True}}
	\EndIf
      \ElsIf{$\code(\rhoG{x}) = \code(\rhoG{y})$}
	\If{$\code^*(\rhoG{x}) \preceq \code^*(\rhoG{y})$}
	  \If{$\code(P(\rhoG{x}, \rhoG{y}; C)) \preceq \code(P(\rhoG{x}, \rhoG{y}; C_i)), i \in [1, 3]$}
	    \State{Answer :=  {\tt True}}
	  \EndIf
	\EndIf
      \EndIf
    \EndIf
    \EndIf
  \EndIf;
\State {\bf output } Answer.
\end{algorithmic}

\section{Intra-Duplication of Mono-block 2-Augmented Trees}
\label{sec:avoid_duplicate}

The parent-child relationship helps us avoid inter-duplication,
that is, generating isomorphic structures by adding edges to
topologically different monocyclic graphs.
However, the parent-child relationship is not sufficient to 
eliminate intra-duplications, 
since isomorphic children might occur from 
a single monocyclic graph $G$ by adding an edge between a pair of
non-adjacent vertices.
Henceforth, we treat a graph as a labeled one and use the information on the labeling, 
since there is no other way to distinguish isomorphic mono-block 2-augmented 
trees generated from a single monocyclic graph.

For a monocyclic graph $G$, two distinct non-adjacent vertex pairs $\{x, y\}$ and $\{x', y'\}$, 
and integers $p \in [1, \min \{ \res(x), \res(y) \}]$ 
and $q \in [1, \min \{ \res(x'), \res(y')\} ]$,
we examine under which conditions 
$G + p \cdot xy$ and $G + q \cdot x'y'$ are isomorphic.

\begin{theorem}\label{thm:isomorphism} 
Let $G$ be a connected graph that contains 
exactly one cycle $C=(v_0 ,v_1,$ $\ldots,$ $v_{n-1},$ $v_0)$, and 
 $\{x_i,y_i\}$, $i=1, 2$, be  two  pairs of non-adjacent vertices  in $G$
such that
$x_1, x_2\in V(C)\setminus \{v_0\}$,
$y_1, y_2\not\in V(C)$ and $\rho_G(y_1)=\rho_G(y_2)=v_0$. 
Let  $H_i$, $i=1,2$, denote the graph $G + x_iy_i$,  and
$c$ be a coloring  of the graph $H_1+x_2y_2$. 
Let $w_2$ denote the child of $v_0$ in the rooted tree $G\langle v_0\rangle$
such that  $G\langle w_2\rangle$ contains $y_2$. 
Assume that $H_1$ and $H_2$ are isomorphic. 
Then one of the following holds.
 \begin{enumerate}
\item[{\rm (i)}]
    $y_1\neq y_2$, $c(x_1y_1)=c(x_2y_2)$,  
    and $G\langle v_0\rangle$ has an automorphism $\xi$
 such that $\xi(v_0)=v_0$ and $\xi(y_1)=y_2$;
\item[{\rm (ii)}]
    $y_1=y_2$, $c(x_1y_1)=c(x_2y_2)$,   and $G$ has an automorphism $\xi$ such that 
    $\xi(x_1)=x_2$ and   
    $\xi(v_i)=v_{n-i \mod n}$ for each vertex $v_i\in V(C)$; and  
 \item[{\rm (iii)}]
    $y_1=y_2$, $c(x_1y_1)=c(x_2y_2)=c(v_0w_2)$, 
    $G\langle w_2\rangle$ has an automorphism $\phi$ such that
    $\phi(w_2)=y_1$,  and 
    $G- G\langle w_2\rangle$ has an automorphism $\xi$ with an integer $k\geq 1$
    such that $\xi(x_1)=x_2$ and 
$\xi(v_i)=v_{i+k \mod n}$, $v_i\in V(C)$. 
\end{enumerate}
\end{theorem}
\begin{proof}
 For a subgraph $A$ of $H_i$, $i=1,2$, let $A^{(i)}$ denote the subgraph
 of $H_i$ induced by the vertices in $V(A)$ and 
 $V(H_i\langle v\rangle)$, $v\in A$.
Without loss of generality assume that  $\{x_1, y_1\}\neq \{x_2, y_2\}$,  
 $x_1=v_{j_1}$, $x_2=v_{j_2}$ and $1\leq j_1\leq j_2\leq n-1$.
For each $i=1,2$, let $P_i$, $Q_i$ and $R_i$ denote the paths
between  vertices $x_i$ and $\rho_G(y_i)$ in $H_i$, 
where we assume that  $y_i\in V(P_i)$, $v_1\in V(Q_1)$ and $v_{n-1}\in V(Q_2)$. 
Let 
$\psi:V(H_1)\to V(H_2)$ be an isomorphism between $H_1$ and $H_2$. 
Since 
$H_i$ for each $i=1,2$ has exactly one block $B_i$ with exactly two junction
vertices $x_i$ and $\rho_G(y_i)=v_0$, 
each path between $x_1$ and $v_0$ in $H_1$ is mapped by $\psi$ to
a path between $x_2$ and $v_0$ in $H_2$.
This means that  $c(x_1y_1)=c(x_2y_2)$,  $\psi(\{x_1,v_0\})=\{x_2,v_0\}$, 
$\psi(B_1)=V(B_2)$,   $|V(B_1)|=|V(B_2)|$, 
and   $\{\psi(P_1^{(1)}), \psi(Q_1^{(1)}), \psi(R_1^{(1)})\}
=  \{ V(P_2^{(2)}) ,  V(Q_2^{(2)}),  V(R_2^{(2)})\}$. 
Note that  $|V(B_i)|=|V(C)|+|V(P_i)|-2$ for each $i=1,2$.
Since 
$|V(B_1)|=|V(B_2)|$, we see that $|V(P_1)|=|V(P_2)|$ and
 $\{|V(Q_1)|,|V(R_1)|\}=\{|V(Q_2)|,|V(R_2)|\}$.

(i) Assume that $y_1\neq y_2$.
Let $z$ denote the deepest vertex  in $V(P_1)\cap V(P_2)$ 
in the rooted  tree $G\langle v_0\rangle$,   where  possibly $z=v_0$.
When $z\neq v_0$, let   $Z$ 
denote the path from $w_2$ to $z$ in the subtree $G\langle v_0\rangle$,
where we regard $V(Z)$ as an empty set when $z=v_0$.
For each $i=1,2$, let $z_i$ denote the child of $z$
such that  $y_i\in V(G\langle z_i\rangle)$ and 
$\overline{P}_i$  denote the subpath of $P_i$ between $z_i$ and $y_i$.
Note that  
\begin{equation}\label{eq:z(-1)}
\psi(z)\not\in  \psi^{-1}(\overline{P}_2),
\end{equation} 
since otherwise
$H_2\langle z\rangle$ would be isomorphic to a proper subgraph 
$H_2\langle t\rangle$   for the vertex $t=\psi(v_h)\in V(\overline{P}_2)$.
 
To show that $G\langle v_0\rangle$ has an automorphism $\xi$
such that $\xi(v_0)=v_0$ and $\xi(y_1)=y_2$,
it suffices to prove that 
\begin{equation}
\label{eq:Hz}
\mbox{ $H_2\langle z\rangle$ and $H_1\langle z\rangle$ are $(z,z)$-isomorphic.}
\end{equation}

Case~1. $\psi(z)\not\in V(G\langle v_0\rangle)$: 
In this case, $\psi(P_1)=V(P_2)$  and ``$\psi(v_0)=v_0$ or  $z\neq v_0$.'' 
Now 
 $H_1 \langle z\rangle$ and $H_2 \langle \psi(z)\rangle$ 
are $(z,\psi(z))$-isomorphic.
Hence if  $z=v_0$ and $\psi(v_0)=v_0$, where $\psi(z)=z$, then 
(\ref{eq:Hz}) holds.
Assume that  $z\neq v_0$.
We know that $H_1 \langle \psi(z)\rangle$ and $H_2 \langle \psi^2(z)\rangle$ 
are $(\psi(z),\psi^2(z))$-isomorphic.
By (\ref{eq:z(-1)}),    $\psi(z)\in V(Z)$ 
holds, implying  $H_2 \langle \psi(z)\rangle=H_1 \langle \psi(z)\rangle$.
Since $\psi$ maps path $P_1-\{v_0,x_1\}$ to path $P_2-\{v_0,x_2\}$, 
we see that  $\psi^2(z)=z$.  
Therefore $H_1 \langle z\rangle$ and  $H_2 \langle z\rangle$ 
are $(z,z)$-isomorphic. 
   
Case~2.  $\psi(z)\in V(G\langle v_0\rangle)$: 
In this case, ``$\psi(P_1)=V(P_2)$, $\psi(x_1)=v_0$ and  $z= v_0$''
or ``$\psi(P_1)\in \{V(Q_2),V(R_2)\}$.''  
Let $h,k\in [1,n-1]$ denote the indices such that 
 $v_h=\psi(z)\in \psi(P_1)$ and  
 $v_k=\psi^{-1}(z)\in \psi^{-1}(P_2)$.
 Define subtrees $T(v)$, $v\in V(C)$ to be $G\langle v\rangle$ if $v\neq v_0$
 and  $T(v_0) =H_2\langle v_0\rangle$.
 %
For the subset $S_1=V(C)\setminus \psi^{-1}(\overline{P}_2)$  of $V(C)$,
define a function $f: S_1\to V(G)$ such that
 \[ f(v)=\left \{ \begin{array}{ll}
   \psi(v) & \mbox{ if } v\in S_1\setminus  \psi^{-1}(Z) \\
    \psi^2(v) & \mbox{ if } v\in \psi^{-1}(Z) .
 \end{array} 
 \right.
 \] 
We here prove the following properties: \\
(a-1)~For the subset  $S_2=V(C)\setminus \psi(\overline{P}_1)$ of $V(C)$,
  $f$ is a bijection from $S_1$ to $S_2$; \\
(a-2)~Let $v^{\dagger}=v_k$ if $z\neq v_0$ and
$v^{\dagger}=v_0$ if $z=v_0$.
Then $f(v^{\dagger})=v_h$.
For each vertex $v\in S_1\setminus\{v^{\dagger}\}$, 
  $T(v)$ and $T(f(v))$ are $(v,f(v))$-isomorphic; and \\
(a-3)~$f^p(v_h)\in S_2\setminus \psi(\overline{P}_2)\subseteq S_1$
 for any integer  $p\geq 1$.
 
 (a-1)~Since $S_1\setminus \psi^{-1}(Z)
 =V(C) \setminus \psi^{-1}(\overline{P}_2) \setminus \psi^{-1}(Z) 
 =V(C)\setminus \psi^{-1}(V(P_2)\setminus \{v_0\})$,
 the set $S_1\setminus \psi^{-1}(Z)$ is mapped by 
 $\psi$ to $V(Q_2)\cup V(R_2)\subseteq V(C)$.
 Observe that  $f(\psi^{-1}(Z))=\psi(Z)
 \subseteq V(C)$
 and $V(Z)\subseteq V(P_1)\setminus \{v_0\}\subseteq V(G)\setminus V(C)$.
 Since  $S_1\setminus \psi^{-1}(Z)$ and $V(Z)$  are disjoint
 and $\psi$ is a bijection from $V(G)$ to $V(G)$,
 we see that  $f(S_1\setminus \psi^{-1}(Z))=\psi(S_1\setminus \psi^{-1}(Z))$ 
 and $f(\psi^{-1}(Z))=\psi(Z)$  are   disjoint.  
 Since $V(\overline{P}_1)$ is disjoint with  $(S_1\setminus \psi^{-1}(Z)) \cup V(Z)$,
 this also means that $f(S_1)\subseteq V(C)\setminus \psi(\overline{P}_1)=S_2$.
Therefore  $f(S_1)=S_2$, since $|S_1|=|S_2|$ and $f$ is a bijection from $S_1$ to $f(S_1)$.

 (a-2)~We distinguish two cases.
 
 Case of $z\neq v_0$, where $v^{\dagger}=v_k \in \psi^{-1}(z)\in \psi^{-1}(Z)$: 
 Then $f(v^{\dagger})=f(v_k)=\psi^2(v_k)=\psi(z)=v_h$.
 Let  $v\in S_1\setminus\{v_k\}$.
 If  $\psi(v)\in V(C)$, then
    $T(v)=H_1\langle v\rangle$ and 
    $T(f(v))=H_2\langle \psi(v)\rangle$ are $(v,\psi(v))$-isomorphic. 
  Note that for any vertex $u\in V(Z)\setminus\{z\}$,  
    $H_1\langle u\rangle=H_2\langle u\rangle$.
    Also $v\neq v_k$ means that   $\psi(v)\neq z$. 
If   $\psi(v)\in V(Z)\setminus\{z\}$,  then
   $H_1\langle v\rangle$ and $H_2\langle \psi(v))$ are $(v,\psi(v))$-isomorphic
  and
     $H_1\langle \psi(v)\rangle$ and $H_2\langle \psi^2(v)\rangle$ 
     are $(\psi(v),\psi^2(v))$-isomorphic,
  implying that
    $T(v)=H_1\langle v\rangle$ and 
    $T(f(v))=H_2\langle \psi^2(v)\rangle$ are $(v,\psi^2(v))$-isomorphic.
  
  Case of $z=v_0$, where  $v^{\dagger}=v_0$:
 In this case,  $V(Z)=\emptyset$ and $f=\psi$.
 Then $f(v^{\dagger})=f(v_0)=\psi(z)=v_h$.
Therefore for any vertex $v\in V(C)\setminus\{v_0\}$,  
   $T(v)=H_1\langle v\rangle$ and $T(f(v))=H_2\langle \psi(v)\rangle$ 
   are $(v,\psi(v))$-isomorphic.

(a-3)~By definition, 
 $v_h=\psi(z)\in \psi(P_1)\subseteq V(Q_2)\cup V(R_2)\subseteq V(C)$.
By (\ref{eq:z(-1)}),   $v_h\in  V(C)\setminus \psi^{-1}(\overline{P}_2)=S_1$.   
By  (a-2),  $f$ maps a vertex $v\in S$ to a vertex $f(v)$ so that 
 $T(v)$ and  $T(f(v))$ are $(v,f(v))$-isomorphic.
 Hence  for any integer $i\geq 1$,
 $f^i(v_h)$ is not a vertex in  $\psi^{-1}(\overline{P}_2)$,
 since otherwise $T(v)$ would be isomorphic to a tree $T(f(v))$ that
 is a proper subgraph of $T(v)$. 
Therefore $f^i(v_h)$, $i\geq 1$ is  a vertex in 
$V(C)\setminus \psi^{-1}(\overline{P}_2)=S_1$.

We are ready to prove  (\ref{eq:Hz}). 
By (a-3), it holds $f^0(v_h),f(v_h),f^2(v_2),\ldots,f^{|S_1|+1}\in S_1$, and
there is an integer $i\in [1,|S_1|]$ such that $f^j(v_h)=f^i(v_h)$
for some $j\in [0,i-1]$.
Let $p$ denote the smallest such integer $i$,
where $j=0$ and $v_h=f^p(v_h)$ since $f$ is a bijection by (a-1). 
By (a-2),   if $v_0\neq z$ (resp., $v_0=z$),  
then $f^{p-1}(v_h)=f^{-1}(v_h)=v^{\dagger}=v_k$ 
(resp., $f^{p-1}(v_h)=f^{-1}(v_h)=v^{\dagger}=v_0$) 
 and  
 $T(v_h)$  and $T(v_k)$ are $(v_h,v_k)$-isomorphic
 (resp.,  
 $T(v_h)$ and $T(v_0)=H_2\langle z\rangle$  are  $(v_h,v_0)$-isomorphic).
 This proves~(\ref{eq:Hz}). 
  
In the following, we assume that $y_1= y_2$, from which  
it follows that
$x_1\neq x_2$, $j_1=n-j_2$ and $|V(Q_1)|=|V(Q_2)|<|V(R_1)|=|V(R_2)|$.
Hence $\psi(R_1^{(1)})\in \{V(P_2^{(2)}), V(R_2^{(2)})\}$.
We first prove that
 \begin{equation}\label{eq:P1P2}
 \mbox{
$P_1^{(1)}$ and $P_2^{(2)}$  have
 an isomorphism $\eta$ such that $\eta(x_1)=x_2$ and $\eta(v_0)=v_0$.
} \end{equation}
Recall that  $c(x_1y_1)=c(x_2y_2)$.
To prove (\ref{eq:P1P2}),  it suffices to show that
 $G\langle x_1\rangle$ and $G\langle x_2\rangle$ are $(x_1, x_2)$-isomorphic,
 which immediately holds when  $\psi(x_1)=x_2$. 
 When $\psi(x_1)=v_0$ and $\psi(v_0)=x_2$, 
 we see that 
 $G\langle x_1\rangle$ and $G\langle x_2\rangle$ are $(x_1, x_2)$-isomorphic,
 because 
 $H_1\langle x_1\rangle=G\langle x_1\rangle$ and 
 $H_2\langle v_0\rangle=H_1\langle v_0\rangle$ are $(x_1, v_0)$-isomorphic
 and  
 $H_1\langle v_0\rangle$ and $H_2\langle x_1\rangle=G\langle x_2\rangle$ 
 are $(v_0, x_1)$-isomorphic.
 This proves~(\ref{eq:P1P2}).

We distinguish two cases. 
  
 (ii) Assume that $y_1= y_2$ and ``$\psi(x_1)=x_2$ or $\psi(P_1)\in \{V(Q_2),V(R_2)\}$.''  
 For two indices $i,j\in[0,n-1]$ with $i\leq  j$ (resp., $i> j$),
  let $G[i,j]$ denote the subpath of $G$ induced by the 
  vertices $v_{\ell}$ with $\ell\in [i, j]$ (resp., $\ell\in [i, n-1]\cup [0, j]$).
  Note that $R_1=G[j_1,0]$ and $R_2=G[0,j_2]$.
 In this case of (ii), we   prove that
 \begin{equation}\label{eq:eta}
 \begin{array}{l}
   \mbox{
$G[j_1,0]^{(1)}$ and $G[0,j_2]^{(2)}$ admit an isomorphism $\eta$
} \\
\mbox{such that $\eta(v_{j_1})=v_{j_2}$ and $\eta(v_0)=v_0$. 
}
 \end{array}
\end{equation}
When such an isomorphism $\eta$ exists, 
 then $G$ has an automorphism $\xi$ such that 
$\xi(x_1)=x_2$ and   
$\xi(v_i)=v_{n-i \mod n}$ for each vertex $v_i\in V(C)$. 
 
 In what follows, we prove (\ref{eq:eta}). 
  When 
  $\psi(x_1)=x_2$ and  $\psi(R_1^{(1)})=V(R_2^{(2)})$,
  we see that (\ref{eq:eta}) holds.

  We first consider the case of $\psi(P_1^{(1)})=V(R_2^{(2)})$, where
    $\psi(R_1)=V(P_2)$ and $\psi(P_1)=V(R_2)$.
    We see that $R_1^{(1)}$ and $P_2^{(2)}$ are $(v_0,v_0)$-isomorphic
  (resp.,  $(x_1,v_0)$-isomorphic) 
  and  $P_1^{(1)}$ and $R_2^{(2)}$ are $(v_0,v_0)$-isomorphic
  (resp.,  $(x_1,v_0)$-isomorphic) 
  if $\psi(x_1)=x_2$ (resp., $\psi(x_1)=v_0$).
 By (\ref{eq:P1P2}),  this means that
  $R_1^{(1)}$ and $R_2^{(2)}$ are $(v_0,v_0)$-isomorphic,
  implying (\ref{eq:eta}).

  We next consider that  $\psi(x_1)=v_0$ and $\psi(P_1^{(1)})=V(Q_2^{(2)})$, where
    $\psi(Q_1)=V(P_2)$
    and  
    $\psi^{-1}(P_2)= V(G[0,j_1])$ and $\psi(P_1)=V(G[j_2, 0])$. 
In this case, 
$G[0, j_1]^{(1)}$ and $G[j_2, 0]^{(2)}$ are $(v_0, v_0)$-isomorphic. 
We observe that $G[j_1,0]^{(1)}$ (resp., $G[0,j_2]^{(1)}$) is a repetition of $G[0,j_1]^{(1)}$
 (resp., $G[j_2,0]^{(1)}$) in the following sense. 
Let $a=\lfloor\frac{j_2 - j_1}{j_1}\rfloor$ and
$b=(j_2 - j_1)-a\cdot j_1$.
 Since
$G[0,j_1]^{(1)}$ and $G[j_1,2 j_1]^{(1)}$ are $(v_0,v_{j_1})$-isomorphic
under the isomorphism $\psi$, we see that for each integer $j\in [0,a-1]$,
$G[0,j_1]^{(1)}$ and $G[j\cdot j_1, (j+1)\cdot j_1]^{(1)}$ 
are $(v_0,v_{(j+1)\cdot j_1})$-isomorphic.
When $b\geq 1$, 
$G[0,b]^{(1)}$ and $G[a\cdot j_1, a\cdot j_1+b]^{(1)}$ 
are $(v_0, v_{a\cdot j_1})$-isomorphic.
Symmetrically 
$G[j_2,0]^{(2)}$ and $G[n-2 j_2, n-j_2]^{(2)}$ are $(v_0,v_{ n-j_2})$-isomorphic,
and we see that $G[j_1,0]^{(2)}$ and $G[n-(j+1)\cdot j_1, n-j\cdot j_1]^{(2)}$ 
are $(v_0,v_{n-j\cdot j_1})$-isomorphic for each integer $j\in [0,a-1]$,
where for  $b\geq 1$, 
$G[n-b,0]^{(2)}$ and $G[n-a\cdot j_1-b, n-a\cdot j_1]^{(2)}$ 
are $(v_0,v_{n-a\cdot j_1})$-isomorphic.
Recall that
$G[0,j_1]^{(1)}$ and $G[j_2,0]^{(2)}$ are $(v_0,v_0)$-isomorphic. 
Hence  $G[j_1,0]^{(1)}$ and $G[0,j_2]^{(2)}$, which are repetitions of $G[0,j_1]^{(1)}$
and $G[j_2,0]^{(2)}$, respectively,  are  $(v_0,v_0)$-isomorphic. 
This proves~(\ref{eq:eta}).

(iii) Finally assume that $y_1= y_2$, $\psi(x_1)=v_0$ and $\psi(P_1)=V(P_2)$.
 Since $P_1^{(1)}$ and $P_2^{(2)}$ are $(x_1,v_0)$-isomorphic,
 we see that $c(x_1y_1)=c(x_2y_2)=c(v_0w_2)$, 
$G\langle w_2\rangle$ has  an automorphism $\phi$ such that
$\phi(w_2)=y_1$.
Let $k=|E(Q_1)|$. 
Now $R_1^{(1)}$ and $R_2^{(2)}$ are $(x_1,v_0)$-isomorphic
and $Q_1^{(1)}$ and $Q_2^{(2)}$ are $(v_0,x_2)$-isomorphic.
This means that  
$G- G\langle w_2\rangle$ has an automorphism $\xi$  
such that $\xi(x_1)=x_2$ and 
$\xi(v_i)=v_{i+k \mod n}$, $v_i\in V(C)$. 

\begin{figure}[htbp]
\begin{center}
\includegraphics[scale=0.40]{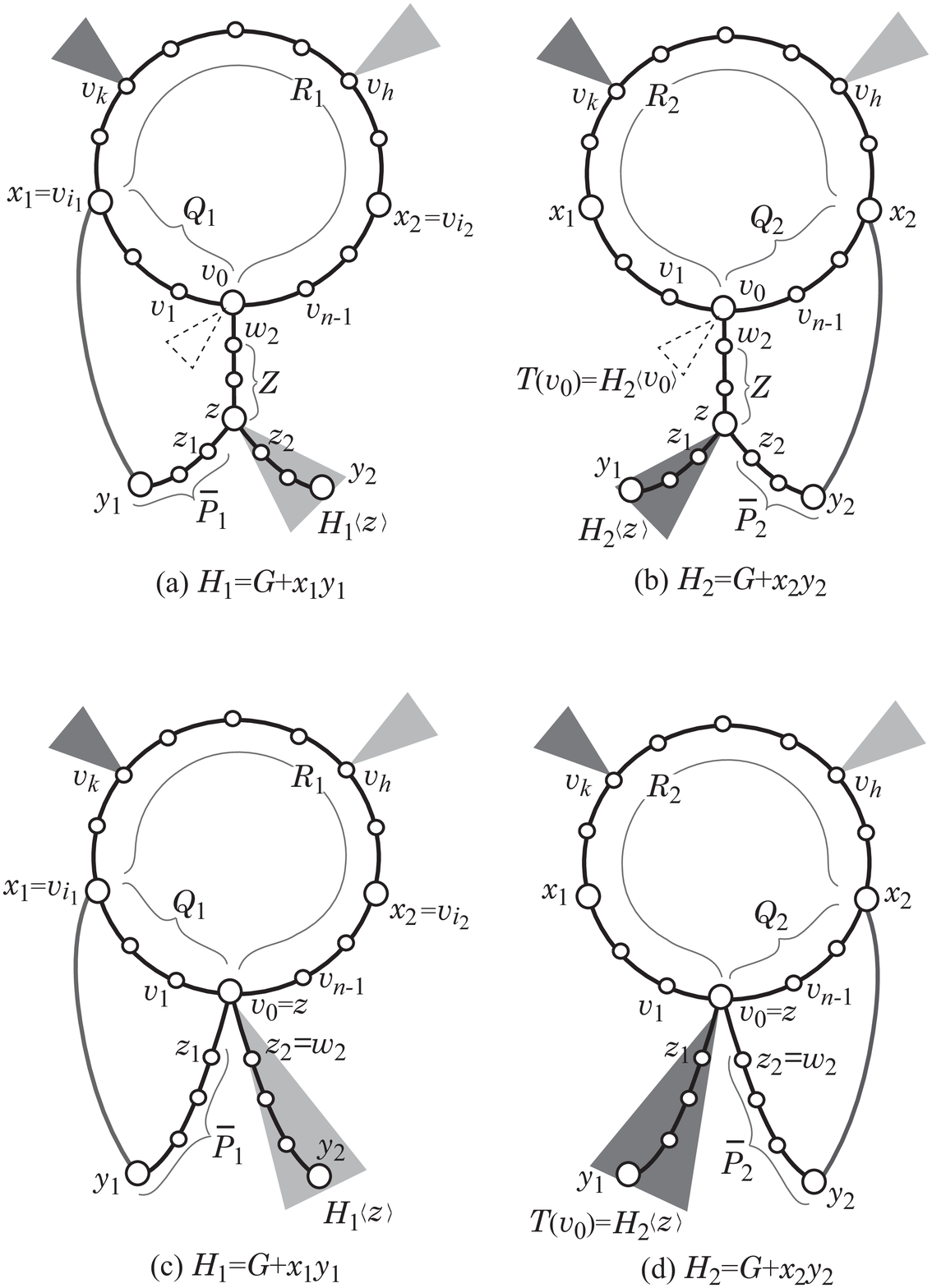}
\end{center}
\caption{Graphs augmented from $G$ by adding an edge $x_iy_i$, $i=1,2$:
(a)~$H_1=G+x_1y_1$ with $z\neq v_0$,
(b)~$H_2=G+x_2y_2$ with $z\neq v_0$, 
(c)~$H_1=G+x_1y_1$ with $z= v_0$,  
(d)~$H_2=G+x_2y_2$ with $z= v_0$.}
\label{fig:augmented_isomorphism}
\end{figure}
\end{proof}

In addition, we have the following lemma.

\begin{lemma}\label{lem:isomorphism-len}
Let $G$ be a connected graph that contains 
exactly one cycle $C=(v_0, v_1, \ldots,$ $v_{n-1},$ $v_0)$, 
let $x_1$ and $x_2$ be distinct vertices in $V(C)\setminus \{v_0\}$, 
where $x_1 = v_k$, and
let $y$ be a vertex in  $V(G \langle v_0 \rangle) \setminus \{v_0\}$.
For $i=1,2$, let $p_i \in [1, \min \{ \res(x_i), \res(y) \} ]$,
and let  $H_i$ denote $G + p_i \cdot x_iy$.
If $H_1$ and $H_2$ are isomorphic, then it holds that $x_2 = v_{n - k}$.
\end{lemma}
\begin{proof}
Let $B_i$, $i = 1, 2$, denote the block in $H_i$, and
let $P_1$, $P_2$, and  $P_3$ be the paths in $H_1$ from $x_1$ to $v_0$, 
where $y \in V(P_1)$, $v_1 \in V(P_2)$ and $v_{n-1} \in V(P_3)$.
Let $Q_1$, $Q_2$, $Q_3$ be the paths in $H_2$ from $x_2$ to $v_0$, 
where $y \in V(Q_1)$, $v_1 \in V(Q_2)$ and $v_{n-1} \in V(Q_3)$.
We have $ |V(C)| + V(P_1) - 2 = |V(B_1)| = |V(B_2)| = |V(C)| + V(Q_1) - 2$. 
Hence, we get $|P_1| = |Q_1|$
and  $\{|P_2|, |P_3|\} = \{|Q_2|, |Q_3|\}$.

Suppose that $|P_2| = |P_3|$ and $|Q_2| = |Q_3|$ hold.
Then we have $|P_2| = |P_3| = |Q_2| = |Q_3| = n/2$.
Since the length of a path is an integer, in order to satisfy the condition, 
$n$ must be even, and then $x_1 = x_2 = v_{n/2}$ holds.
However, this contradicts that $x_1$ and $x_2$ are distinct vertices.
Therefore, we have $|P_2| \neq |P_3|$ and $|Q_2| \neq |Q_3|$.
Next, we see that $|P_2| = |Q_2|$ would again imply that
$x_1 = x_2$, and therefore it holds that $|P_2| = |Q_3|$ and $|P_3| = |Q_2|$.
Let $k \in [1, n-1]$ be an integer such that $x_1 = v_k$.
In order to satisfy the condition $|P_2| = |Q_3|$, $x_2$ must be $v_{n-k}$, as required.
\end{proof}

Let $G$ be a monocyclic graph
and $C = (v_0, v_1, \ldots, v_{n-1}, v_0)$ denote the cycle in $G$.
If $G$ admits an automorphism $\xi$
such that $\xi(v_i)=v_{n-i \mod n}$ for each vertex $v_i\in V(C)$,
as in Theorem~\ref{thm:isomorphism}(ii),
then we say that $G$ admits an \emph{axial symmetry}~$\xi$.
Further, for a vertex $y \in V(\langle v_0 \rangle) \setminus \{ v_0 \}$,
let $q$ be the child of $v_0$ such that $G \langle q \rangle$
contains~$y$.
If there exists an automorphism
$\phi$ on $G\langle q \rangle$ such that $\phi(q) = y$, and
an automorphism
$\xi$ on $G- G\langle q \rangle$ such that
$\xi(v_i)=v_{i+k \mod n}$, $v_i\in V(C)$ and $k \geq 1$,
as in Theorem~\ref{thm:isomorphism}(iii),
then we say that 
the pair $(G, q)$
admits a \emph{rotational symmetry}~$(\xi, \phi)$ for $k \geq 1$.

Let $G$ be a monocyclic graph
and $C = (v_0, v_1, \ldots, v_{n-1}, v_0)$ denote the cycle in $G$
such that the pendent tree $G \langle v_0 \rangle$ 
has the maximum number of vertices over all pendent trees in~$G$
and is represented as a left-heavy tree.
Let $\copyv: V(G \langle v_0 \rangle) \to \{0, 1\}$
be a function such that for $v \in V(G \langle v_0 \rangle)$ it holds that
$\copyv(v) = 1$ (resp., $\copyv(v) = 0$) 
if $v$ has a left sibling $u$ 
and for the parent $q = \parent(v) = \parent(u)$
it holds that
$G\langle q, v \rangle \approx G\langle q, u \rangle$
(resp., $v$ does not have a sibling on its left,
or for the sibling $u$ on its left it holds 
$G\langle q, v \rangle \not\approx G\langle q, u \rangle$)~\cite{Suzuki14},
and let $Y = \{ y \in V(G\langle v_0 \rangle) \setminus \{v_0\} \mid
		  \copyv(v) = 0 \text{ for } v \in V(P(v_0, y)) \}$.
We define the  \emph{potential edge set} 
$S(G)$ of non-adjacent vertex pairs in $G$ as follows:
\begin{description}
 \item[Case (i):]
 The pendent tree $G \langle v_0 \rangle$ is not exceeding, 
or $G$ has more than one pendent tree
with at least $|V(G \langle v_0 \rangle)|$ vertices.
Then, we define 
$
 S(G) \triangleq \emptyset.
$
\item[Case (ii):]
The pendent tree $G \langle v_0 \rangle$ is exceeding
and there is no other pendent tree of $G$ 
with at least $|V(G \langle v_0 \rangle)|$ vertices.
\begin{description}
 \item[ Case (ii)(a):] 
 $G$ admits an axial symmetry. 
Then we define 
\[
 S(G) \triangleq \{\{x, y\} \mid x \in \{v_i \mid i \in [\lfloor n/2\rfloor] \}, 
y \in Y\}.
\]
\item[Case (ii)(b):]
$G$ does not admit an axial symmetry.
Then we define 
\begin{align*}
 S(G) \triangleq \{\{x, y\} \mid  y \in Y, &
 \text{ for the child $q$ of $v_0$ such that $y \in V(G \langle q \rangle)$ } \\
 & \hspace{-10mm} ``(G, q) \text{ admits a rotational symmetry $(\xi, \phi)$ with $k \geq 1$'' or} \\
 & \hspace{-10mm} ``(G, q)
      \text{ does not admit a rotational symmetry and}\\
    &   \hspace{4mm}    \text{$\xi : V(C) \to V(C)$ is }
\text{an identity mapping,''}  \\
 &\hspace{-10mm} x \in \{v_i \mid i \in [1, \lfloor n/2\rfloor]\} \\
&  \hspace{-1mm} \cup \{v_{n-i \mod n} \mid i \in [1, \lfloor n/2\rfloor], 
\xi(v_i) \neq v_{n-i \mod n} \}.
\end{align*}
\end{description}
\end{description}
Then, we have the following lemma.

\begin{lemma}
\label{lem:potential_edge_set_proper}
For a monocyclic graph $G$, the potential edge set $S(G)$ is proper.
\end{lemma}
\begin{proof}
Let $C = (v_0, v_1, \ldots, v_{n-1}, v_0)$ denote the 
unique cycle in $G$, such that
the pendent tree $G \langle v_0 \rangle$ has the maximum number
of vertices among all pendent trees in~$G$,
$G \langle v_0 \rangle$ is represented as a left-heavy tree,
and let $\copyv: V(G \langle v_0 \rangle) \to \{0, 1\}$
be a function such that for $v \in V(G \langle v_0 \rangle)$ it holds that
$\copyv(v) = 1$ (resp., $\copyv(v) = 0$) 
if $v$ has a left sibling $u$ 
and $G\langle \parent(v), v \rangle \approx G\langle \parent(u),  u \rangle$
(resp., $v$ does not have a sibling on its left,
or for the sibling $u$ on its left it holds 
$G\langle  \parent(v), v \rangle \not\approx G\langle  \parent(u), u \rangle$)~\cite{Suzuki14}.

\begin{description}
 \item[Case (i).]
 If $G \langle v_0 \rangle$ is not exceeding, 
or $G$ has more than one pendent tree with maximum number of vertices, 
then $S(G) = \emptyset$ is  proper for $G$ since 
by Lemmas~\ref{lem:children:pendent_tree_size}
and~\ref{lem:two_pendent_tree} $G$ has no children.
\item[Case (ii).]
By Theorem~\ref{thm:isomorphism}(i),
for two non-adjacent vertex pairs $\{x_i, y_i\}$, $i = 1, 2$,
in $G$ with $x_1, x_2 \in V(C)$ and
$y_1, y_2 \in V(G \langle v_0 \rangle) \setminus \{v_0\}$
intra-duplication occurs
if $G \langle v_0 \rangle$ admits an automorphism $\xi$
such that $\xi(y_1) = \xi(y_2)$.
By choosing vertices $y \in  V(G \langle v_0 \rangle) \setminus \{v_0\}$
such that $\copyv(v) = 0$ holds for each $v \in V(P(v_0, y))$,
we know that no two vertices $y_1$ and $y_2$ will be chosen
such that $G \langle v_0 \rangle$ admits an
automorphism $\xi$ with $\xi(y_1)=y_2$~\cite{Suzuki14}.
Next we consider the case when 
for vertices $x_1, x_2 \in V(C)$ and $y \in V(G \langle v_0 \rangle) \setminus \{v_0\}$,
and integers $p_i \in [1, \min \{ \res(x_i), \res(y) \} ]$, $i = 1, 2$,
it holds that $G + p_1 \cdot x_1 y$ is isomorphic to $G + p_2 \cdot x_2 y$.
\begin{description}
 \item[Case (ii)(a).]
By Theorem~\ref{thm:isomorphism}(ii),
$G$ admits an automorphism $\xi$
such that for $i \in [ 1, \lfloor n/2 \rfloor]$
it holds $\xi(v_i) = v_{n-i \mod n}$,
and therefore it suffices to consider
vertices 
$v_{i}, i \in [1, \lfloor n/2 \rfloor]$,
for the choice of~$x$,
thereby for each 
$x' \in \{v_i \mid i \in [\lfloor n/2 \rfloor + 1, n-1]\}$
there exists an $x$ such that $\xi(x) = x'$ and therefore
$\res(x) = \res(x')$, and for $p \in [1, \min \{ \res(x), \res(y) \} ]$
$G + p \cdot xy$ and $G + p\cdot x'y$ are isomorphic.
On the other hand, by Lemma~\ref{lem:isomorphism-len}
for $x_1, x_2 \in  \{v_i \mid i \in [1, \lfloor n/2 \rfloor] \}$
and $p_i \in [1, \min \{ \res(x_i), \res(y) \} ]$, $i = 1, 2$,
$G + p_1 \cdot x_1 y$ and $G + p_2 \cdot x_2 y$ are not isomorphic, satisfying
the conditions for a proper set.
\item[Case (ii)(b).]
In case $G$ does not admit an axial symmetry,
for each choice of $y \in V(G \langle v_0 \rangle) \setminus \{v_0\}$
such that $\copyv(v) = 0$ holds for all $v \in V(P(v_0, y))$,
for the child $q$ of $v_0$ such that $G \langle q \rangle$
contains $y$, we check whether 
the pair $(G, q)$
admits a rotational symmetry $(\xi, \phi)$
for $k \geq 1$, as in Theorem~\ref{thm:isomorphism}(iii).
In case there does not exist an automorphism $\xi$ on 
$G - G \langle q \rangle$ such that
$\xi(v_i)=v_{i+k \mod n}$,
$v_i\in V(C)$ and $k \geq 1$,
and an automorphism $\phi$ on 
$G \langle q \rangle$ with $\phi(q) = y$,
we take trivial automorphism $\xi(x) = x, x \in V(C)$.
Again, by the automorphism $\xi$,
for two vertices $x_1, x_2 \in V(C)$
either $x_1 = \xi(x_2)$ and then
$\res(x_1) = \res(x_2)$,
$\{x_1, y\} \in S(G)$ and
$\{x_2, y\} \notin S(G)$
but for any $p \in [1, \min \{ \res(x_1), \res(y) \}]$
it holds that $G + p \cdot x_1 y \approx G + p \cdot x_2 y$,
or $x_1 \neq \xi(x_2)$, in which case 
$\{x_1, y\} \in S$ and
$\{x_2, y\} \in S$, 
but for any 
$p_i \in [1, \min \{ \res(x_i), \res(y) \} ]$, $i = 1, 2$,
it holds that 
$G + p_1 \cdot x_1 y \not\approx G + p_2 \cdot x_2 y$,
as required.
\end{description}
\end{description}
\end{proof}

We give a description of an algorithm to compute the potential edge set
of a given monocyclic graph~$G$ as  Procedure~{5}  {\sc GeneratePotentialEdgeSet}.

\bigskip
\noindent {\bf Procedure~{5}}
{\sc GeneratePotentialEdgeSet}$(G)$
\begin{algorithmic}[1]
\Require A monocyclic graph $G$ with a cycle $C = (v_0, v_1, \ldots, v_{n-1}, v_0)$,
 such that $G \langle v_0 \rangle$ has the maximum number of vertices 
 over all pendent trees 
 in $G$,
 and a function $\copyv : V(G \langle v_0 \rangle) \to \{0, 1\}$
 such that  
$\copyv(v) = 1$ 
if $v$ has a left sibling $u$ 
and $G\langle \parent(v), v \rangle \approx G\langle \parent(u), u \rangle$
and $\copyv(v) = 0$ otherwise.
\Ensure The potential edge set $S(G)$ of $G$.
\State{$S := \emptyset$};
\If {$|V(G \langle v_0 \rangle)| \geq |V(G)|/3$ 
    and for $i \in [1, n-1]$, $|V(G \langle v_0 \rangle)|  > |V(G \langle v_i \rangle)|$  }
  \For{{\bf each} child $c$ of $v_0$ such that $\copyv(c) = 0$}
    \For{{\bf each} $y \in V(G\langle c \rangle)$ such that $\copyv(v) = 0$ for $v \in V(P(v_0, y))$}
	  \State{$S := S \cup \left \{\{v_i, y\} \mid i \in [1, \lfloor n/2 \rfloor] \right \}$}
      \If{$G$ does not admit an axial symmetry}
      \State {Let $q \in V(G\langle v_0 \rangle)$ be the child of 
			$v_0$ such that $y \in V(G\langle q \rangle))$};      
	  \If{$(G, q)$ admits a rotational symmetry $(\xi, \phi)$ for $k \geq 1$}
		 \State{$S := S \cup \{ \{v_{n-i \mod n}, y\} \mid i \in [1, \lfloor n/2 \rfloor],
											      \xi(v_i) \neq v_{n-i \mod n}\}$}
	    \Else ~/* $(G, q)$ does not admit a rotational symmetry for $k \geq 1$ */
		\State{$S := S \cup \left \{\{v_i, y\} \mid i \in [\lfloor n/2 \rfloor + 1, n-1] \right \}$}
	    \EndIf
	\EndIf
   \EndFor
  \EndFor
\EndIf;
\State{{\bf output} $S$ as $S(G)$.}
\end{algorithmic}

 \section{Experimental results}
 \label{sec:experiments}
 
 To test the effectiveness of our algorithm for enumerating mono-block 2-augmented trees,
 we have implemented it 
 and performed computational comparison with MOLGEN~\cite{MOLGEN5},
 a generator for chemical graphs.

 In particular, we did experiments for two different types of instances, 
 named EULF-$L$-A and  EULF-$L$-P, by considering
 a set $\pathset$ of colored sequences with length at most a given integer~$N$,
 given lower and upper bounds, 
 $\mathbf{g}_{a}: \pathset \to \mathbb{Z}_+$ and
 $\mathbf{g}_{b}: \pathset \to \mathbb{Z}_+$, respectively,
 on the path frequencies of the paths in $\pathset$, and integers $L$ and~$d$.
 For a given set $\pathset$ of colored sequences and a graph $G$,
 let $\fvset_{\pathset}(G): \pathset \to \mathbb{Z}_+$ 
 denote the number $\freq(t, G)$ of rooted paths $P \subseteq G$
 such that $\gamma(P) = t \in \pathset$.
 Assuming that 
 $\mathbf{g}_{a} \leq \mathbf{g}_b$,
 and in particular, that
 $\mathbf{g}_{a}[t] = \mathbf{g}_b[t]$ 
is satisfied for each colored sequence $t \in \pathset \cap \Sigma^{0, d}$,
each of the instance types EULF-$L$-A and  EULF-$L$-P asks 
to enumerate chemical graphs $G$ such that 
$\mathbf{g}_{a} \leq \freq_{\pathset}(G) \leq \mathbf{g}_{b}$,
and for any $P \subseteq G$ such that $\gamma(P) \notin \pathset$,
it holds that $|P| > L$ and $|P| \leq L$, for
 instance types EULF-$L$-A and  EULF-$L$-P, respectively.

We have chosen six compounds 
 from the PubChem database 
 which when represented
 as hydrogen-suppressed chemical graphs have
 mono-block 2-augmented tree structure,  and constructed
 feature vectors based on the path frequencies
 of the paths  in the chemical graphs.
 All compounds have~13 non-hydrogen atoms,
 maximum path length~11,
 and maximum bond multiplicity $d \in \{2, 3\}$.
 All compounds include the three chemical
 elements {\tt C} (carbon), {\tt O} (oxygen), and {\tt N} (nitrogen).
 The information on the chosen compounds, 
 identified by their Compound ID (CID) number
 in the PubChem database is given in Table~\ref{table:6mols}.
\begin{table}[!hb]
\centering
 \caption{Information on the six compounds chosen from 
 the PubChem database for our experiments}
 \label{table:6mols}
   \begin{tabular}{@{} l l l @{} }
  \toprule
    \parbox{20mm}{Molecular formula }
	    & $d$ & CID \vspace{2mm} \\
   \toprule
   \multirow{2}{20mm}{\tt C$_9$N$_1$O$_3$} 
						    & 2 &  301729 	 	 \\ 
						    & 3 &  57320502 		 \\  \midrule
  \multirow{2}{20mm} {\tt C$_9$N$_2$O$_2$} 
 						     & 2 &  6163405 		 \\
						    & 3 &  131335510 	\\ \midrule
  \multirow{2}{20mm} {\tt C$_9$N$_3$O$_1$} 
  						    & 2 &  9942278 		 \\
						    & 3 &  10103630  		\\  \bottomrule
  \end{tabular}
\end{table}

We construct 
instances of types EULF-$L$-A and EULF-$L$-P 
for different values of parameter $L$ in the following way.
We take a set $\Sigma$ of colors to be $\Sigma = \{ \mathtt{C}, \mathtt{O}, \mathtt{N} \}$,
such that $\val(\mathtt{C}) = 4$, $\val(\mathtt{O}) = 2$, and $\val(\mathtt{N}) = 3$.
For each hydrogen suppressed chemical graph $G$
that corresponds to a chemical compound in Table~\ref{table:6mols},
we take $d \in \{2, 3\}$ to be the maximum bond multiplicity in the chemical graph,
and for some choice of values for $N \geq 0$ we construct
a set of colored sequences $\pathset \subseteq \Sigma^{\leq N, d}$
that consists of all colored sequences $t$ with length $|t| \leq N$
such that $G$ contains a rooted path $P$ with $\gamma(P) = t$. 
Finally, for an integer $s \in [0, 2]$ we set lower and upper bounds,
${\bf g}_a$ and ${\bf g}_b$
on feature vectors
as follows: 
for $t \in \pathset$, if $|t| \in \Sigma^{0, d}$ then ${\bf g}_a[t] = {\bf g}_b[t] = \freq(t, G)$,
otherwise ${\bf g}_{a}[t] = \max \{0, \freq(t, G) - s \}$
and  ${\bf g}_{b}[t] = \freq(t, G) + s$.
The parameter $s$ effectively serves to ``relax''
the path frequency specification.

On the other hand, we used MOLGEN~\cite{MOLGEN5}
without aromaticity detection
by specifying the hydrogen suppressed formula, 
number of cycles to be two in enumerated structures
- thereby enumerating chemical graphs with 2-augmented tree structures
with a maximum allowed bond multiplicity.
Note that there is no option in MOLGEN to specify
whether the enumerated structures have a mono-block 
structure or not.

 We implemented our algorithm in the 
 C++ programming language, and compiled and executed on
the Linux~14.04.6 operating system by the gcc compiler  version~4.8.4 and optimization level~O3.
All experiments were done on a PC with Intel Xeon CPU E5-1660 v3 
running at 3.00~GHz, with 32~GB memory.

 \subsection{Experimental Results for EULF-$L$-A}
 \label{sec:experiments_EULF1}

To test the behavior of our algorithm for 
instance types EULF-$L$-A, especially the effect
the choice of problem parameters have on the running time
and the number of enumerated chemical graphs,
we choose values for parameter $N \in [2, 6]$,
and 
we took values for the parameter $L \in \{2, \lceil {N/2} \rceil, N\}$.
 
The results from our experiments 
for  EULF-$L$-A are summarized in 
Figs.~\ref{fig:result_graphs_1} to~\ref{fig:result_graphs_6}.
We observe that our algorithm has a clear advantage
when we are given a path frequency specification
for instances of type EULF-$L$-A
over using MOLGEN to generate molecules with a specified formula.
 We also observe some trends over the values of 
 the parameters $N$, $L$, and~$s$.
 Namely, the number of generated molecules, as well as the time it takes our algorithm,
 reduces as the length $N$ of the longest path given in the 
 set of paths, as well as the parameter~$L$ increase, 
 but grows with an increasing value~$s$
 that we choose to relax the path frequency specification.

 \begin{figure}[!ht]
  \begin{minipage}{0.45\textwidth}
   \centering
   \includegraphics[width=1.1\textwidth]{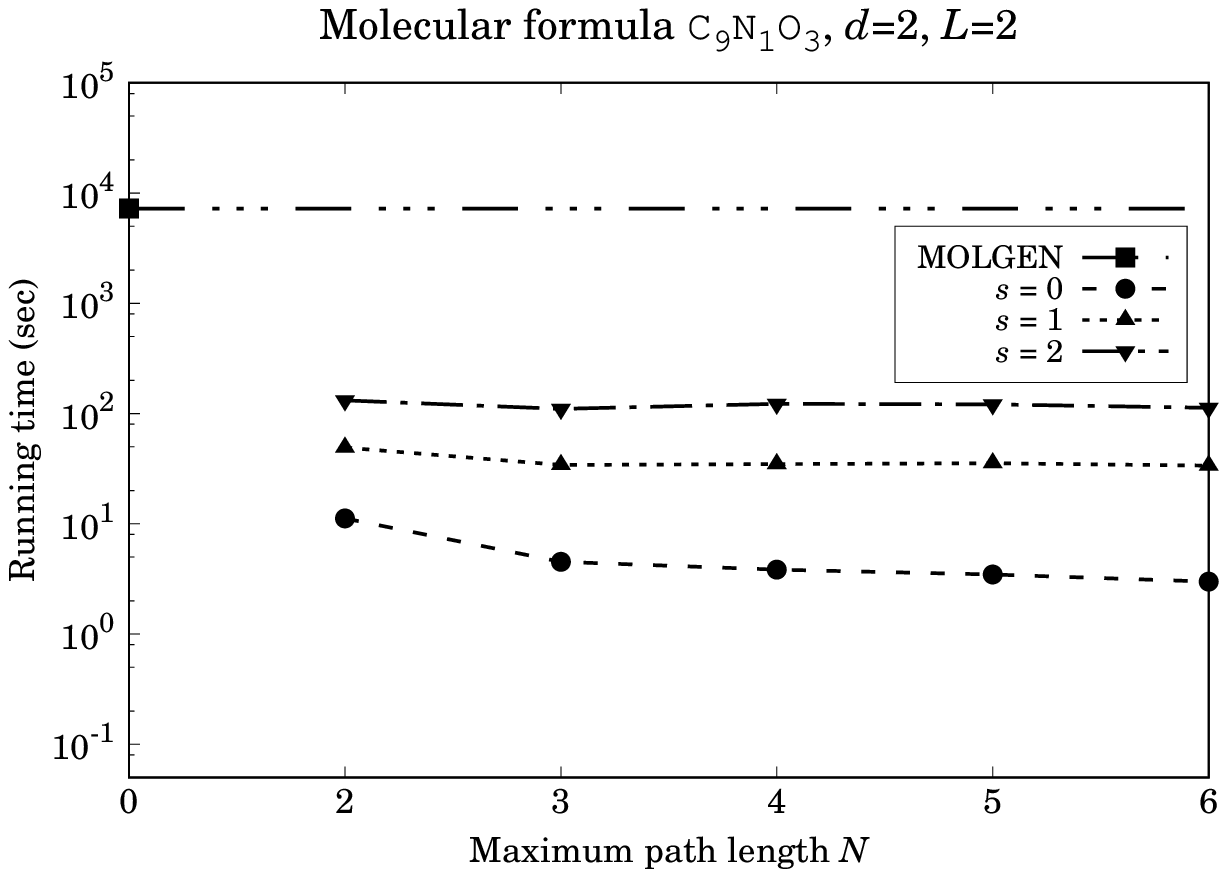}\\
   {\footnotesize (a)}\\
  \end{minipage}
\hfill
  \begin{minipage}{0.45\textwidth}
   \centering
   \includegraphics[width=1.1\textwidth]{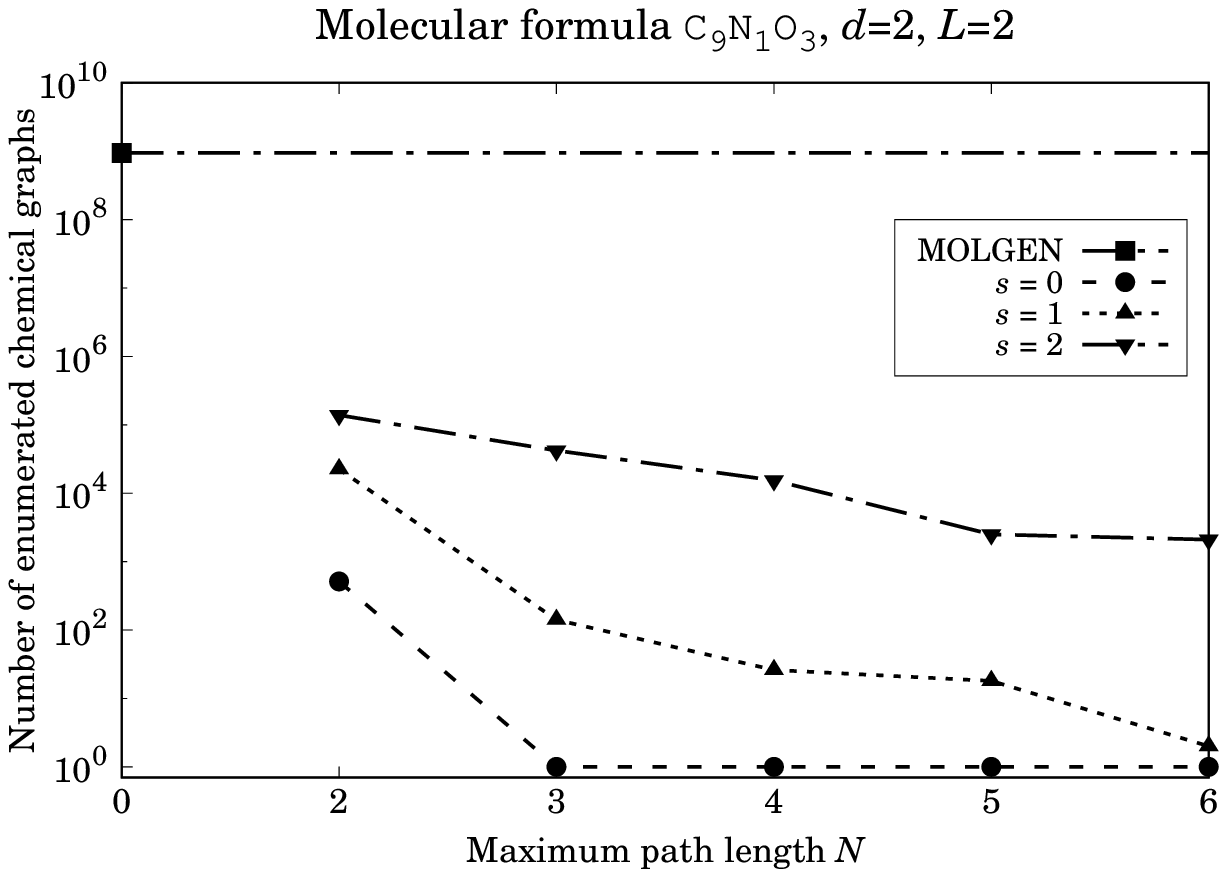}\\
   {\footnotesize (d)}\\
  \end{minipage} 
  \medskip

  \begin{minipage}{0.45\textwidth}
   \centering
      \includegraphics[width=1.1\textwidth]{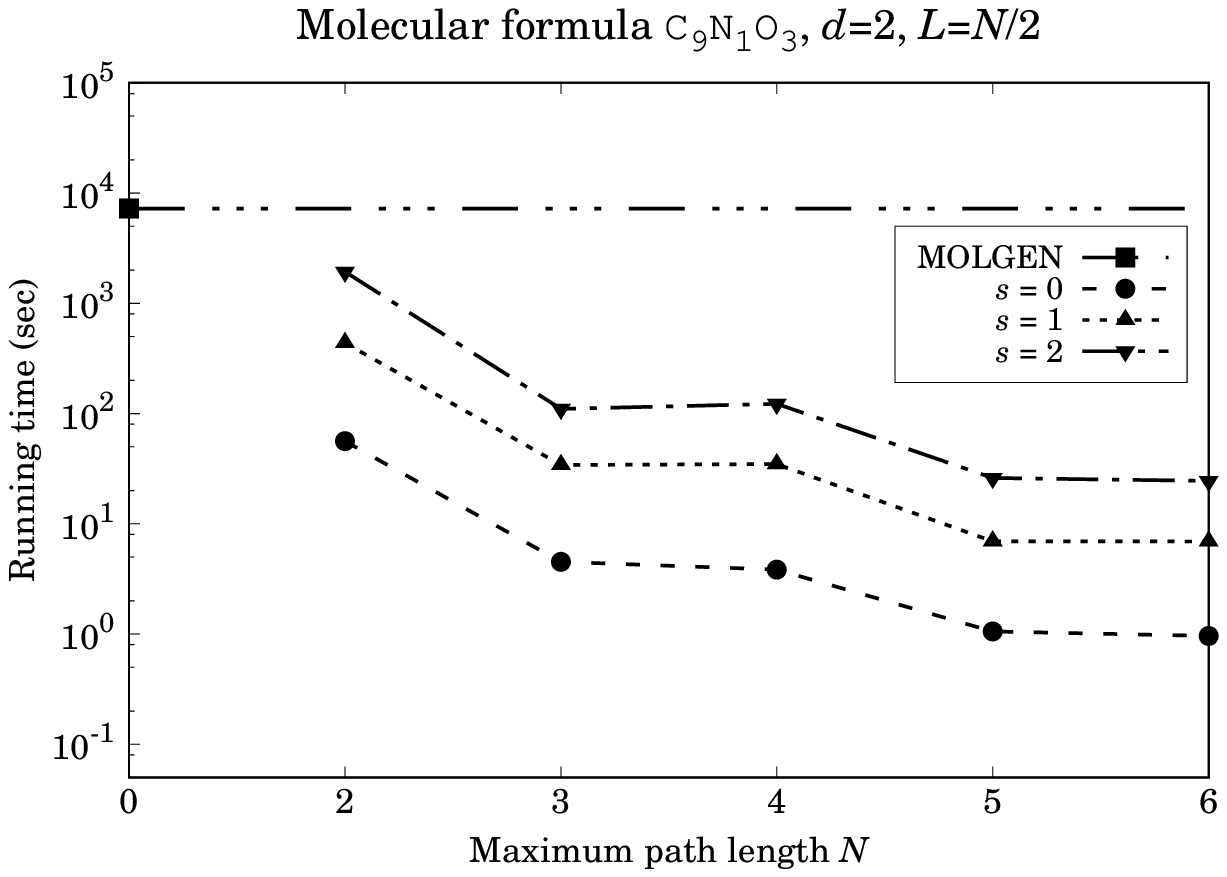}\\
      {\footnotesize (b)}\\
  \end{minipage} 
\hfill
  \begin{minipage}{0.45\textwidth}
   \centering
    \includegraphics[width=1.1\textwidth]{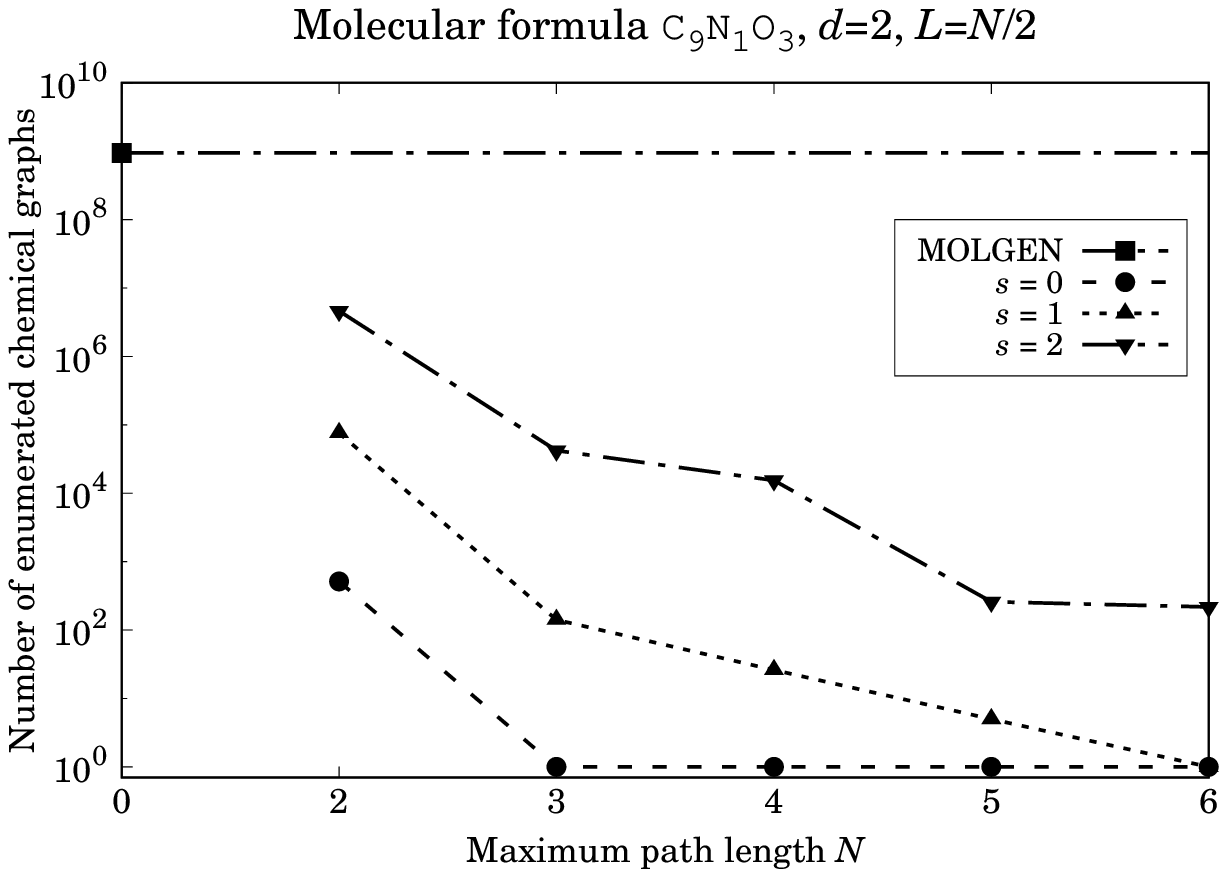}\\
    {\footnotesize (e)}\\
  \end{minipage} 
  \medskip

  \begin{minipage}{0.45\textwidth}
   \centering
      \includegraphics[width=1.1\textwidth]{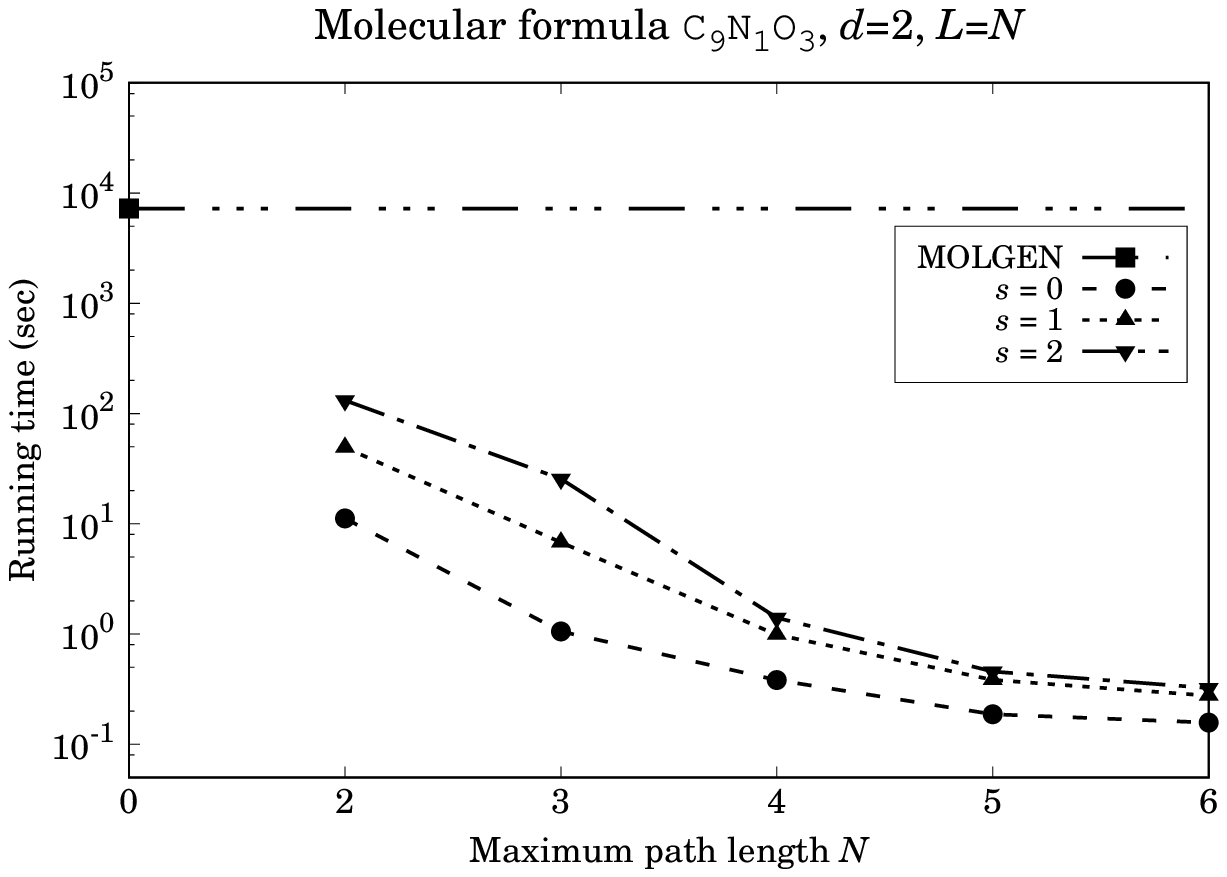}\\
      {\footnotesize (c)}\\
  \end{minipage} 
\hfill
  \begin{minipage}{0.45\textwidth}
   \centering
    \includegraphics[width=1.1\textwidth]{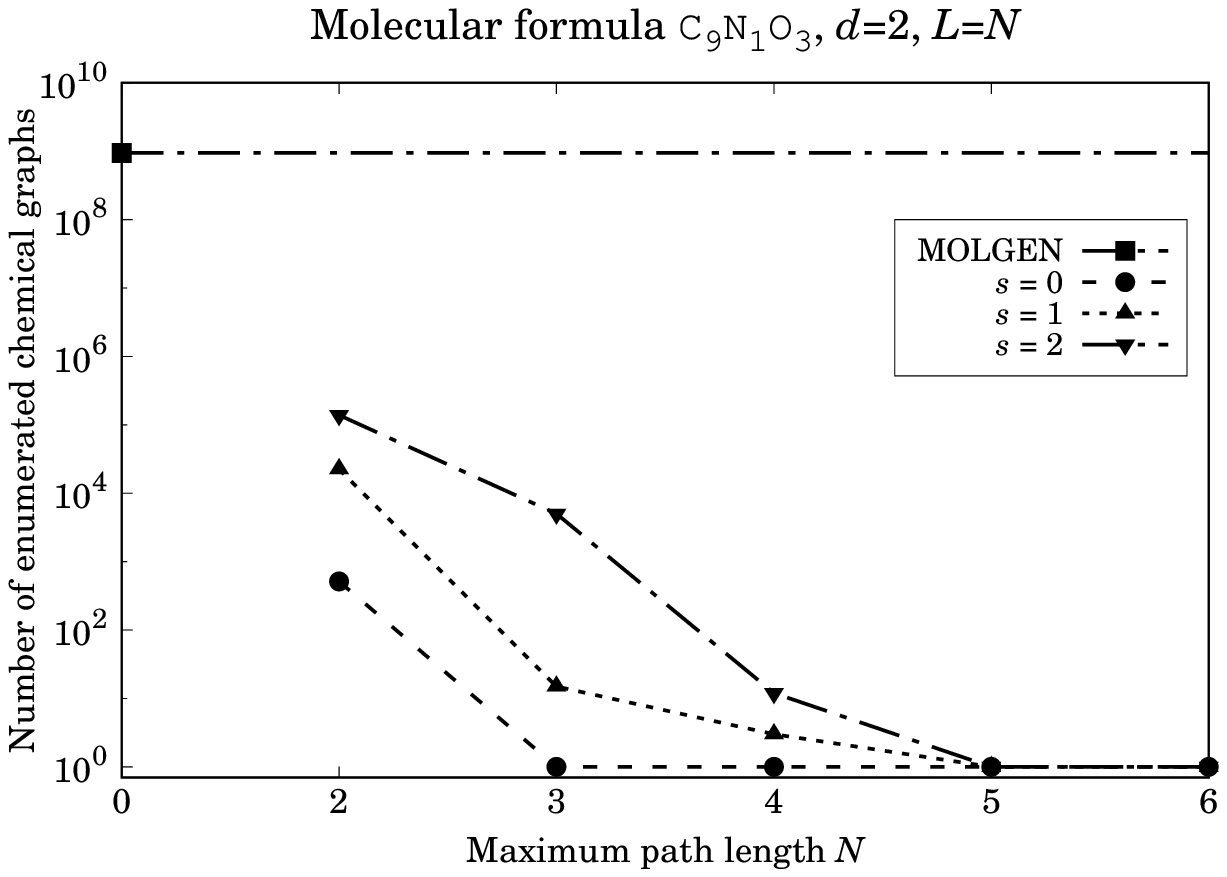}\\
    {\footnotesize (f)}\\
  \end{minipage} 
  \vspace{1cm}
  
  \caption{
    Plots showing the computation time 
    and number of chemical graphs enumerated by our algorithm
    for instance type EULF-$L$-A, as compared to MOLGEN.
    The sample structure from PubChem is with CID 301729,
    molecular formula {\tt C$_9$N$_1$O$_3$},
    and maximum bond multiplicity~$d=2$.
    (a)-(c)~Running time;
    (d)-(f)~Number of enumerated chemical graphs.
  }
 \label{fig:result_graphs_1}
 \end{figure}

  \begin{figure}[!ht]
  \begin{minipage}{0.45\textwidth}
   \centering
   \includegraphics[width=1.1\textwidth]{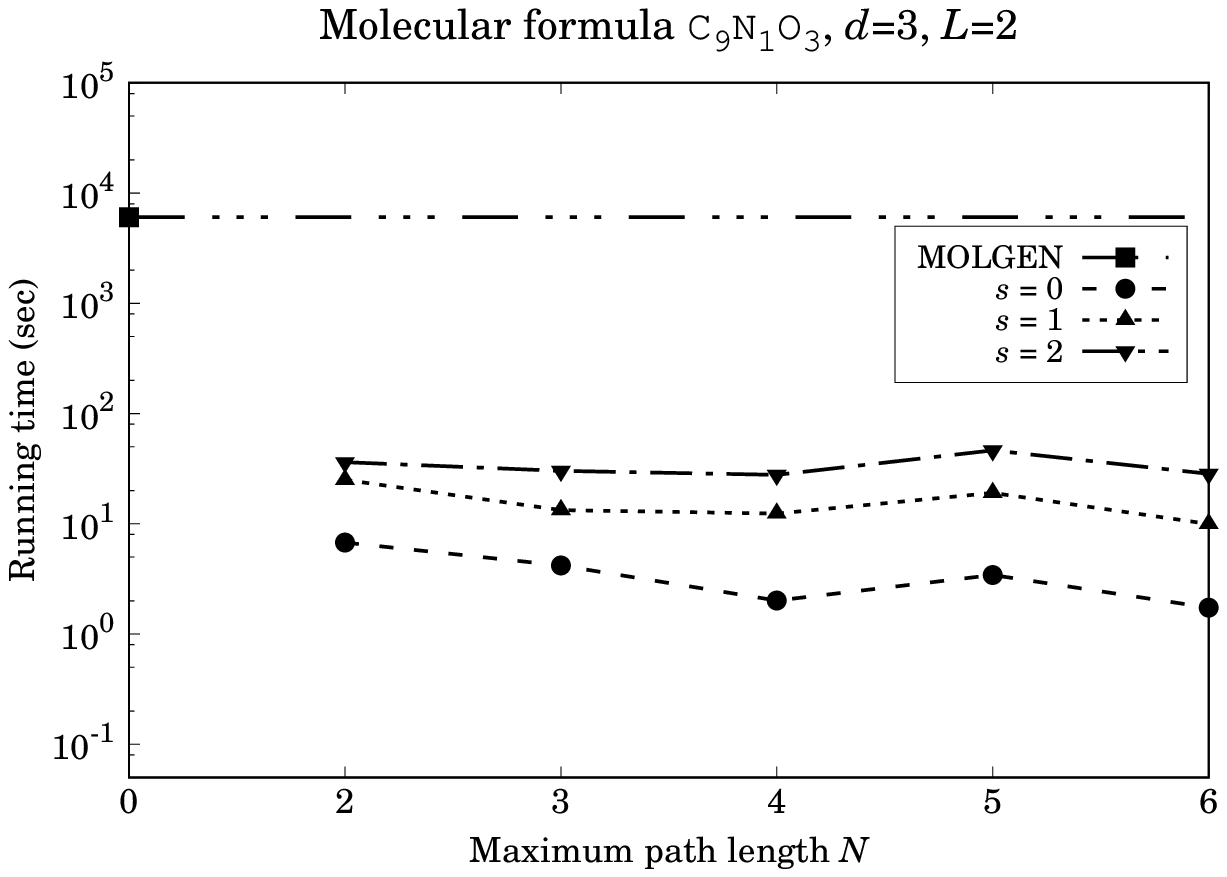}\\
   {\footnotesize (a)}\\
  \end{minipage}
\hfill
  \begin{minipage}{0.45\textwidth}
   \centering
   \includegraphics[width=1.1\textwidth]{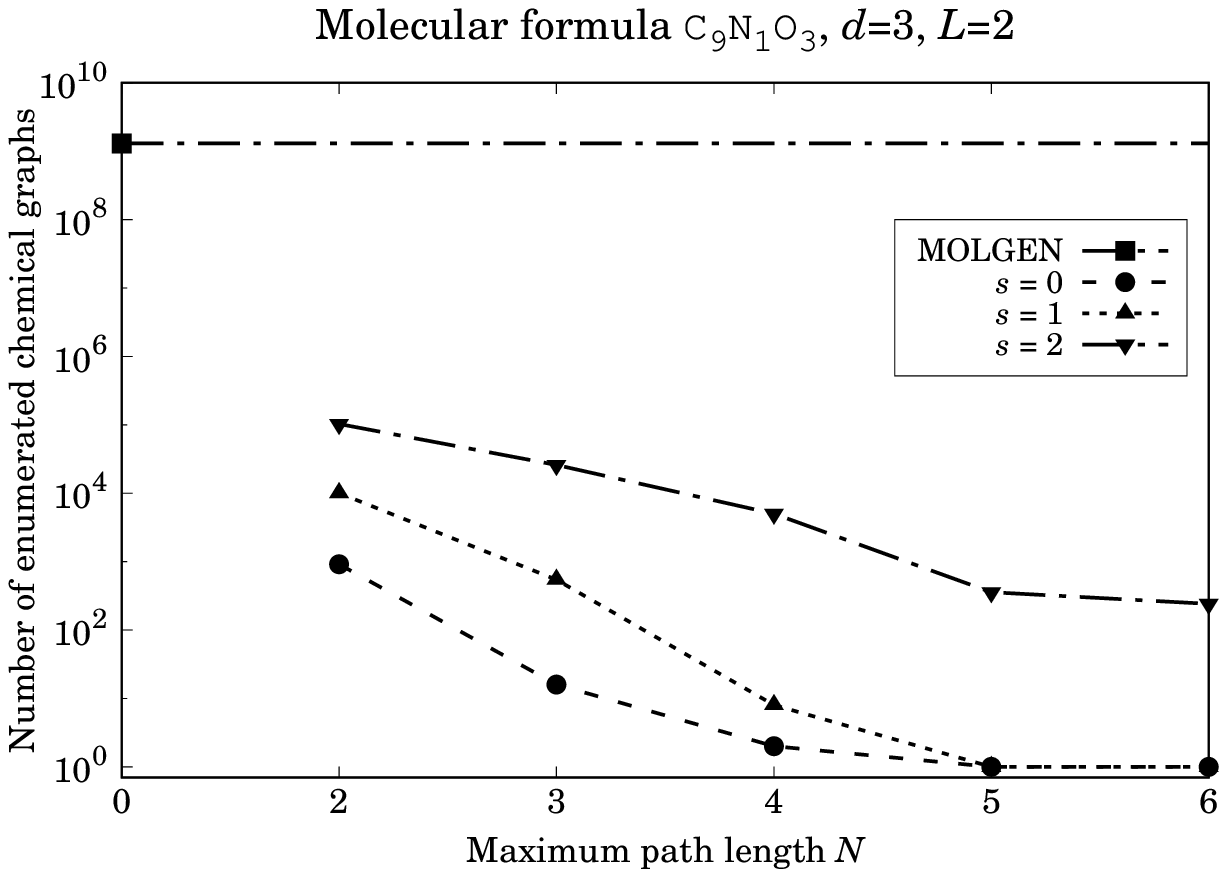}\\
   {\footnotesize (d)}\\
  \end{minipage} 
  \medskip

  \begin{minipage}{0.45\textwidth}
   \centering
      \includegraphics[width=1.1\textwidth]{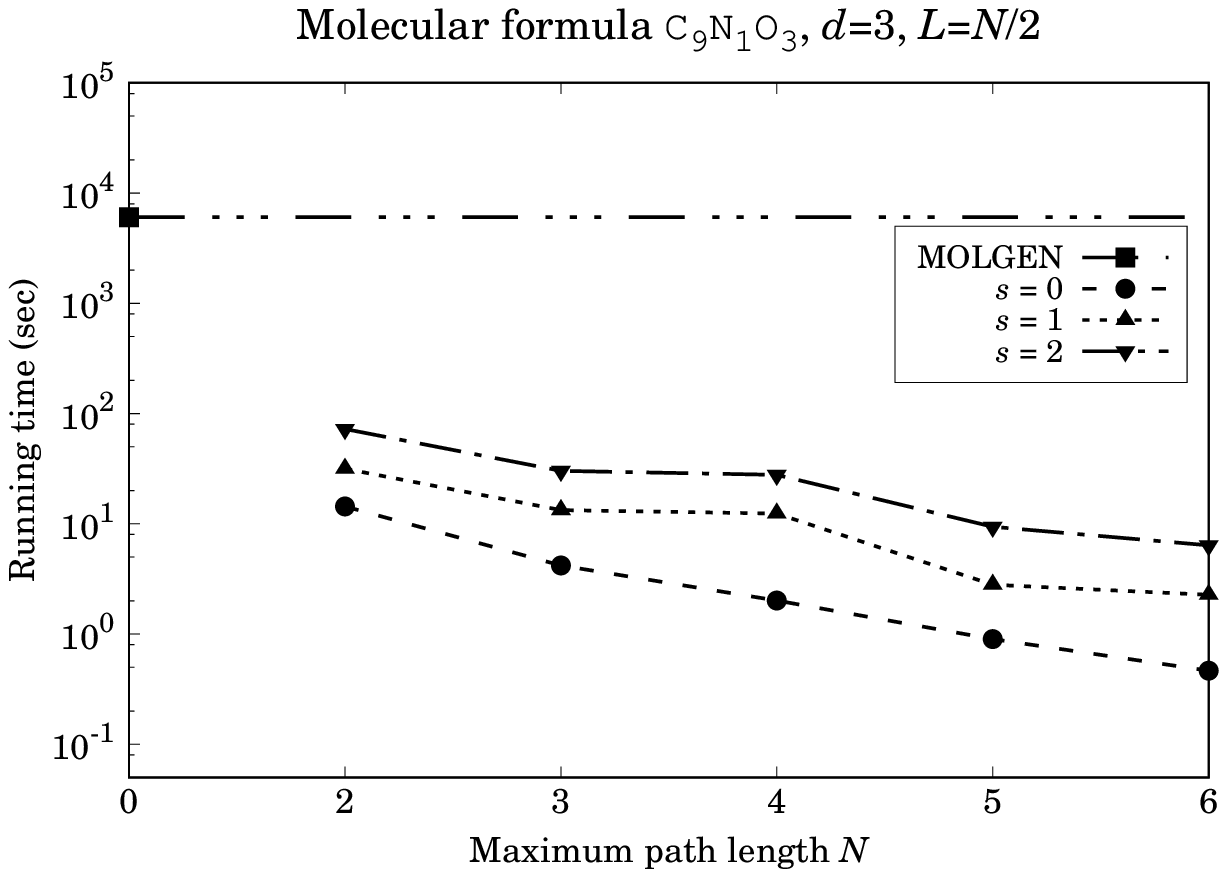}\\
      {\footnotesize (b)}\\
  \end{minipage} 
\hfill
  \begin{minipage}{0.45\textwidth}
   \centering
    \includegraphics[width=1.1\textwidth]{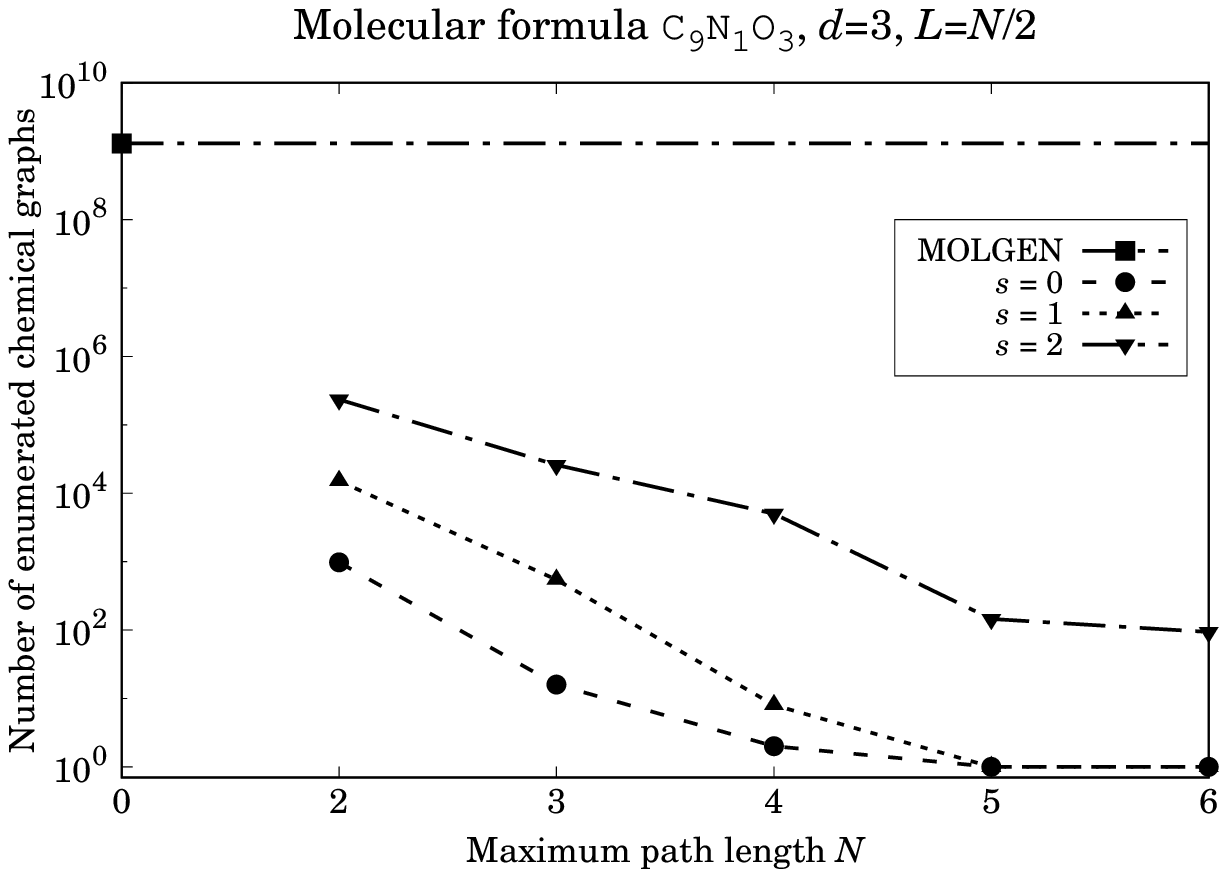}\\
    {\footnotesize (e)}\\
  \end{minipage} 
  \medskip

  \begin{minipage}{0.45\textwidth}
   \centering
      \includegraphics[width=1.1\textwidth]{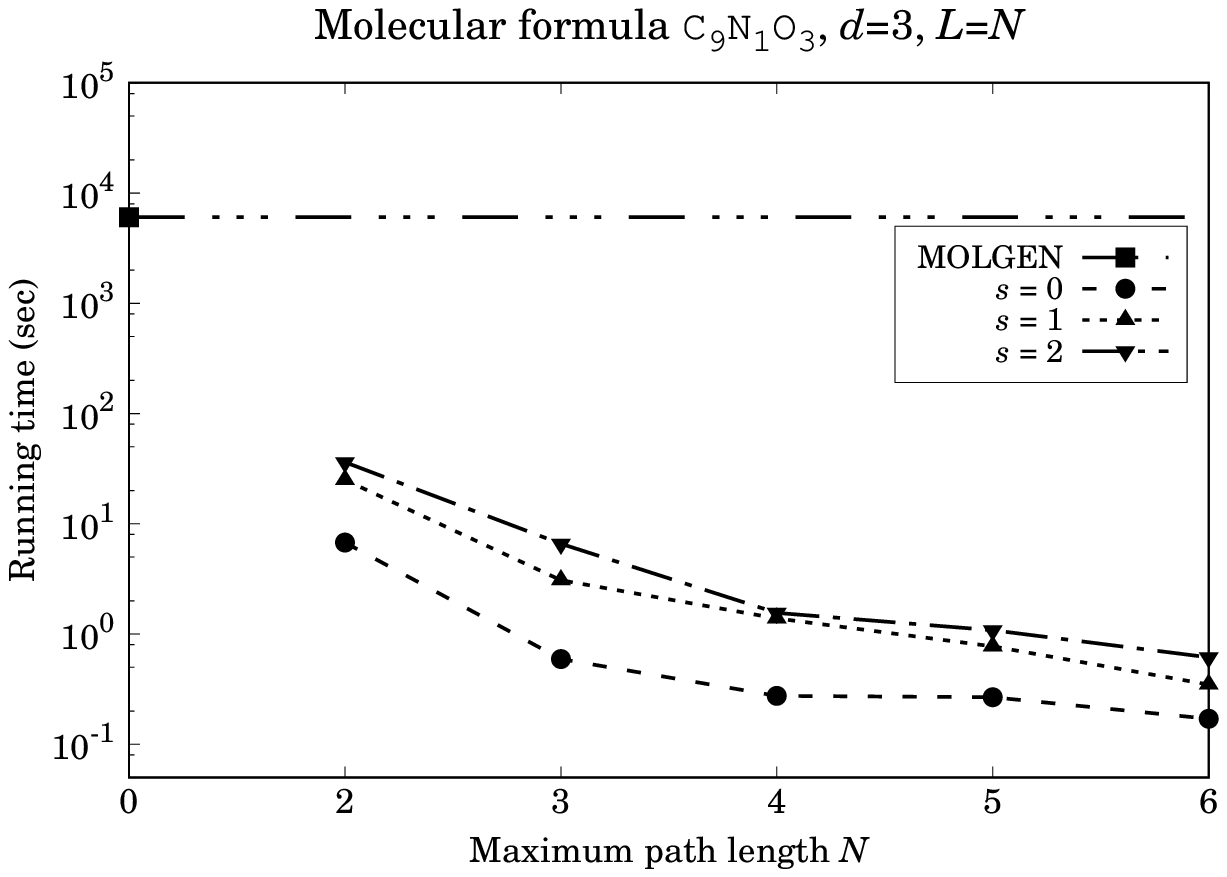}\\
      {\footnotesize (c)}\\
  \end{minipage} 
\hfill
  \begin{minipage}{0.45\textwidth}
   \centering
    \includegraphics[width=1.1\textwidth]{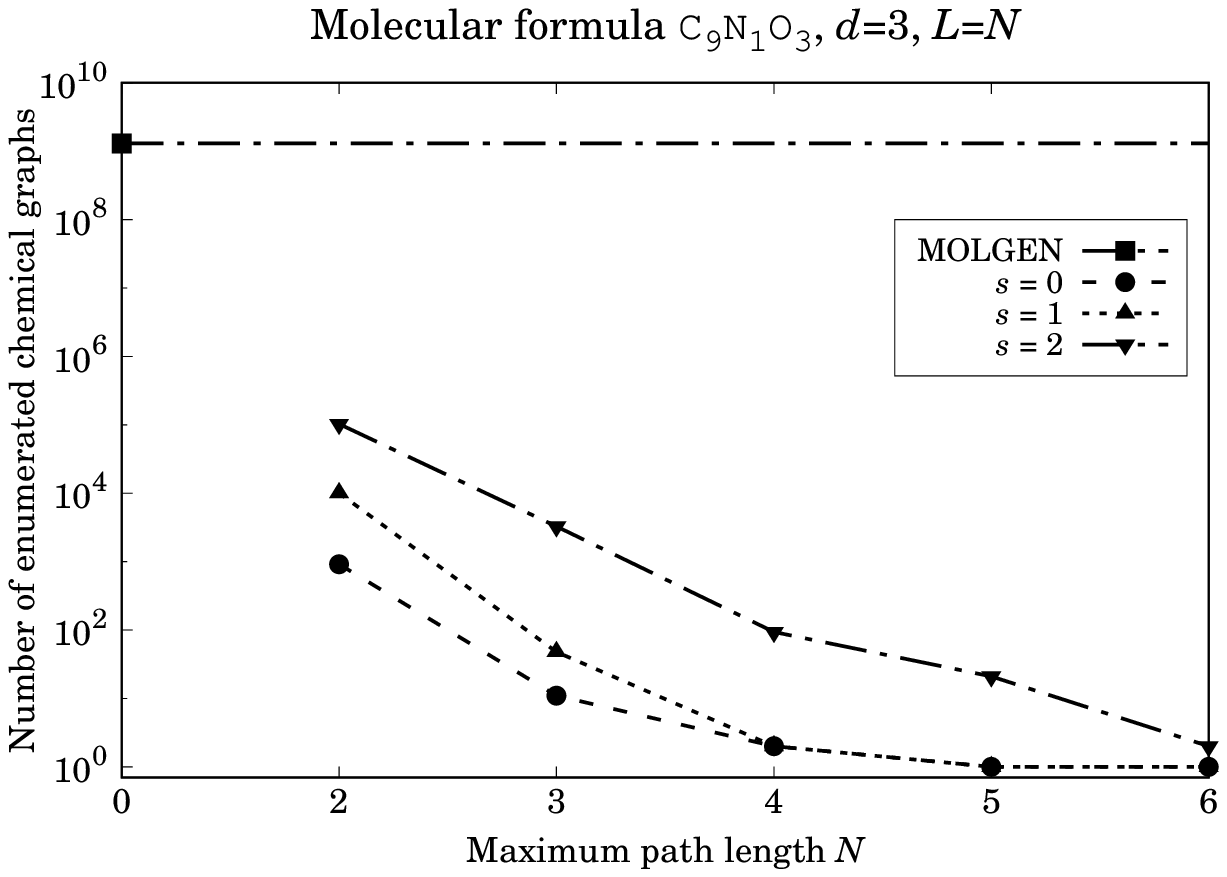}\\
    {\footnotesize (f)}\\
  \end{minipage} 
  \vspace{1cm}
  
  \caption{
    Plots showing the computation time 
    and number of chemical graphs enumerated by our algorithm
    for instance type EULF-$L$-A, as compared to MOLGEN.
    The sample structure from PubChem is with CID~57320502,
    molecular formula {\tt C$_9$N$_1$O$_3$},
    and maximum bond multiplicity~$d=3$.
    (a)-(c)~Running time;
    (d)-(f)~Number of enumerated chemical graphs.
  }
 \label{fig:result_graphs_2}
 \end{figure}

  \begin{figure}[!ht]
  \begin{minipage}{0.45\textwidth}
   \centering
   \includegraphics[width=1.1\textwidth]{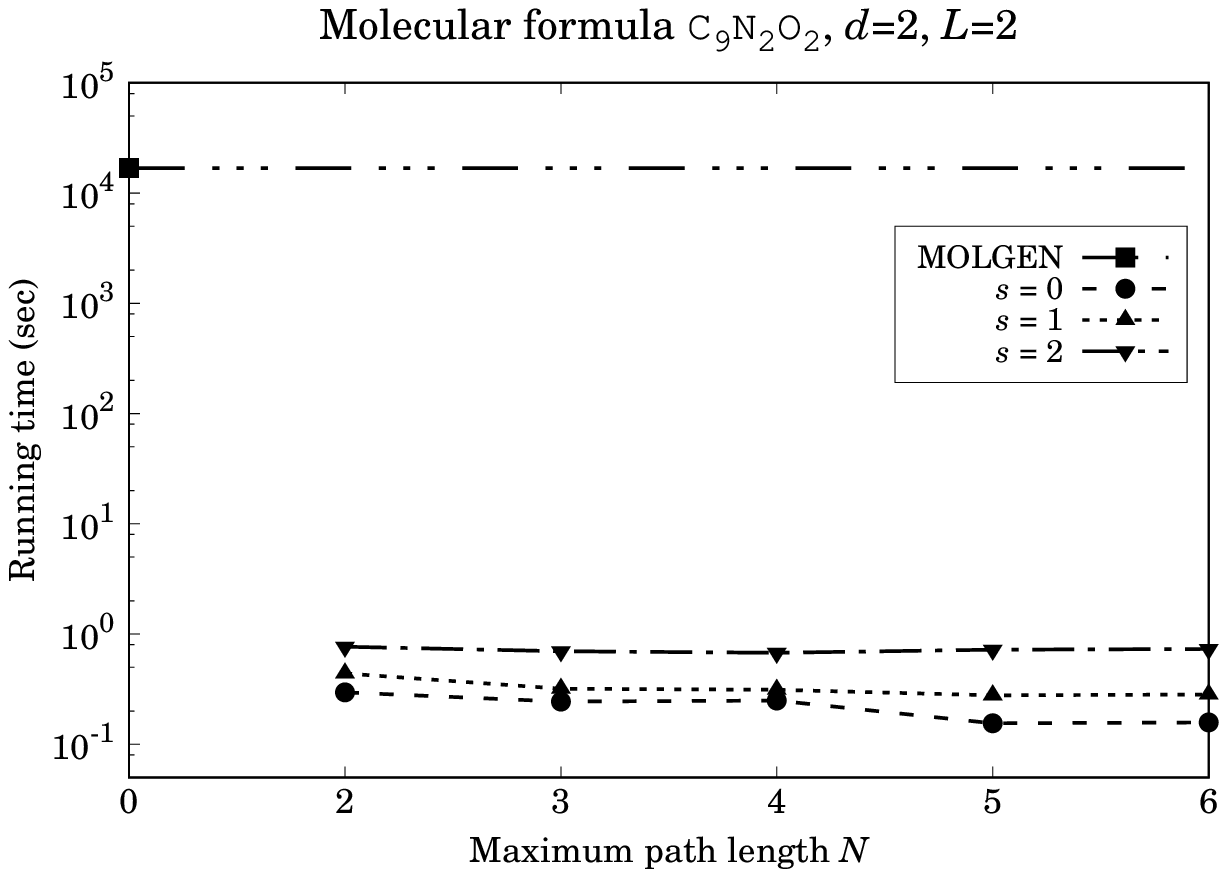}\\
   {\footnotesize (a)}\\
  \end{minipage}
\hfill
  \begin{minipage}{0.45\textwidth}
   \centering
   \includegraphics[width=1.1\textwidth]{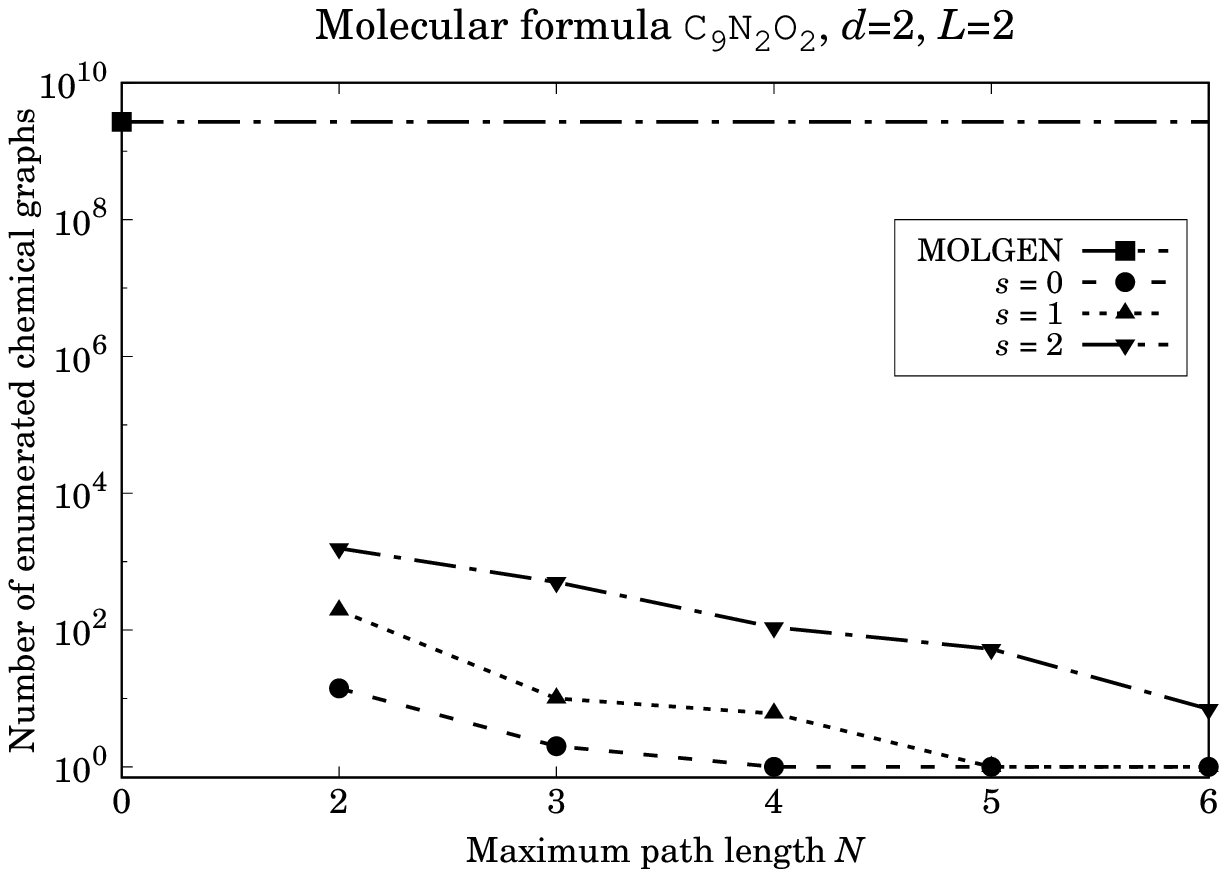}\\
   {\footnotesize (d)}\\
  \end{minipage} 
  \medskip

  \begin{minipage}{0.45\textwidth}
   \centering
      \includegraphics[width=1.1\textwidth]{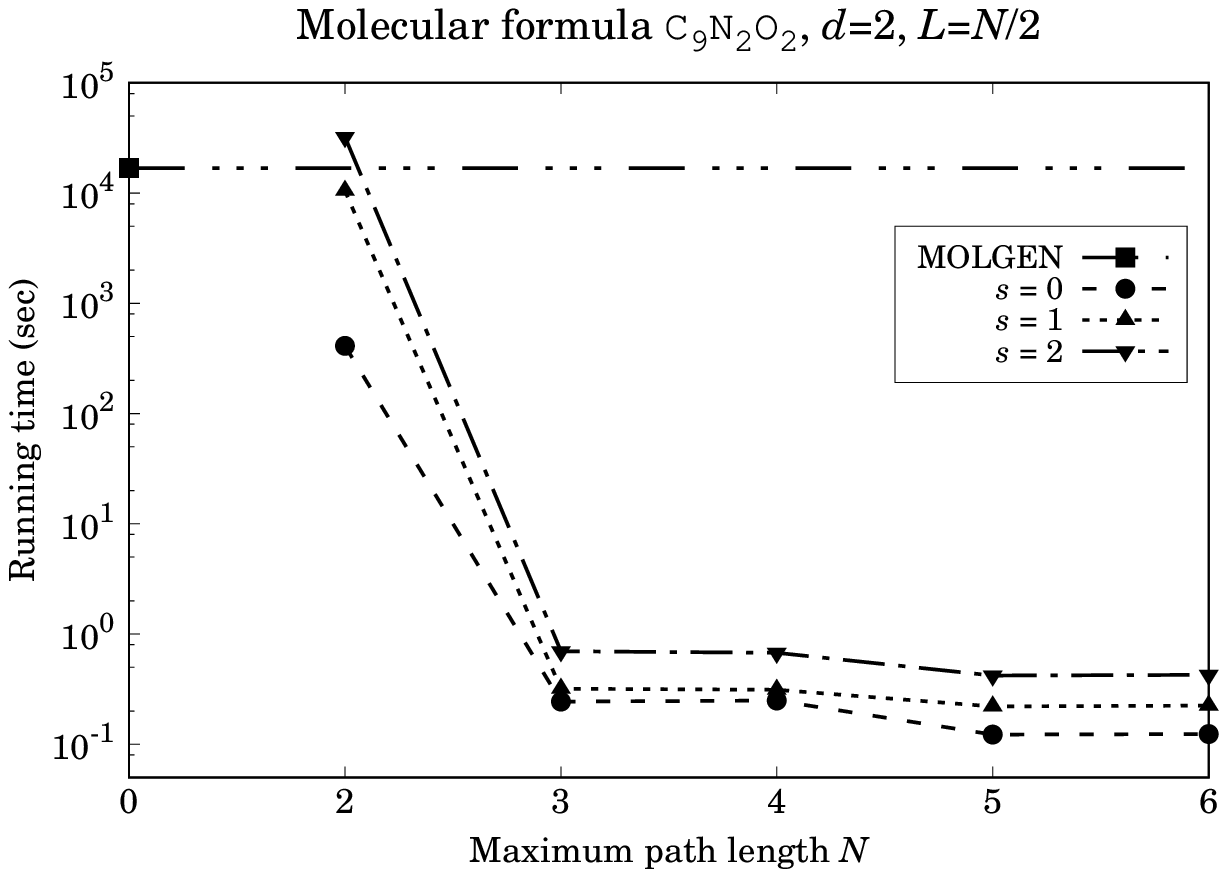}\\
      {\footnotesize (b)}\\
  \end{minipage} 
\hfill
  \begin{minipage}{0.45\textwidth}
   \centering
    \includegraphics[width=1.1\textwidth]{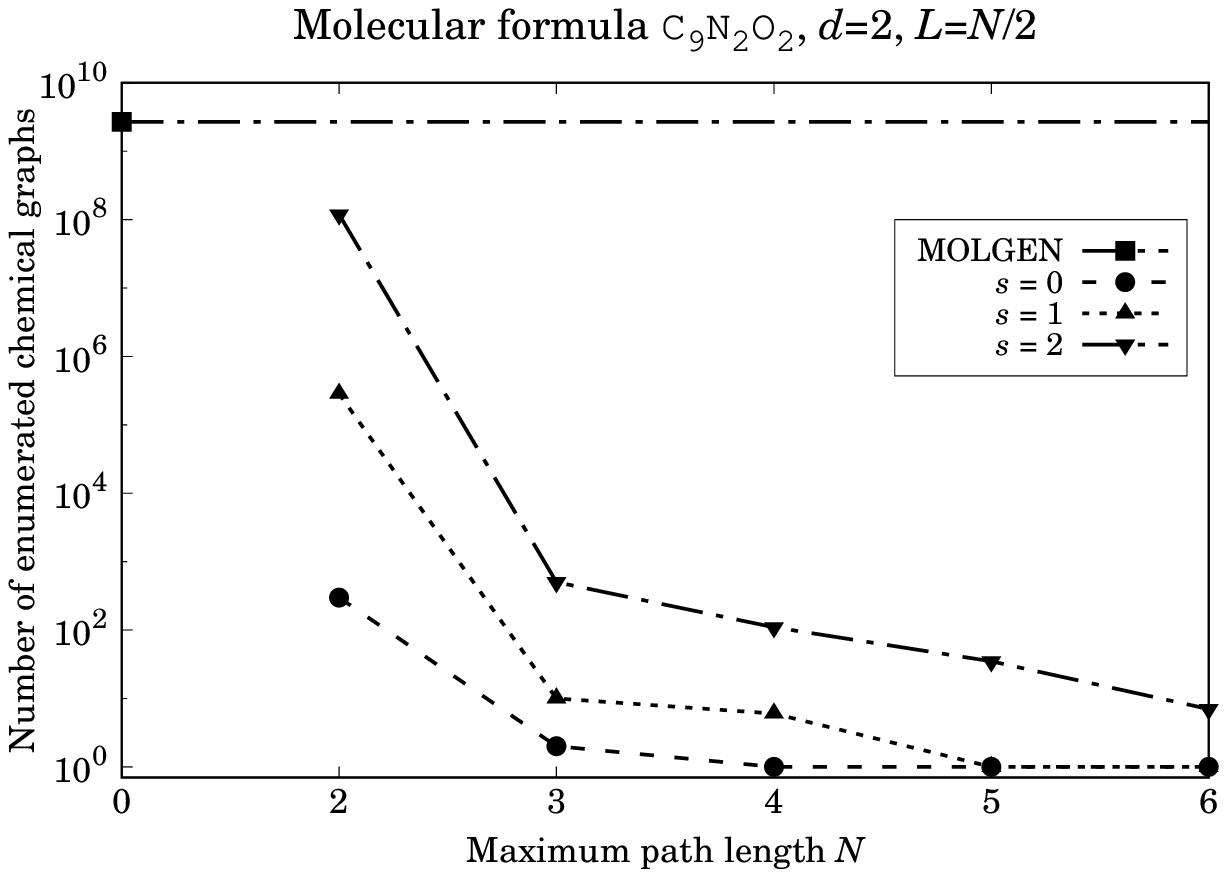}\\
    {\footnotesize (e)}\\
  \end{minipage} 
  \medskip

  \begin{minipage}{0.45\textwidth}
   \centering
      \includegraphics[width=1.1\textwidth]{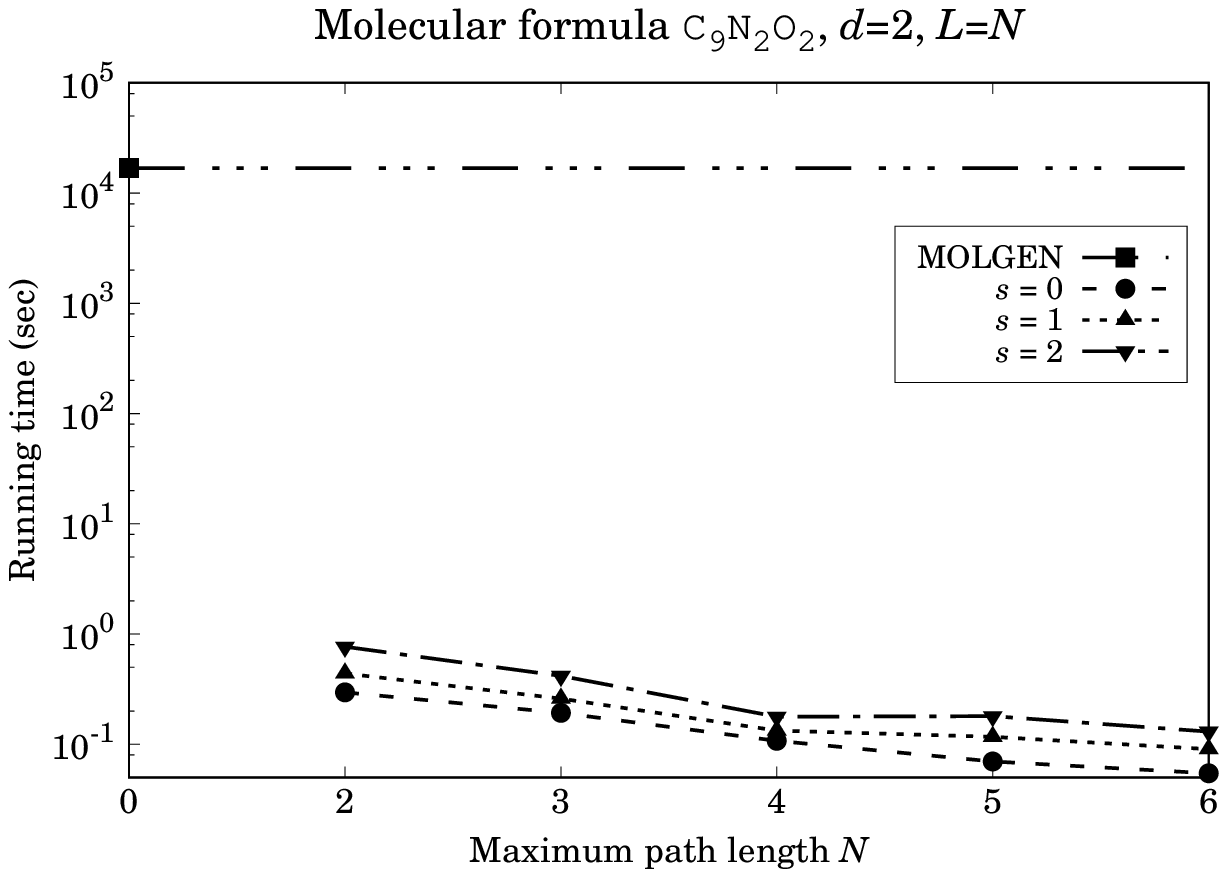}\\
      {\footnotesize (c)}\\
  \end{minipage} 
\hfill
  \begin{minipage}{0.45\textwidth}
   \centering
    \includegraphics[width=1.1\textwidth]{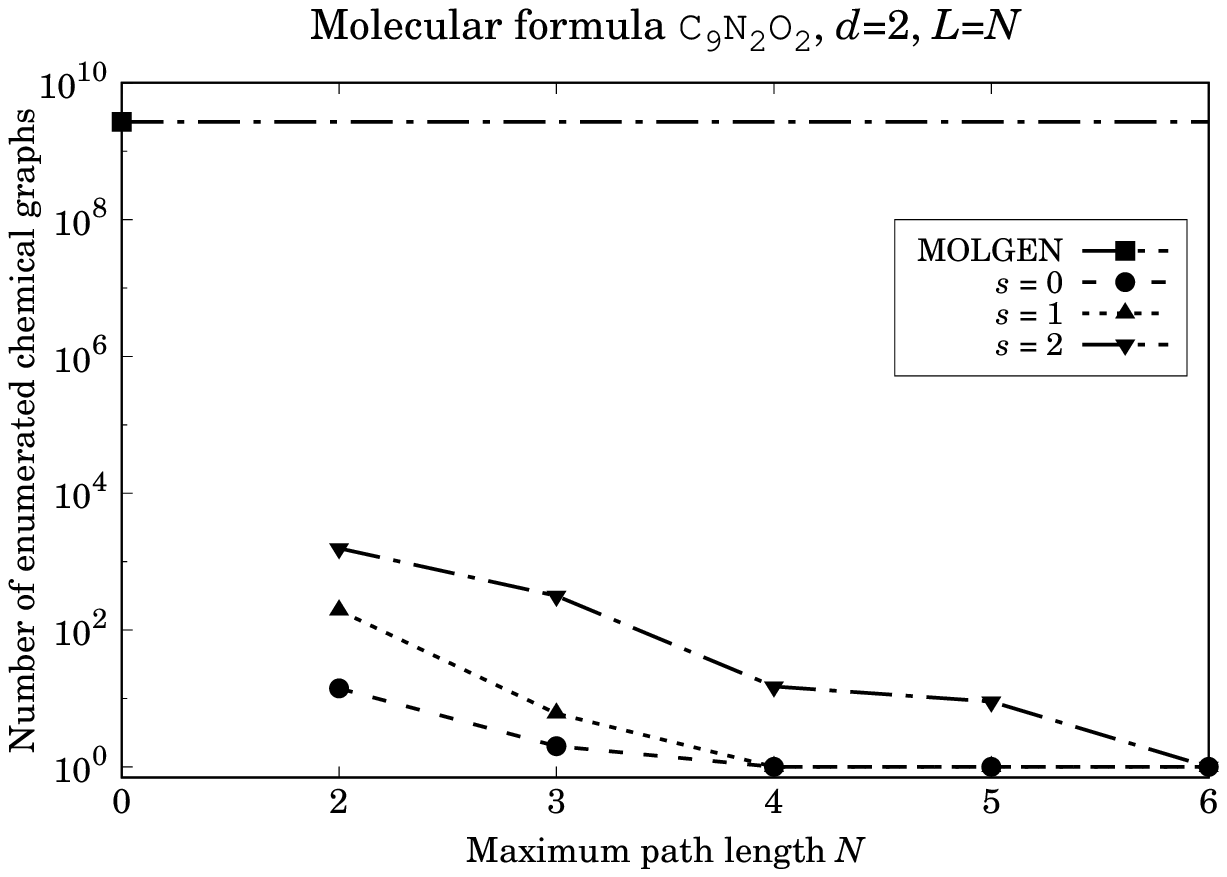}\\
    {\footnotesize (f)}\\
  \end{minipage} 
  \vspace{1cm}
  
  \caption{
    Plots showing the computation time 
    and number of chemical graphs enumerated by our algorithm
    for instance type EULF-$L$-A, as compared to MOLGEN.
    The sample structure from PubChem is with CID~6163405,
    molecular formula {\tt C$_9$N$_2$O$_2$},
    and maximum bond multiplicity~$d=2$.
    (a)-(c)~Running time;
    (d)-(f)~Number of enumerated chemical graphs.
  }
 \label{fig:result_graphs_3}
 \end{figure}

 \begin{figure}[!ht]
  \begin{minipage}{0.45\textwidth}
   \centering
   \includegraphics[width=1.1\textwidth]{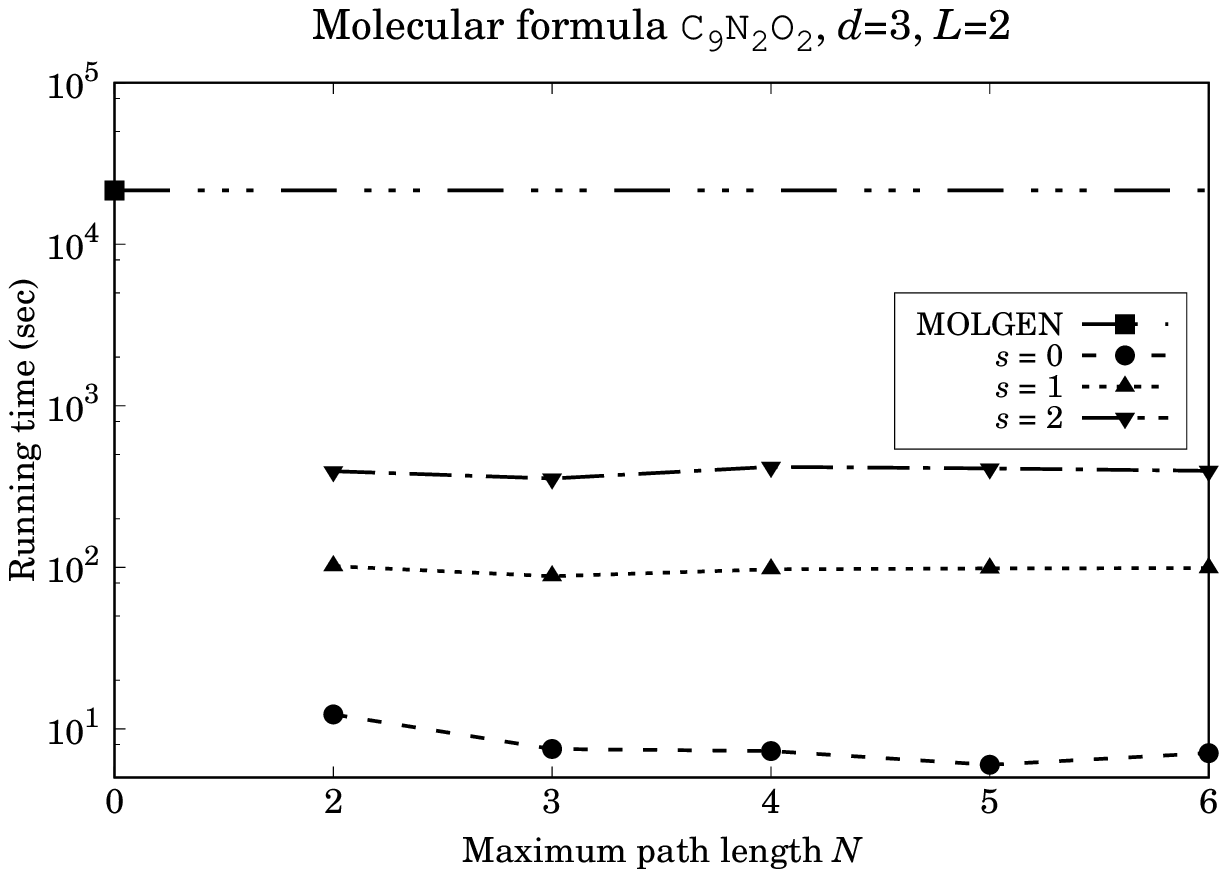}\\
   {\footnotesize (a)}\\
  \end{minipage}
\hfill
  \begin{minipage}{0.45\textwidth}
   \centering
   \includegraphics[width=1.1\textwidth]{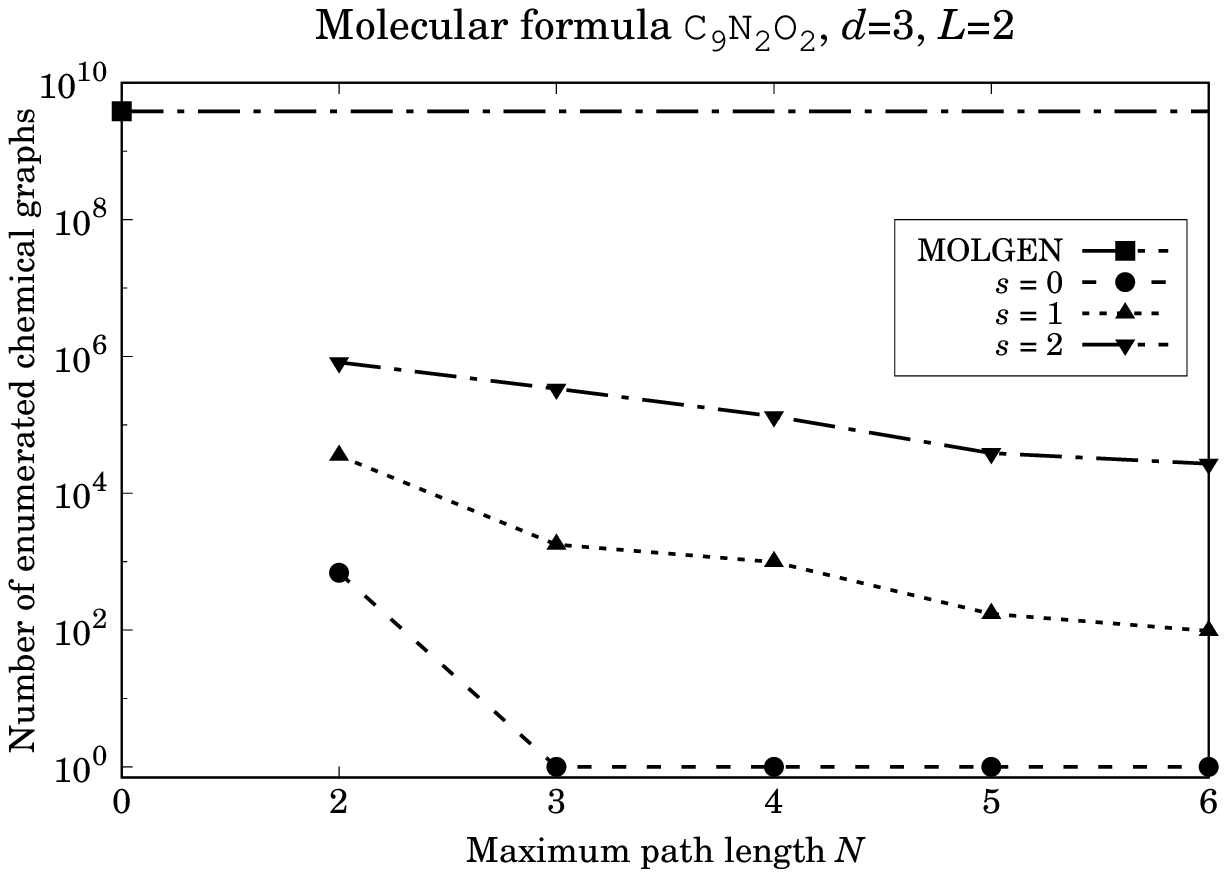}\\
   {\footnotesize (d)}\\
  \end{minipage} 
  \medskip

  \begin{minipage}{0.45\textwidth}
   \centering
      \includegraphics[width=1.1\textwidth]{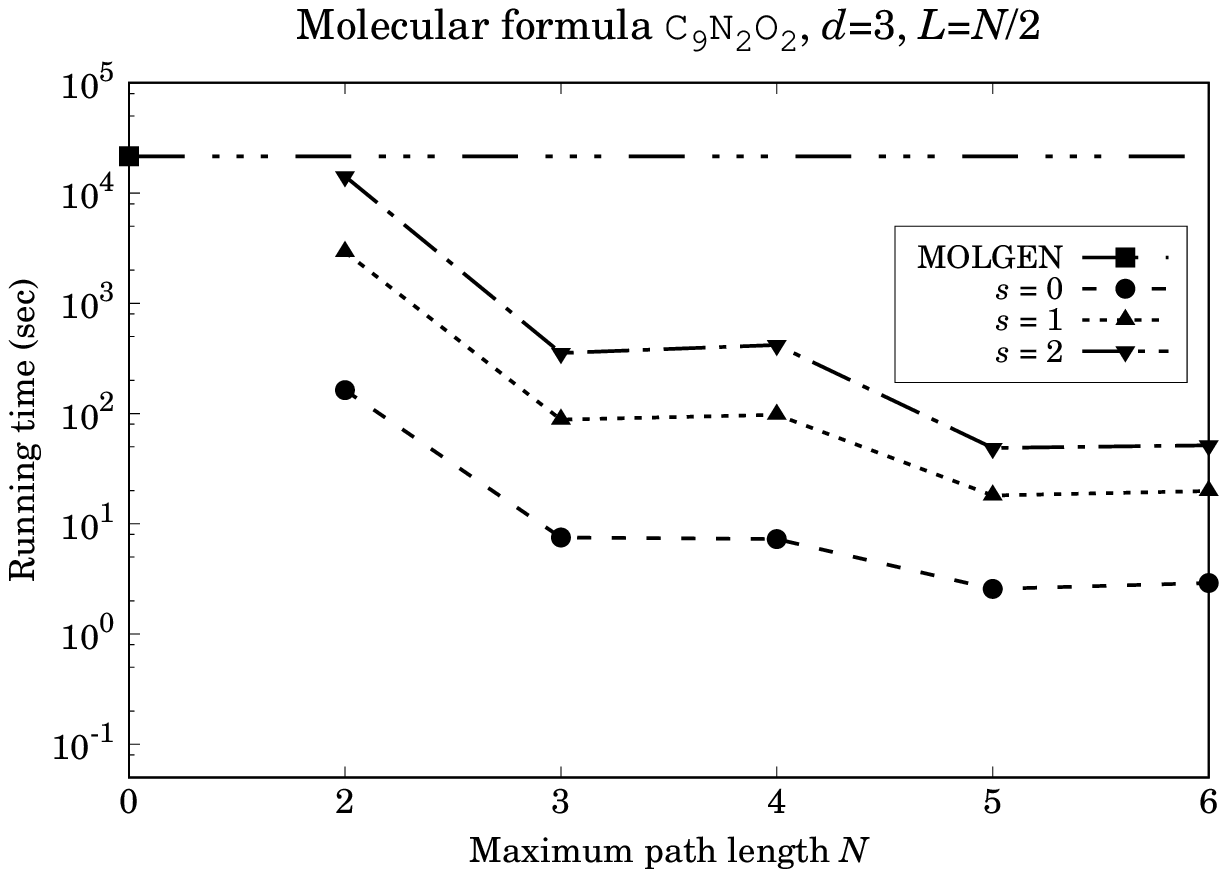}\\
      {\footnotesize (b)}\\
  \end{minipage} 
\hfill
  \begin{minipage}{0.45\textwidth}
   \centering
    \includegraphics[width=1.1\textwidth]{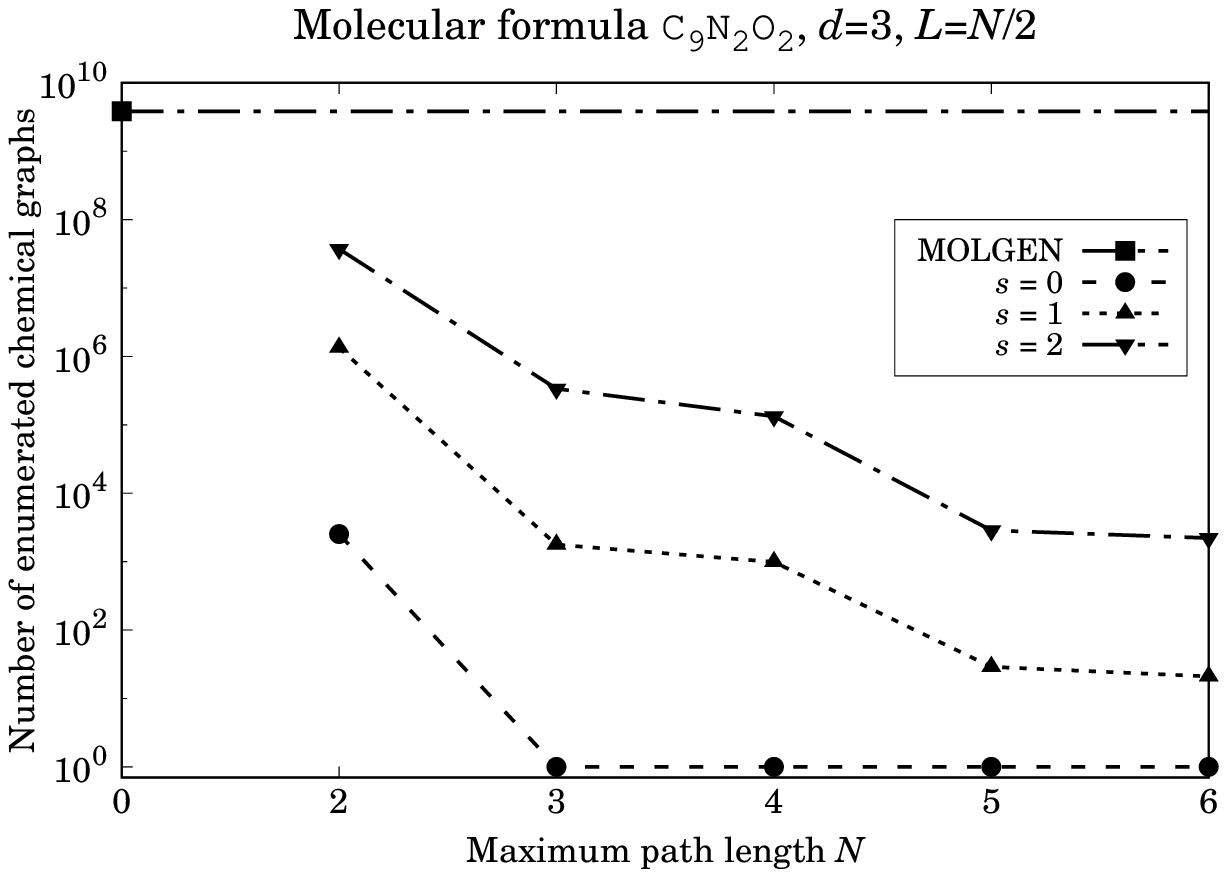}\\
    {\footnotesize (e)}\\
  \end{minipage} 
  \medskip

  \begin{minipage}{0.45\textwidth}
   \centering
      \includegraphics[width=1.1\textwidth]{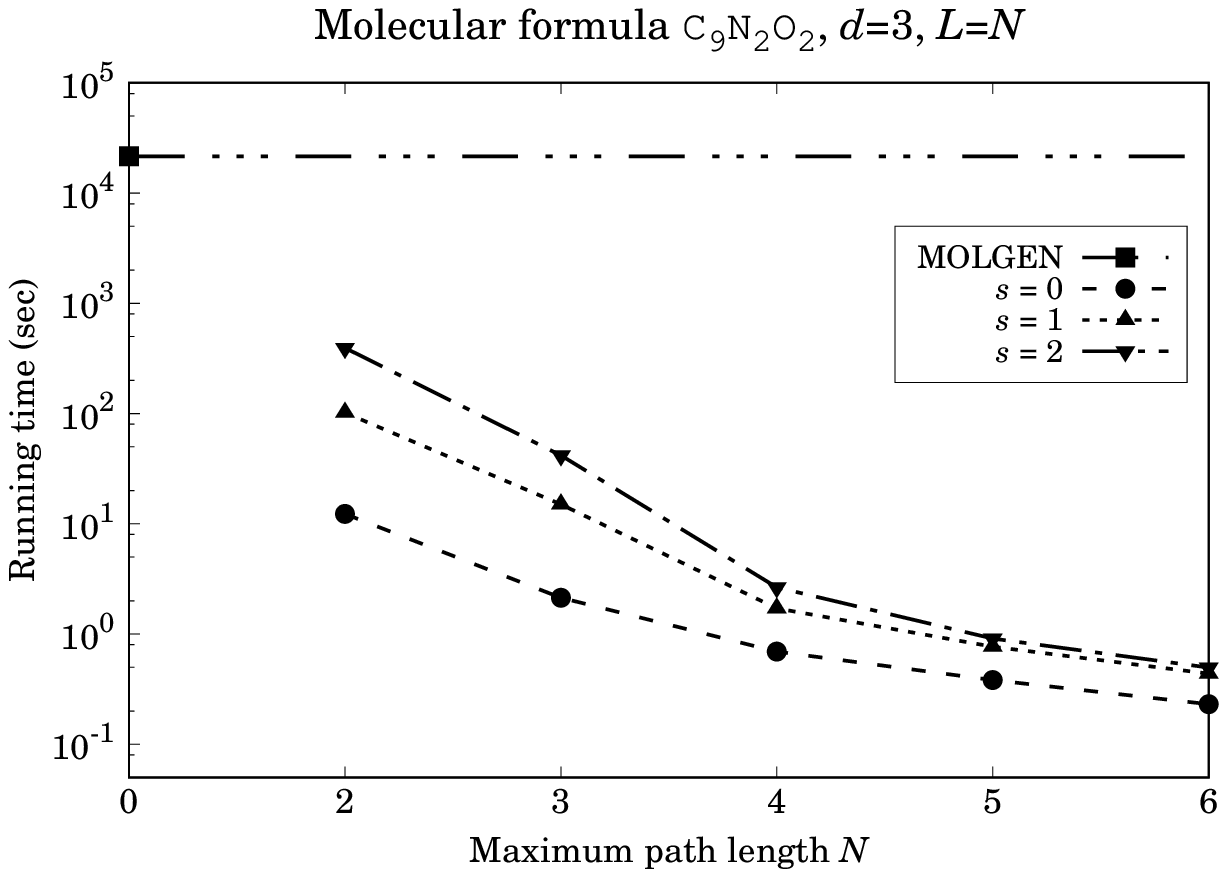}\\
      {\footnotesize (c)}\\
  \end{minipage} 
\hfill
  \begin{minipage}{0.45\textwidth}
   \centering
    \includegraphics[width=1.1\textwidth]{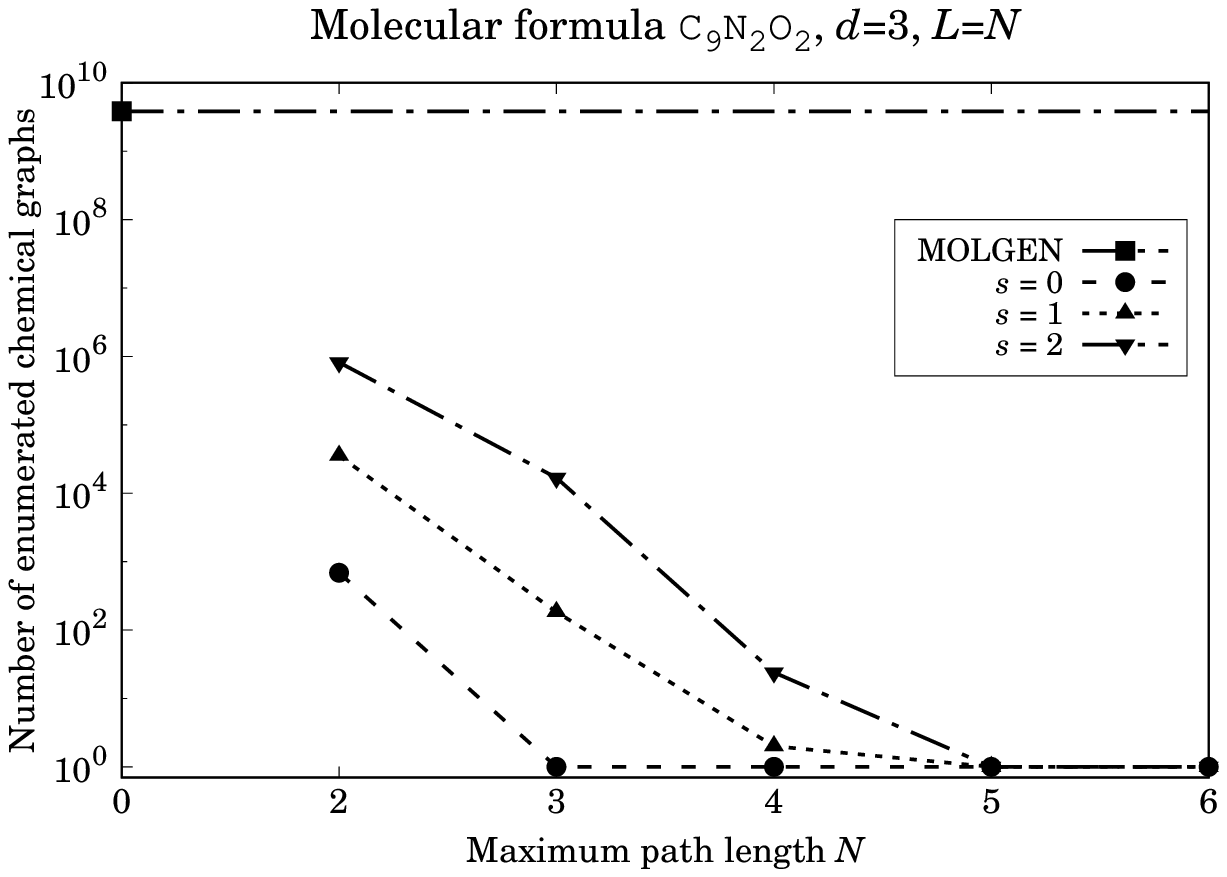}\\
    {\footnotesize (f)}\\
  \end{minipage} 
  \vspace{1cm}
  
  \caption{
    Plots showing the computation time 
    and number of chemical graphs enumerated by our algorithm
    for instance type EULF-$L$-A, as compared to MOLGEN.
    The sample structure from PubChem is with CID~131335510,
    molecular formula {\tt C$_9$N$_2$O$_2$},
    and maximum bond multiplicity~$d=3$.
    (a)-(c)~Running time;
    (d)-(f)~Number of enumerated chemical graphs.
  }
 \label{fig:result_graphs_4}
 \end{figure}

 \begin{figure}[!ht]
  \begin{minipage}{0.45\textwidth}
   \centering
   \includegraphics[width=1.1\textwidth]{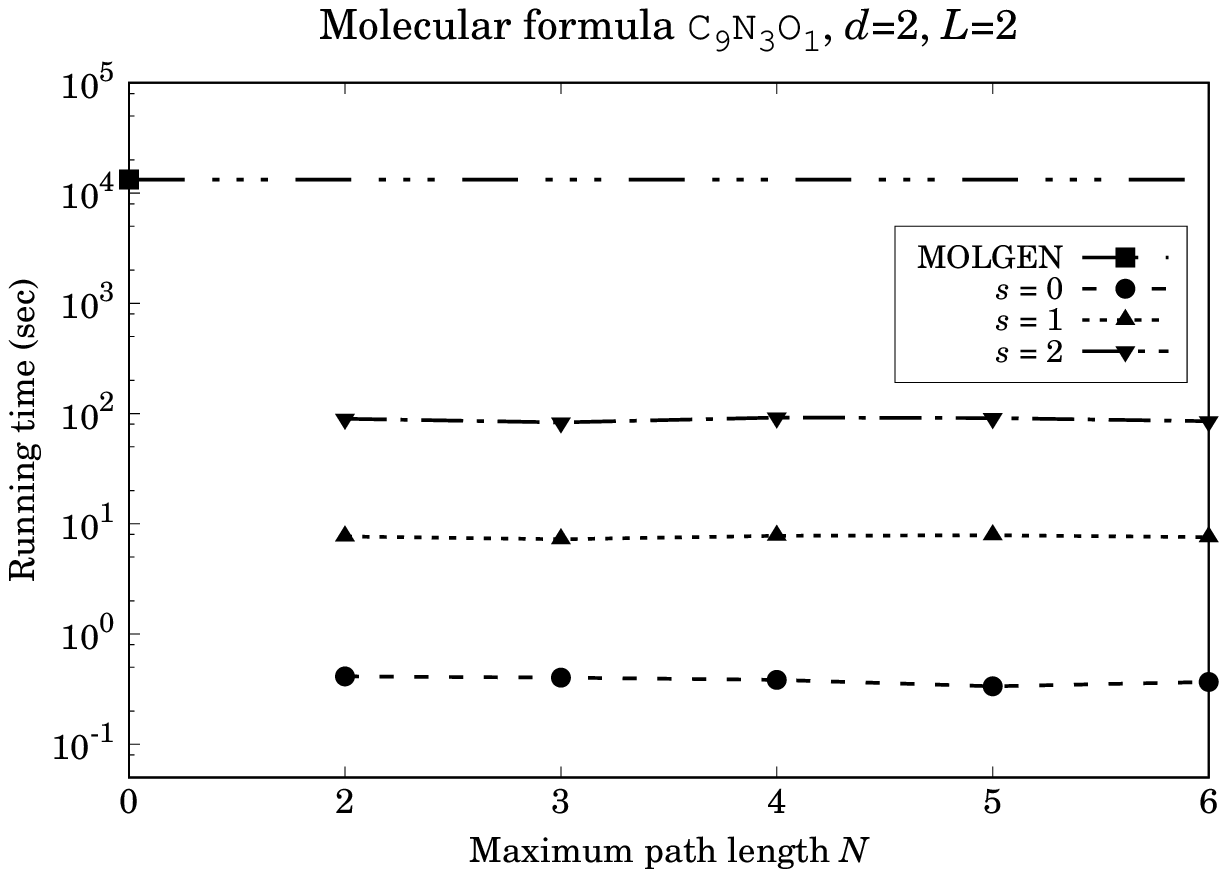}\\
   {\footnotesize (a)}\\
  \end{minipage}
\hfill
  \begin{minipage}{0.45\textwidth}
   \centering
   \includegraphics[width=1.1\textwidth]{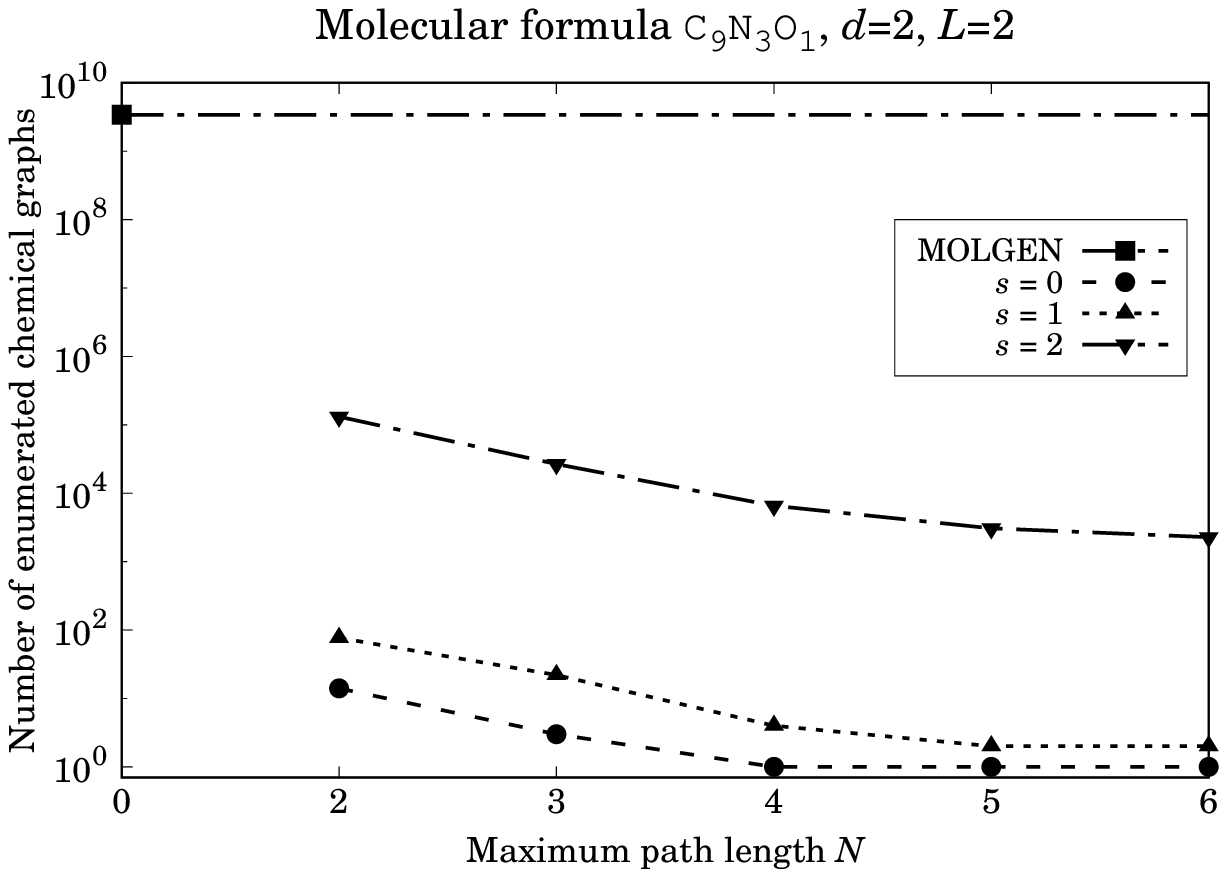}\\
   {\footnotesize (d)}\\
  \end{minipage} 
  \medskip

  \begin{minipage}{0.45\textwidth}
   \centering
      \includegraphics[width=1.1\textwidth]{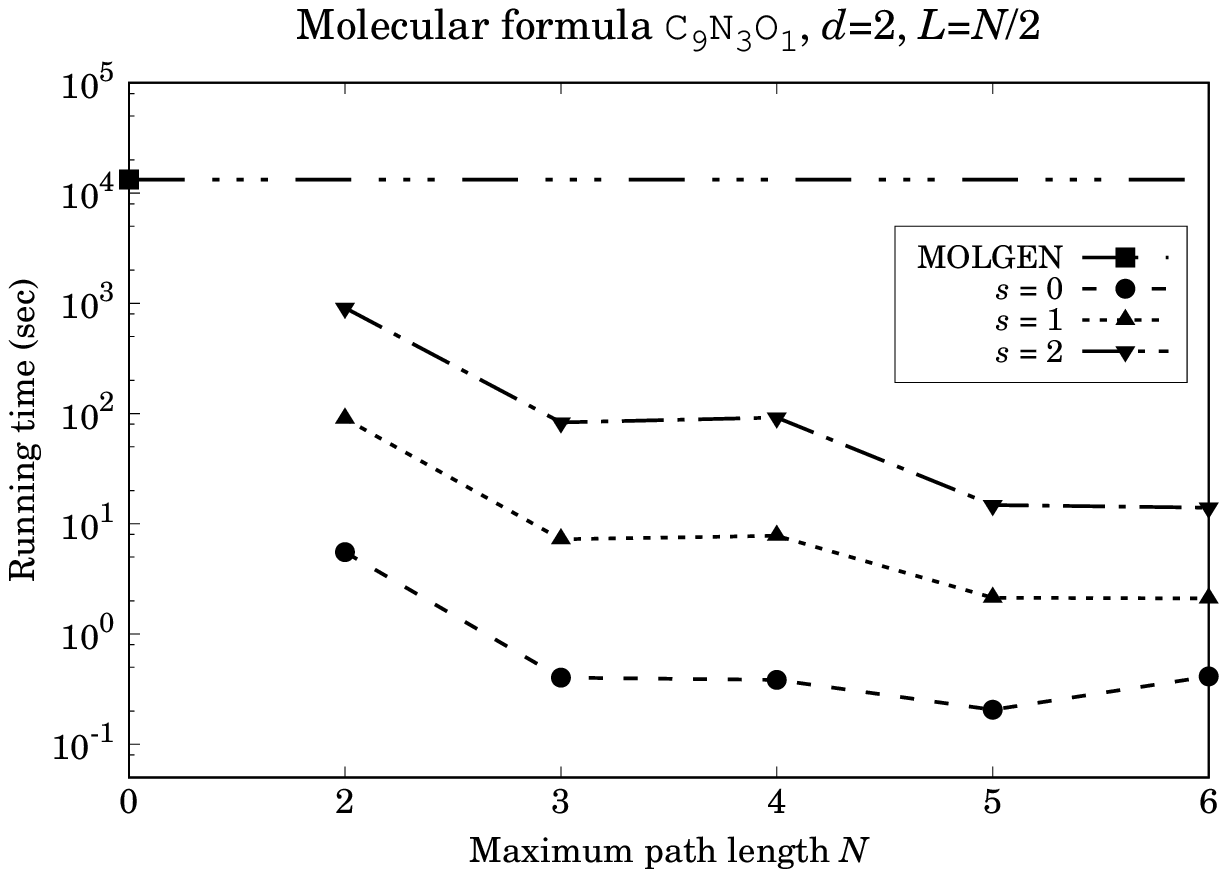}\\
      {\footnotesize (b)}\\
  \end{minipage} 
\hfill
  \begin{minipage}{0.45\textwidth}
   \centering
    \includegraphics[width=1.1\textwidth]{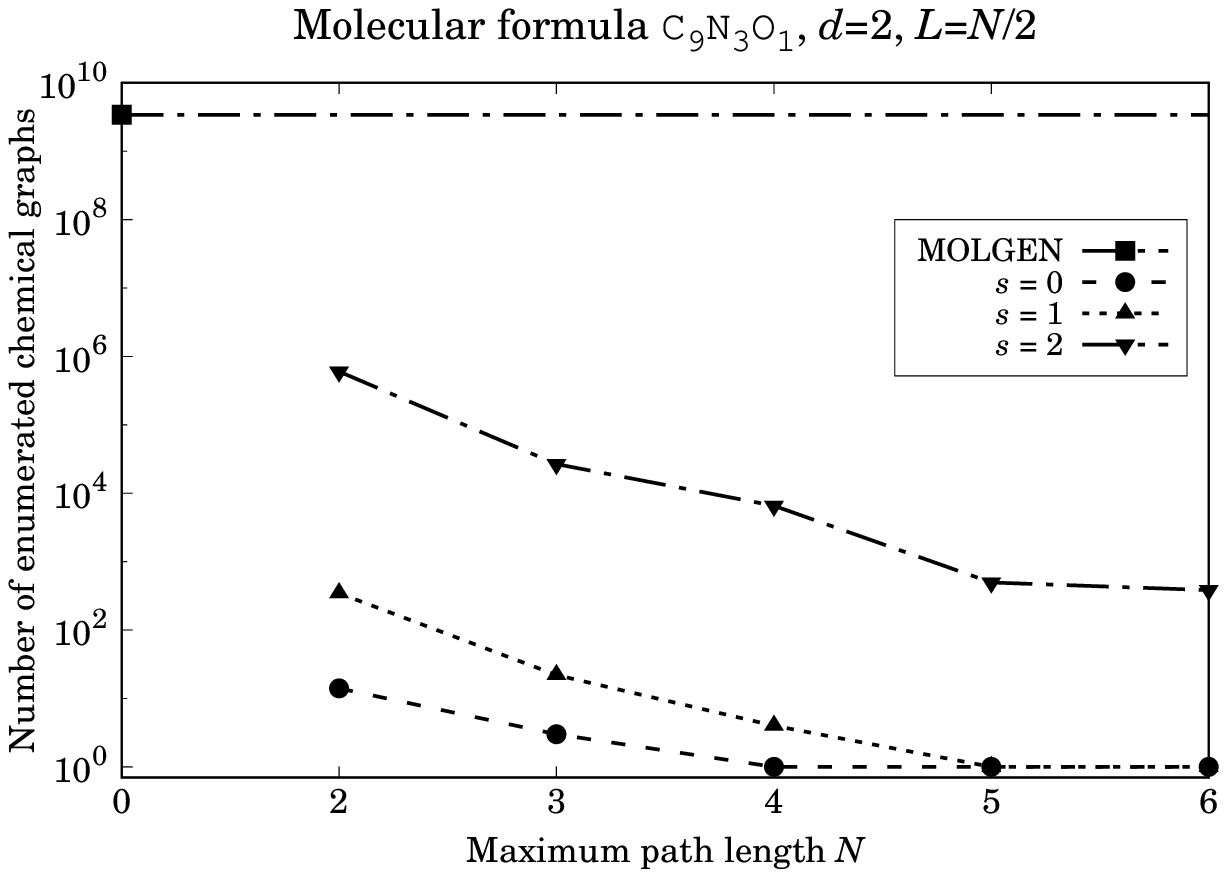}\\
    {\footnotesize (e)}\\
  \end{minipage} 
  \medskip

  \begin{minipage}{0.45\textwidth}
   \centering
      \includegraphics[width=1.1\textwidth]{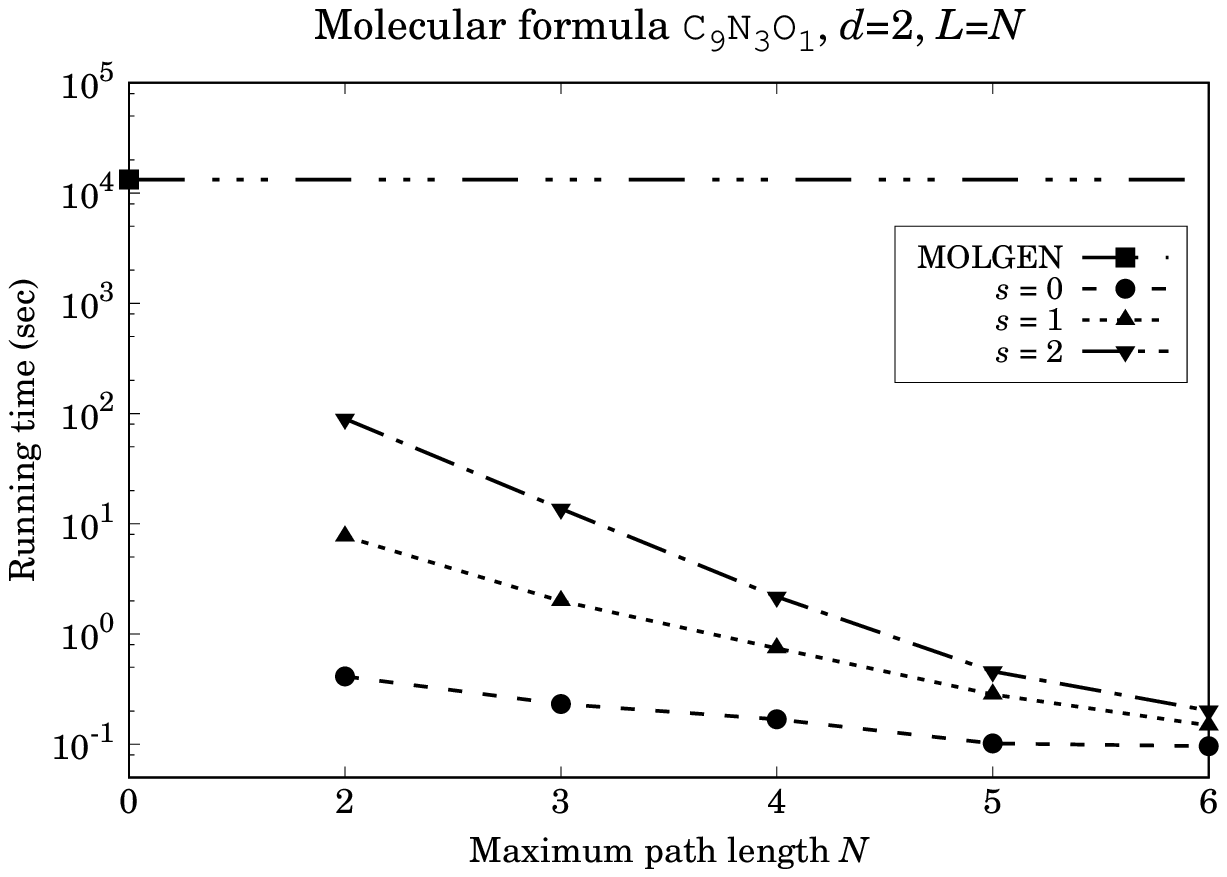}\\
      {\footnotesize (c)}\\
  \end{minipage} 
\hfill
  \begin{minipage}{0.45\textwidth}
   \centering
    \includegraphics[width=1.1\textwidth]{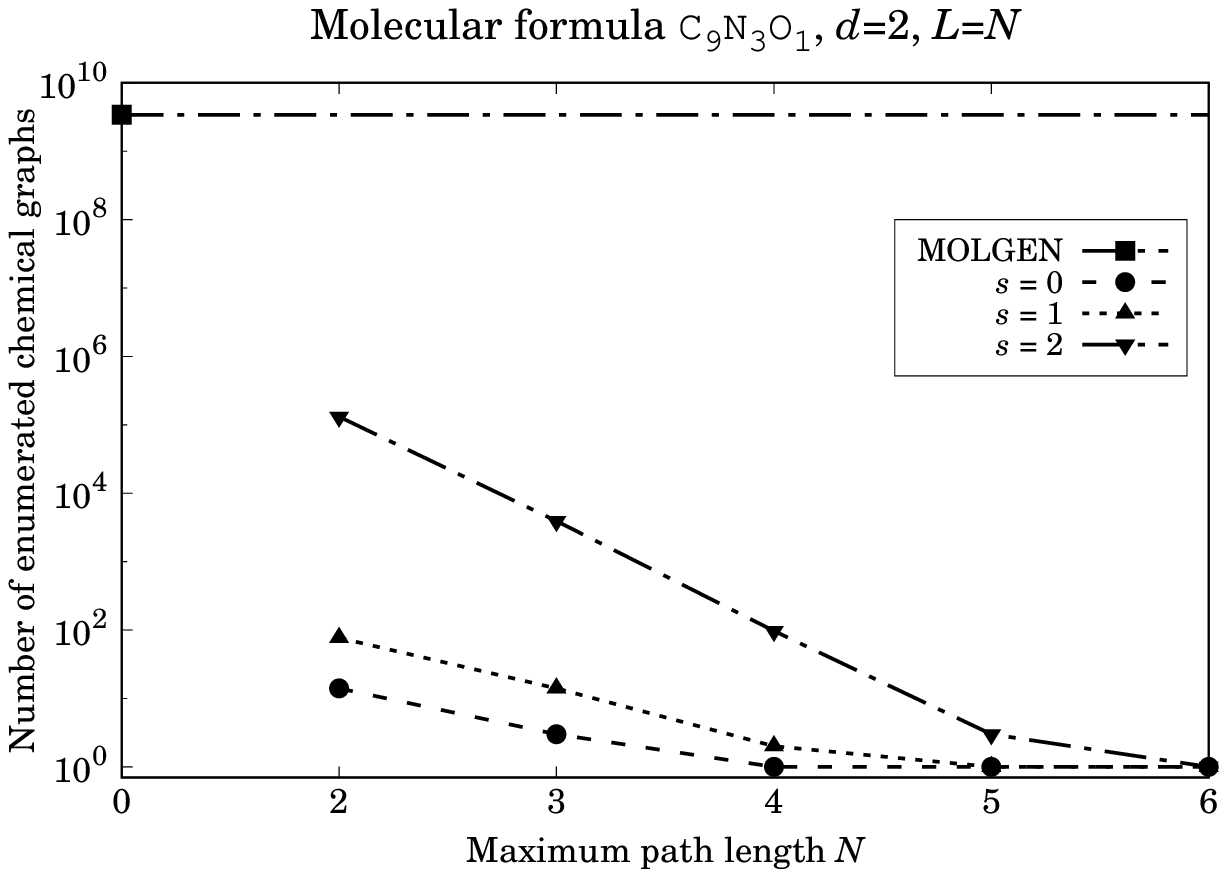}\\
    {\footnotesize (f)}\\
  \end{minipage} 
  \vspace{1cm}
  
  \caption{
    Plots showing the computation time 
    and number of chemical graphs enumerated by our algorithm
    for instance type EULF-$L$-A, as compared to MOLGEN.
    The sample structure from PubChem is with CID~9942278,
    molecular formula {\tt C$_9$N$_3$O$_1$},
    and maximum bond multiplicity~$d=2$.
    (a)-(c)~Running time;
    (d)-(f)~Number of enumerated chemical graphs.
  }
 \label{fig:result_graphs_5}
 \end{figure}

 \begin{figure}[!ht]
  \begin{minipage}{0.45\textwidth}
   \centering
   \includegraphics[width=1.1\textwidth]{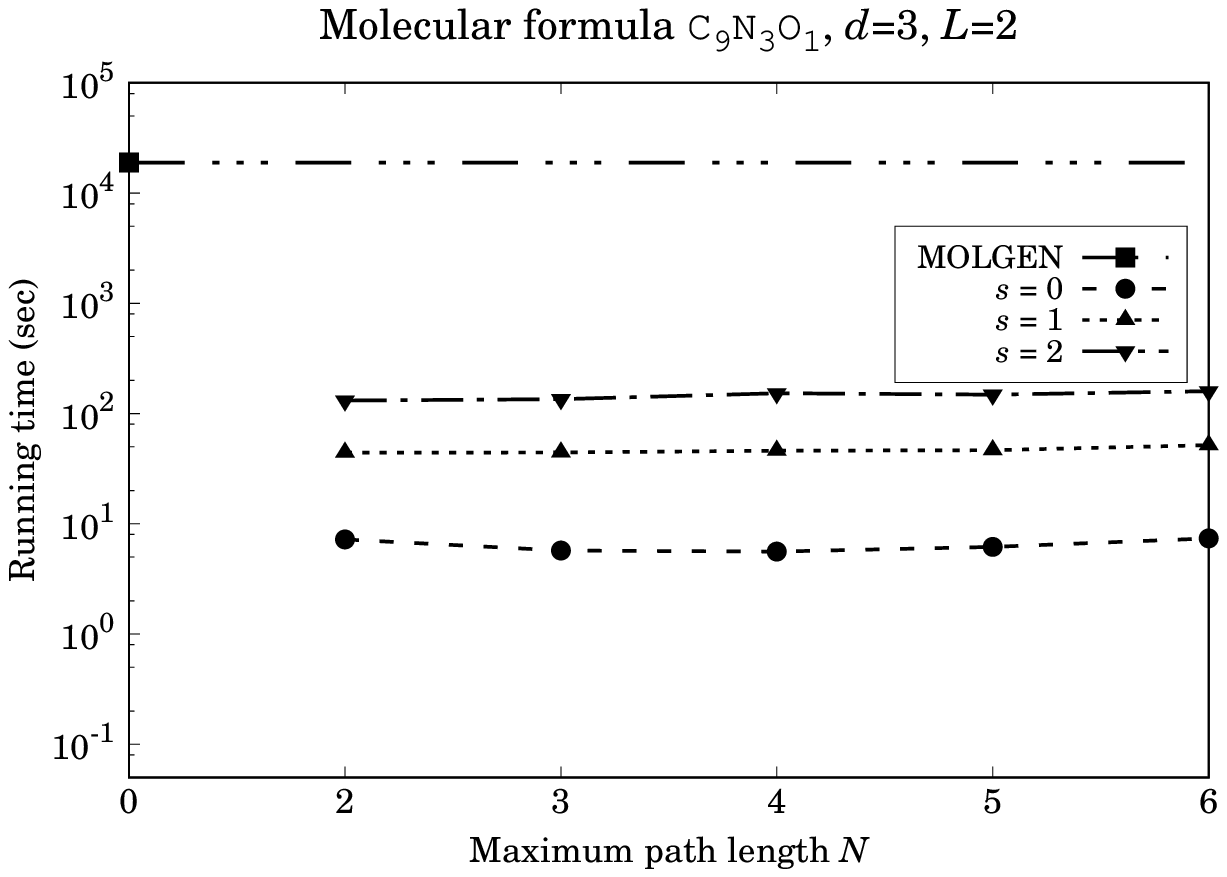}\\
   {\footnotesize (a)}\\
  \end{minipage}
\hfill
  \begin{minipage}{0.45\textwidth}
   \centering
   \includegraphics[width=1.1\textwidth]{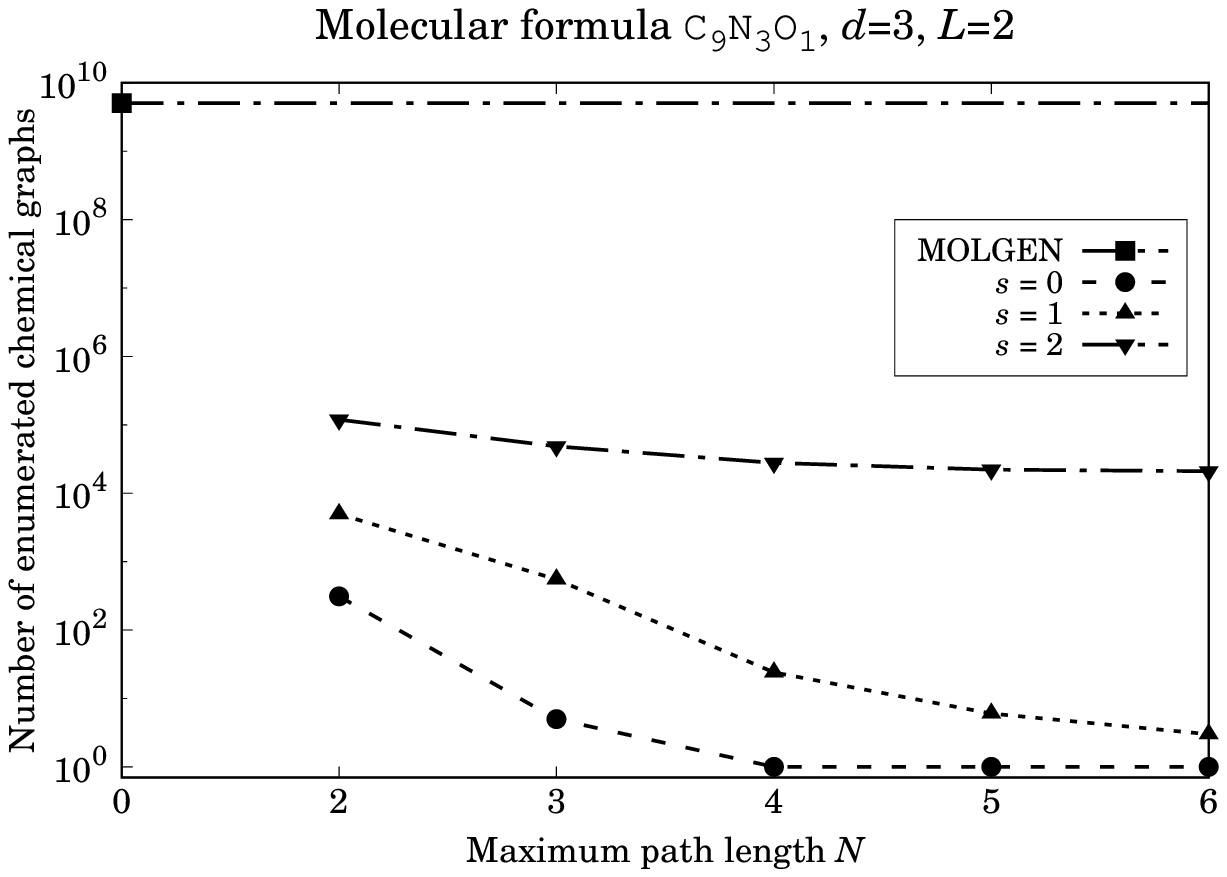}\\
   {\footnotesize (d)}\\
  \end{minipage} 
  \medskip

  \begin{minipage}{0.45\textwidth}
   \centering
      \includegraphics[width=1.1\textwidth]{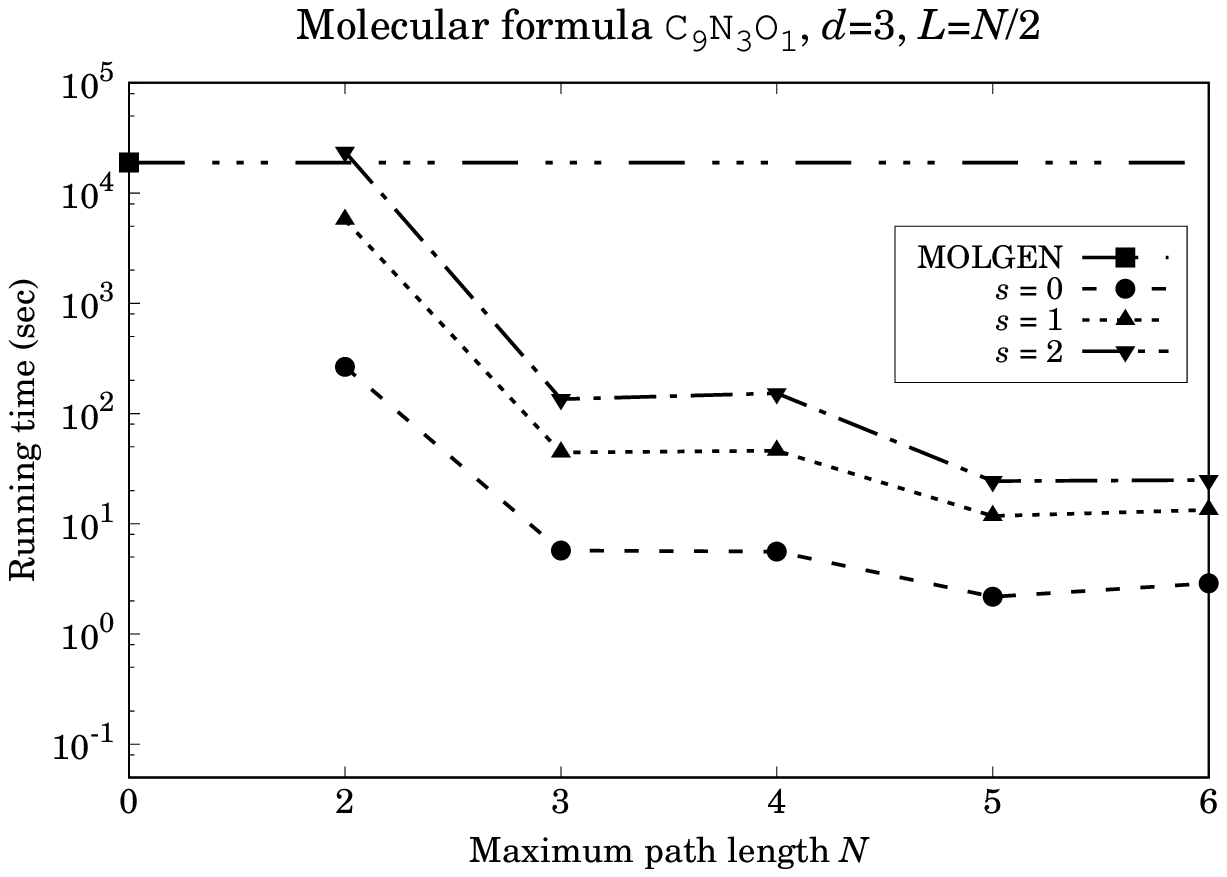}\\
      {\footnotesize (b)}\\
  \end{minipage} 
\hfill
  \begin{minipage}{0.45\textwidth}
   \centering
    \includegraphics[width=1.1\textwidth]{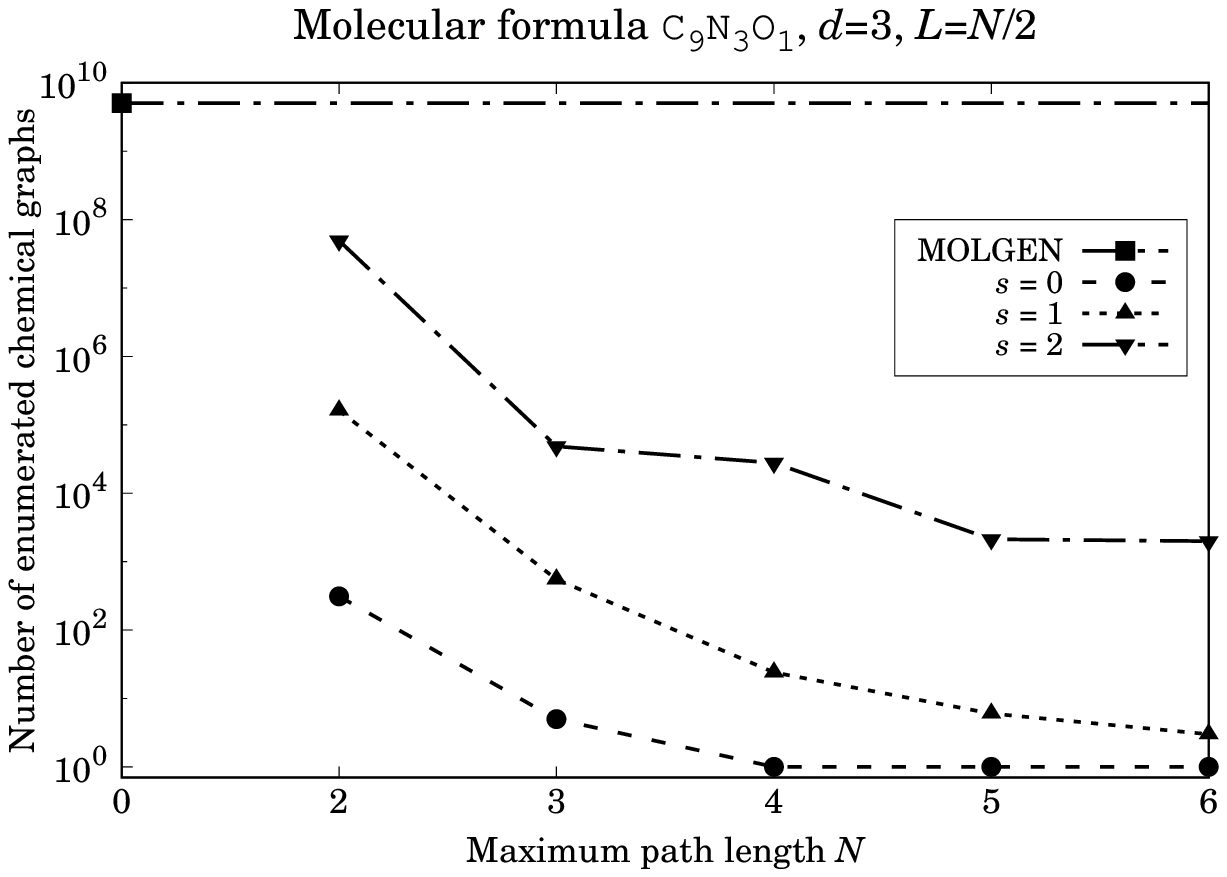}\\
    {\footnotesize (e)}\\
  \end{minipage} 
  \medskip

  \begin{minipage}{0.45\textwidth}
   \centering
      \includegraphics[width=1.1\textwidth]{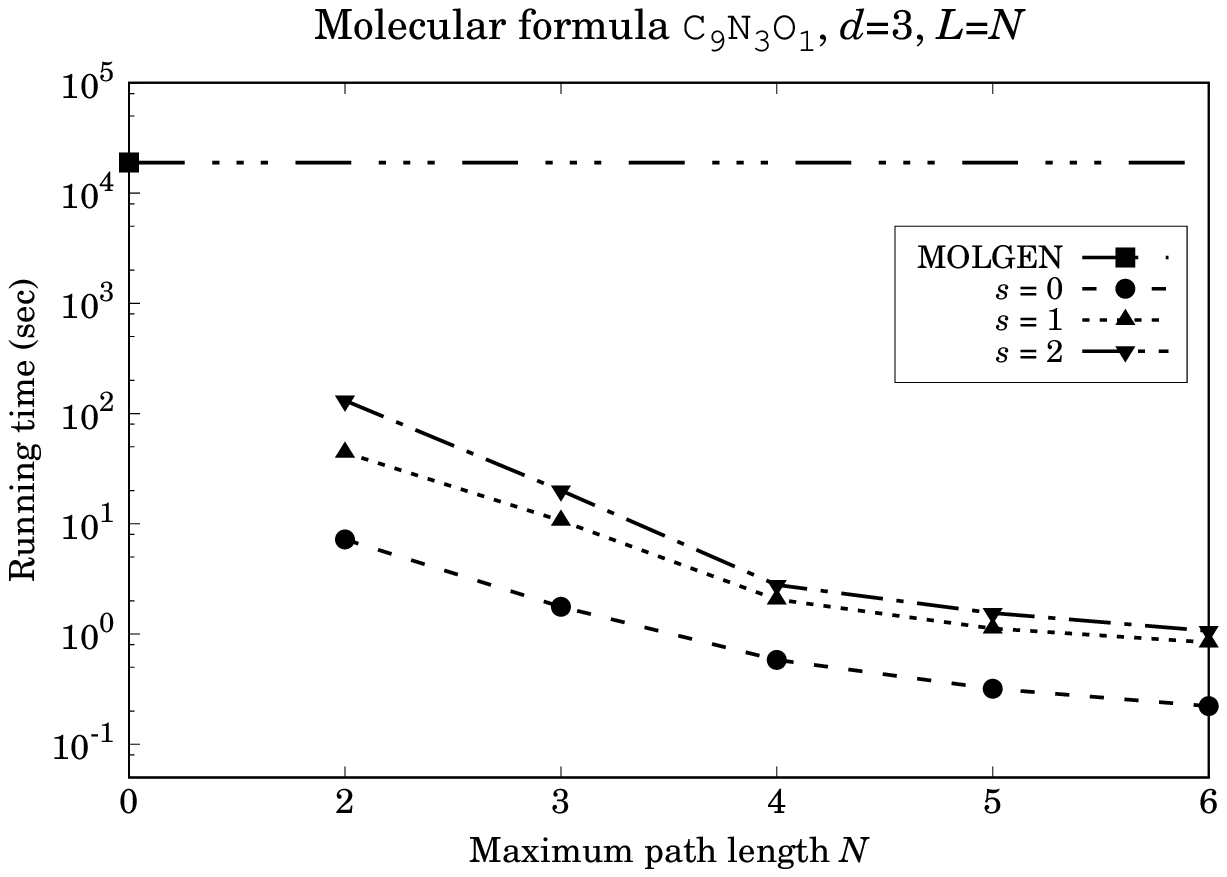}\\
      {\footnotesize (c)}\\
  \end{minipage} 
\hfill
  \begin{minipage}{0.45\textwidth}
   \centering
    \includegraphics[width=1.1\textwidth]{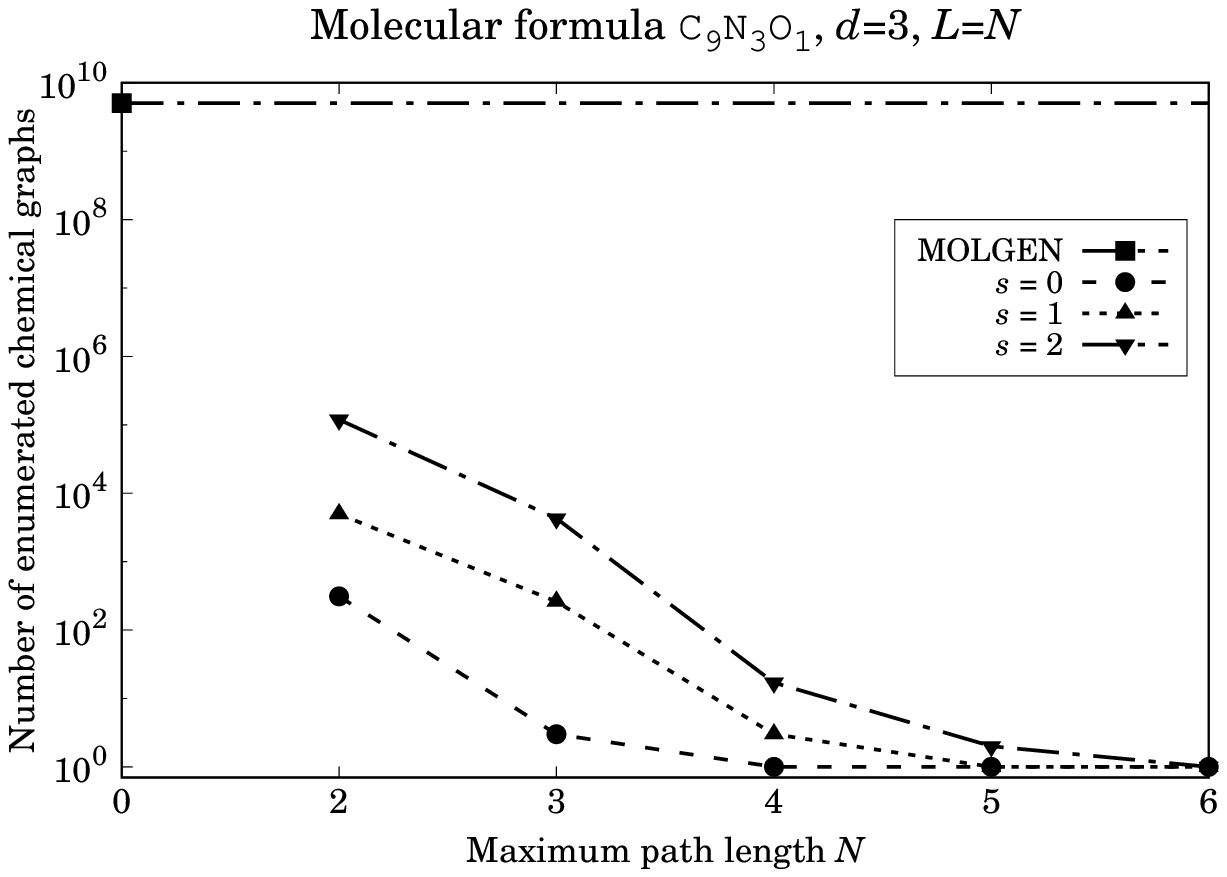}\\
    {\footnotesize (f)}\\
  \end{minipage} 
  \vspace{1cm}
  
  \caption{
    Plots showing the computation time 
    and number of chemical graphs enumerated by our algorithm
    for instance type EULF-$L$-A, as compared to MOLGEN.
    The sample structure from PubChem is with CID~10103630,
    molecular formula {\tt C$_9$N$_3$O$_1$},
    and maximum bond multiplicity~$d=3$.
    (a)-(c)~Running time;
    (d)-(f)~Number of enumerated chemical graphs.
  }
 \label{fig:result_graphs_6}
 \end{figure}

 In addition, to check the limits as to the maximum number of vertices
 in graphs that can be enumerated in a reasonable time, 
 we conducted experiments over a range $n \in [9, 40]$ for the number 
 of vertices in a target chemical graph.
 For a fixed number $n$ of vertices,
 we tested two types of instances, one with molecular formula $\mathtt{C}_n$,
 and the other with molecular formula  $\mathtt{C}_{n-4} \mathtt{N}_2 \mathtt{O}_2$,
 and set an execution time limit of $3,600$ seconds.
 The results are summarized in Fig.~\ref{fig:result_graphs_num_limit}.
 From Fig.~\ref{fig:result_graphs_num_limit}\,(a) and~(b), we see that the
 time limit is quickly reached even when the number $n$ of 
 number of vertices is less than 15, but that the program still enumerates 
 structures within the time limit up to $n=30$, after which there are cases when
 not even a single graph is enumerated during the time limit.

   \begin{figure}[!ht]
  \begin{minipage}{0.45\textwidth}
   \centering
   \includegraphics[width=1.1\textwidth]{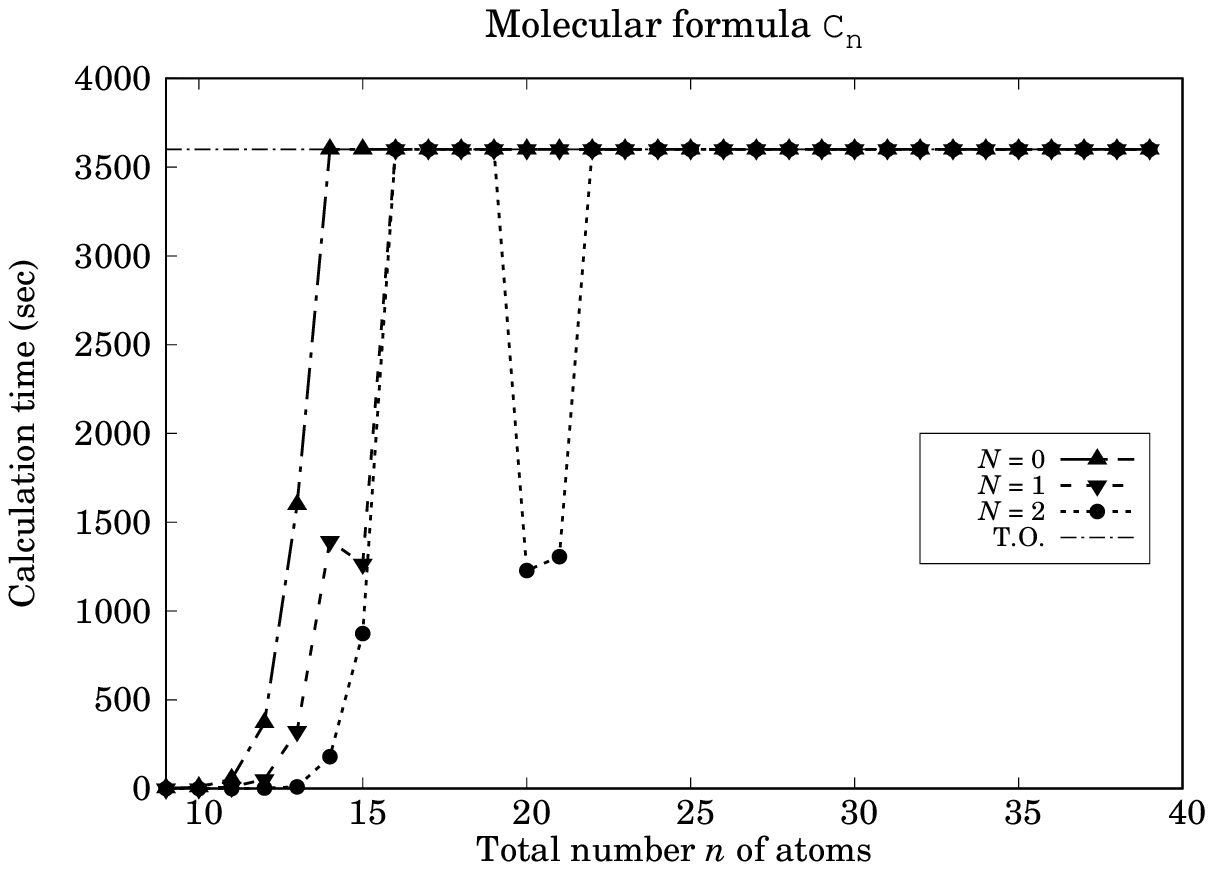}\\
   {\footnotesize (a)}\\
  \end{minipage}
\hfill
  \begin{minipage}{0.45\textwidth}
   \centering
   \includegraphics[width=1.1\textwidth]{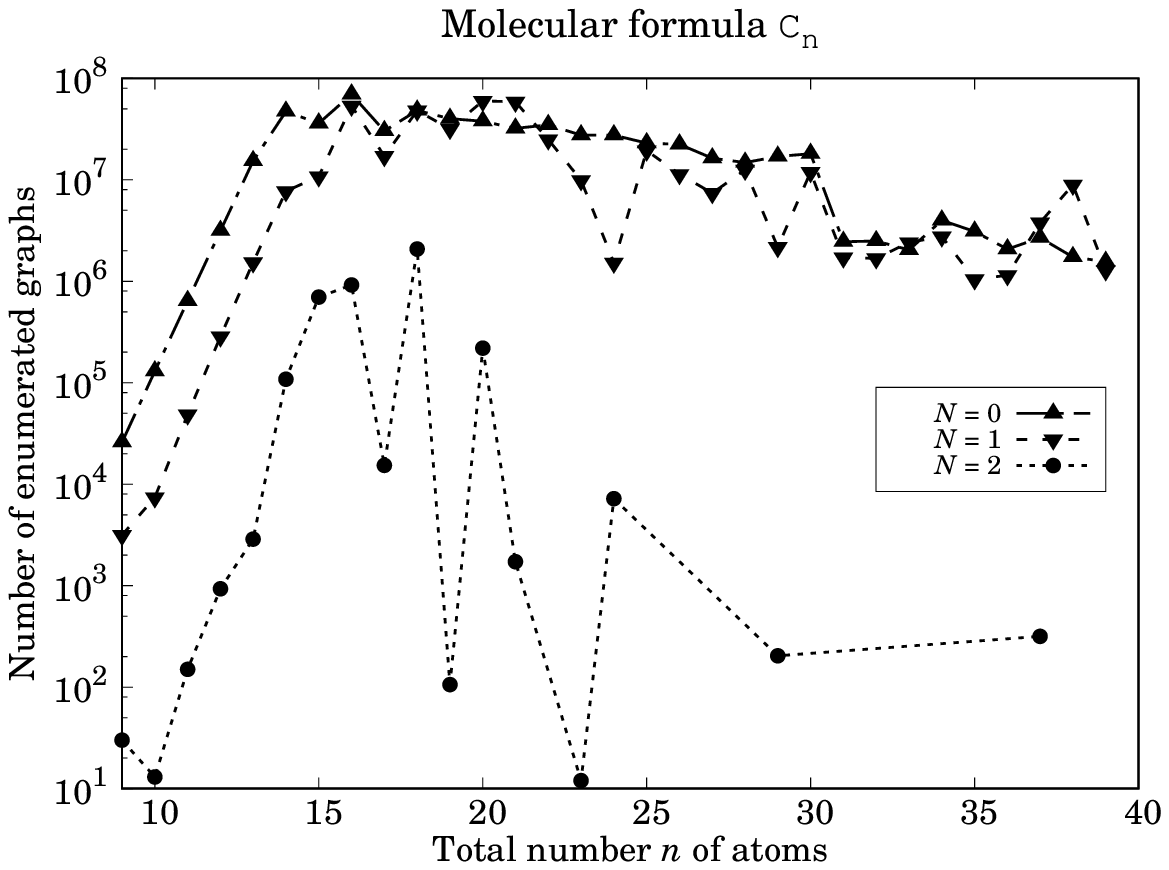}\\
   {\footnotesize (c)}\\
  \end{minipage} 
  \medskip

  \begin{minipage}{0.45\textwidth}
   \centering
      \includegraphics[width=1.1\textwidth]{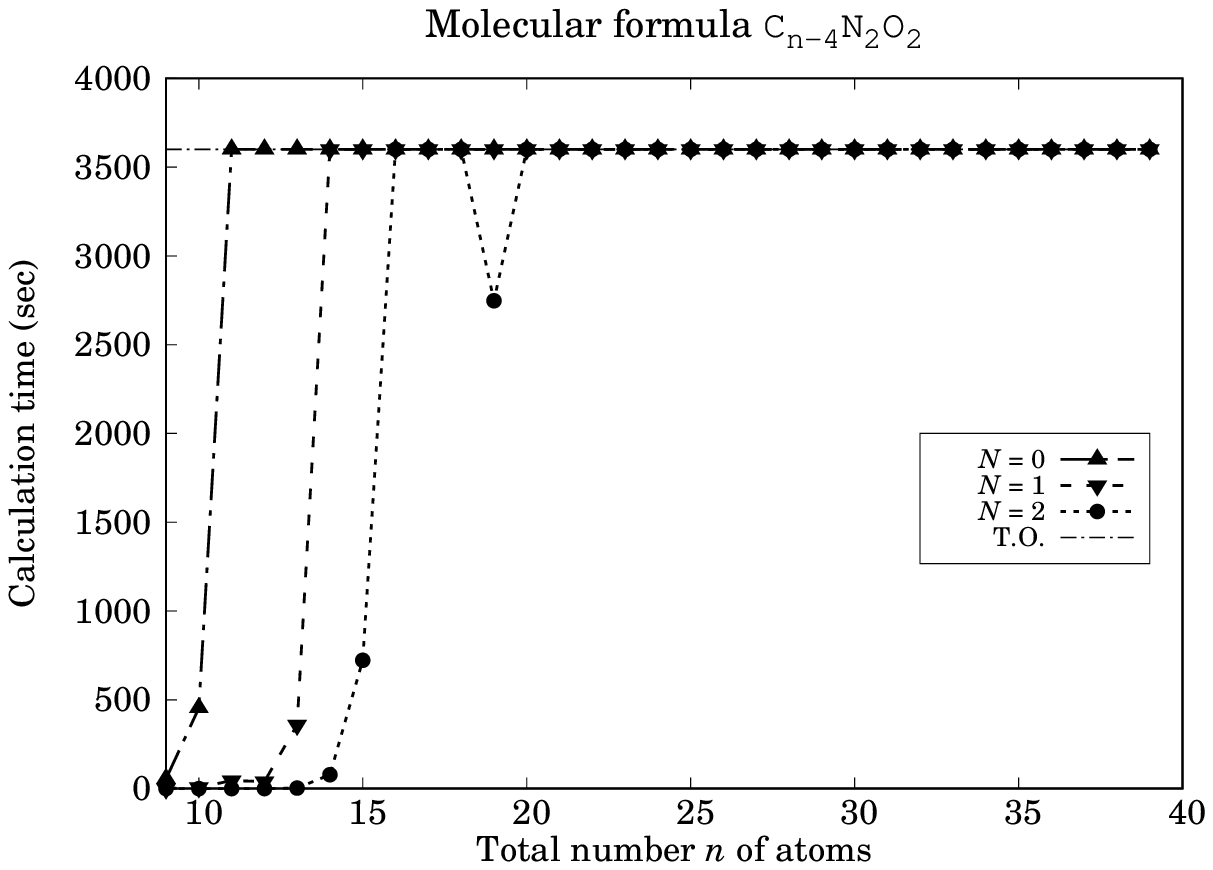}\\
      {\footnotesize (b)}\\
  \end{minipage} 
\hfill
  \begin{minipage}{0.45\textwidth}
   \centering
    \includegraphics[width=1.1\textwidth]{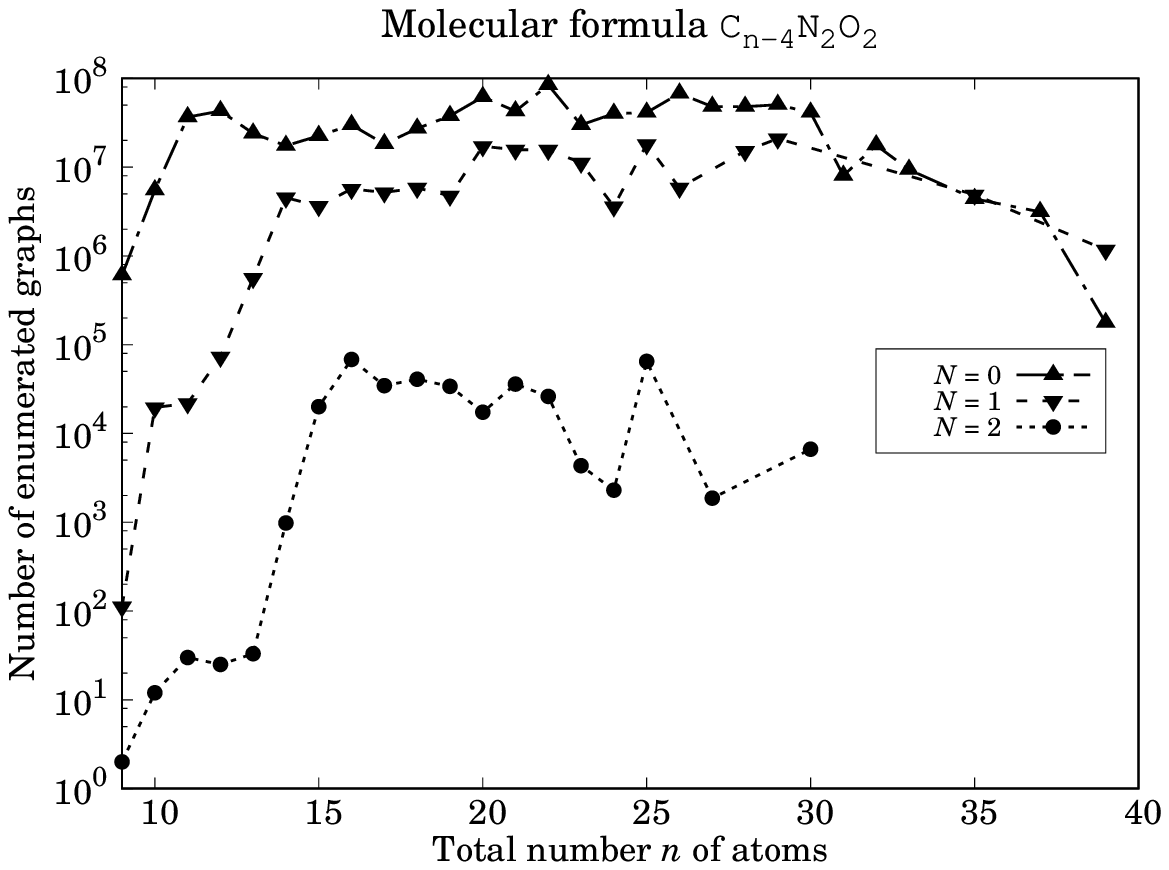}\\
    {\footnotesize (d)}\\
  \end{minipage} 
  \medskip
  
  \vspace{1cm}
  
  \caption{
    Plots showing the computation time 
    and number of chemical graphs enumerated by our algorithm
    for instance type EULF-$L$-A, over ranges for $N \in [0, 2]$, $L = N$,
    $d = 3$ and $s = 0$.
    (a), (b)~Running time, T.O. stands for ``Time Out'';
    (c), (d)~Number of enumerated chemical graphs.
  }
 \label{fig:result_graphs_num_limit}
 \end{figure}

 \subsection{Experimental Results on EULF-$L$-P}
 \label{sec:experiments_EULF2}

We conduct similar computational 
experiments to test the performance of
our algorithm for Problem EULF-$L$-P
as in Section~\ref{sec:experiments_EULF1}.
We took values for $N \in [8, 10]$,
and $L \in \{ 2, 3 \}$.

The results from our experiments 
for instance type EULF-$L$-P are summarized in 
Figs.~\ref{fig:result_graphs_1.2} to~\ref{fig:result_graphs_6.2}.
Our results for instance type EULF-$L$-P 
indicate that there exist very few chemical graphs that
satisfy the path frequency specification for
our choice of a set $\pathset$ of colored paths
obtained from the six compounds from the PubChem database, and parameter~$L$.
In fact, the only two instances where our algorithm enumerates
any chemical graphs are for two of our chosen compounds; 
with CID 301729, molecular formula {\tt C$_9$N$_1$O$_3$} and bond multiplicity at most~2
(Fig.~\ref{fig:result_graphs_1.2}\,(c) and~(d)),
and
the compound with CID 10103630, molecular formula {\tt C$_9$N$_3$O$_1$}
and bond multiplicity at most~3
(Fig.~\ref{fig:result_graphs_6.2}\,(c) and~(d)).
This could be due to the very nature of mono-block 2-augmented structures,
namely, due to the biconnectedness of a block, a single path frequency specification
exhibits a strong influence on the structure of a chemical graph.

In addition, we observe that the running time of our algorithm, even when there are no
enumerated chemical graphs, grows rapidly with the value of the parameter~$L$.
It is an interesting idea for future research to 
improve our algorithm in such a way that the non-existence of any
chemical graphs that satisfy a given path frequency specification is determined 
much quicker.

  \begin{figure}[!ht]
  \begin{minipage}{0.45\textwidth}
   \centering
   \includegraphics[width=1.1\textwidth]{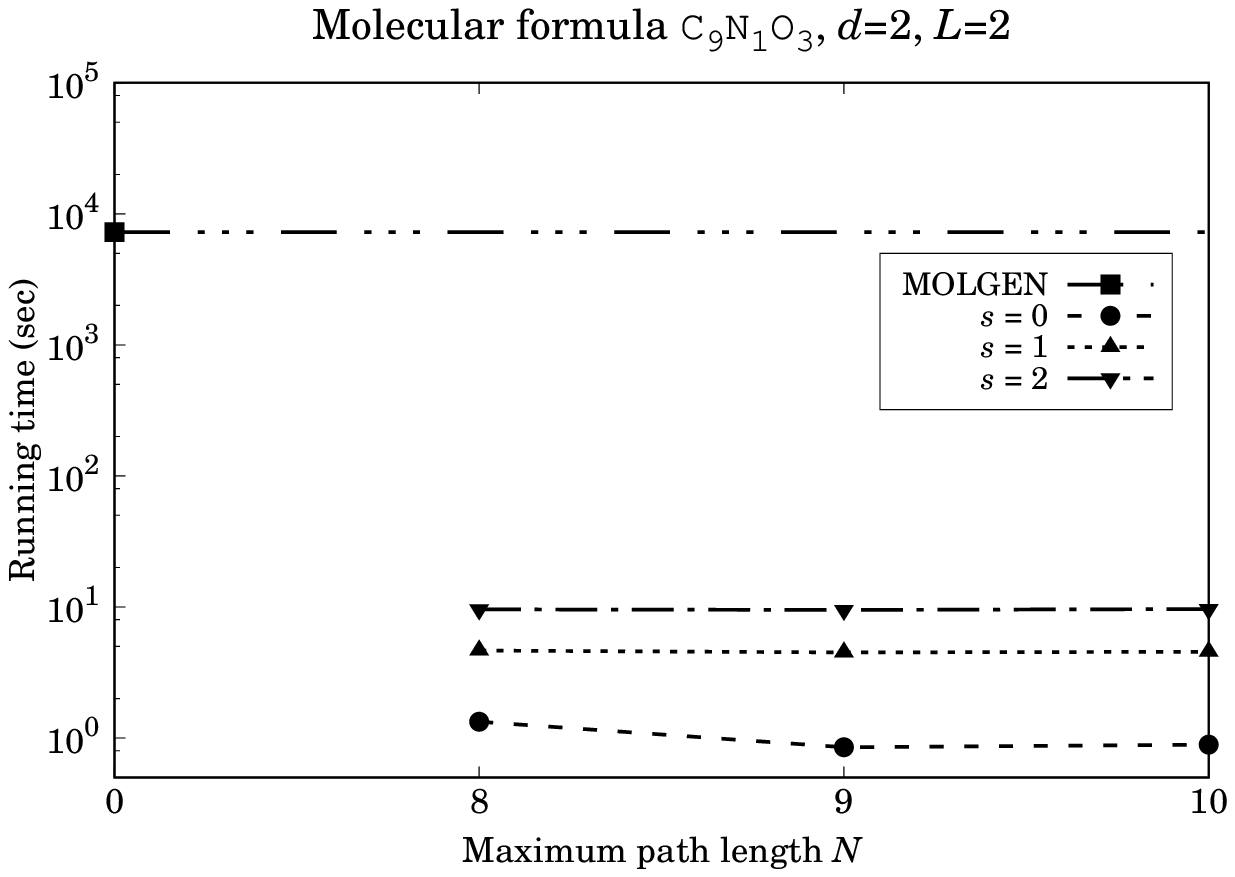}\\
   {\footnotesize (a)}\\
  \end{minipage}
\hfill
  \begin{minipage}{0.45\textwidth}
   \centering
   \includegraphics[width=1.1\textwidth]{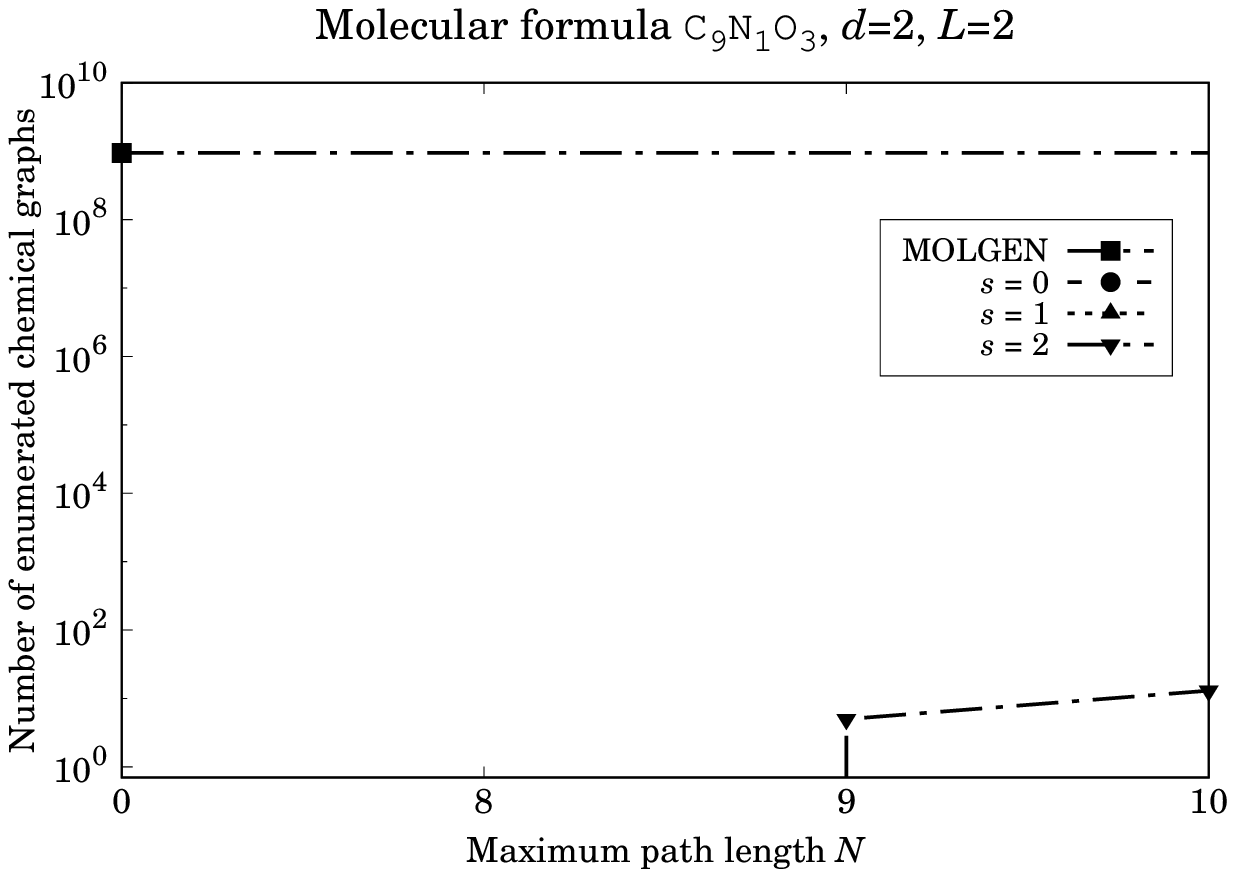}\\
   {\footnotesize (c)}\\
  \end{minipage} 
  \medskip

  \begin{minipage}{0.45\textwidth}
   \centering
      \includegraphics[width=1.1\textwidth]{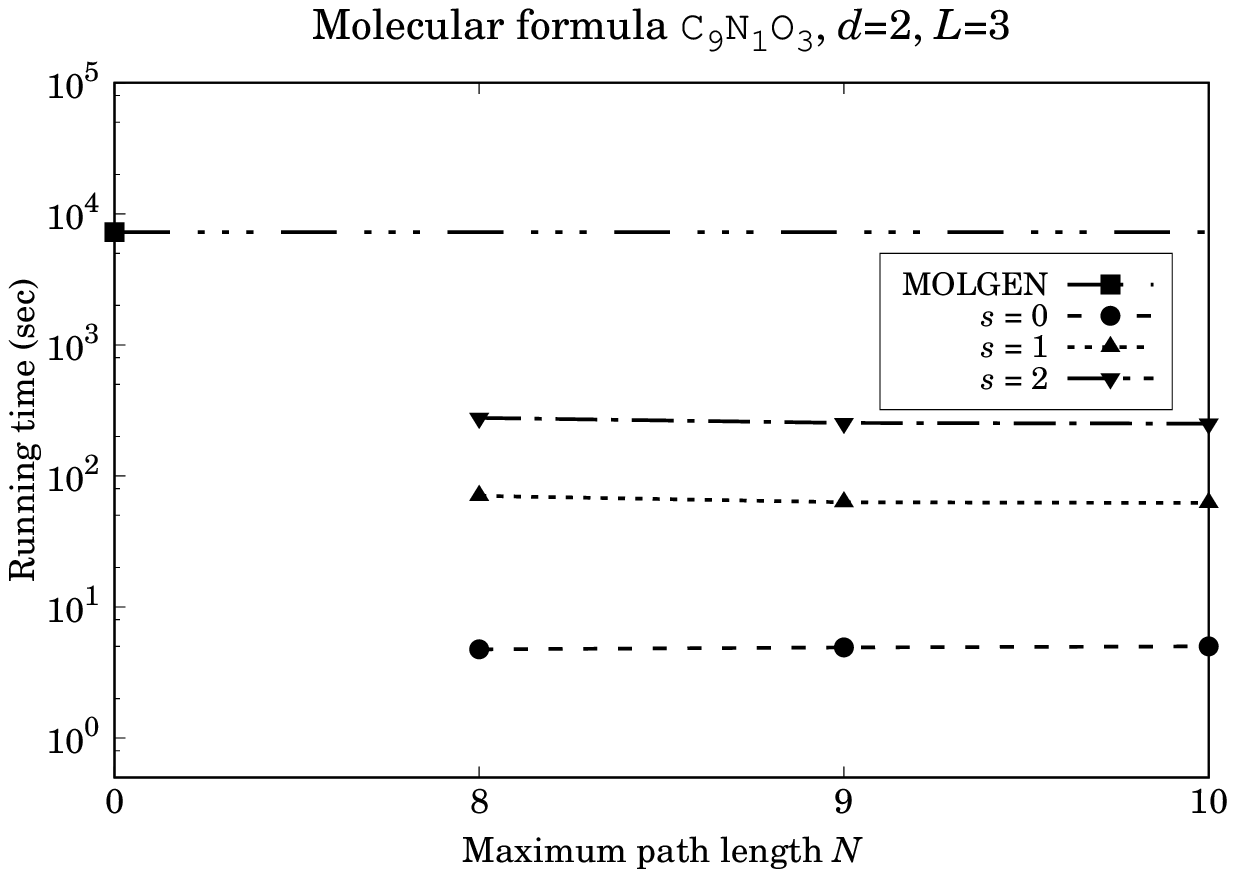}\\
      {\footnotesize (b)}\\
  \end{minipage} 
\hfill
  \begin{minipage}{0.45\textwidth}
   \centering
    \includegraphics[width=1.1\textwidth]{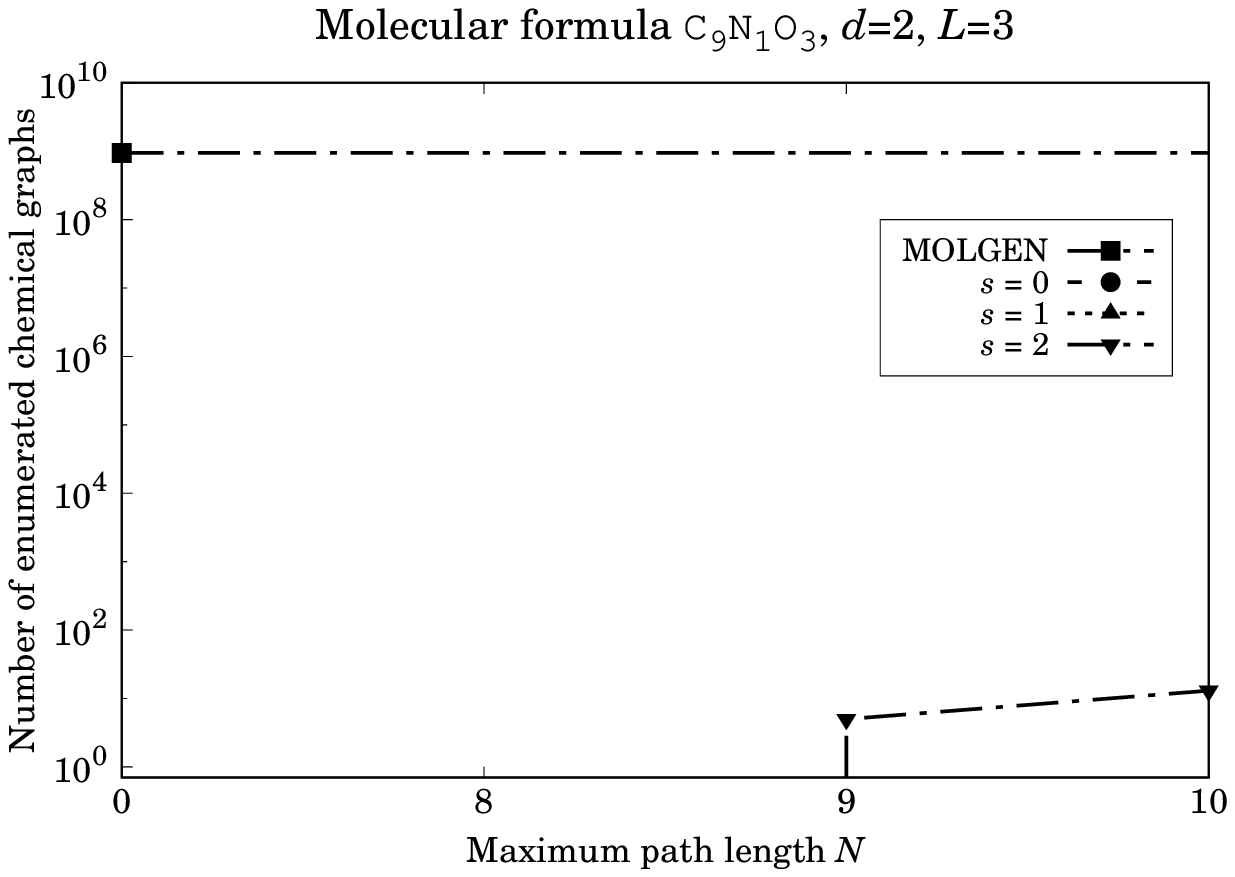}\\
    {\footnotesize (d)}\\
  \end{minipage} 
  \medskip
  
  \vspace{1cm}
  
  \caption{
    Plots showing the computation time 
    and number of chemical graphs enumerated by our algorithm
    for instance type EULF-$L$-P, as compared to MOLGEN.
    The sample structure from PubChem is with CID~301729,
    molecular formula {\tt C$_9$N$_1$O$_3$},
    and maximum bond multiplicity~$d=2$.
    (a), (b)~Running time;
    (c), (d)~Number of enumerated chemical graphs.
  }
 \label{fig:result_graphs_1.2}
 \end{figure}

  \begin{figure}[!ht]
  \begin{minipage}{0.45\textwidth}
   \centering
   \includegraphics[width=1.1\textwidth]{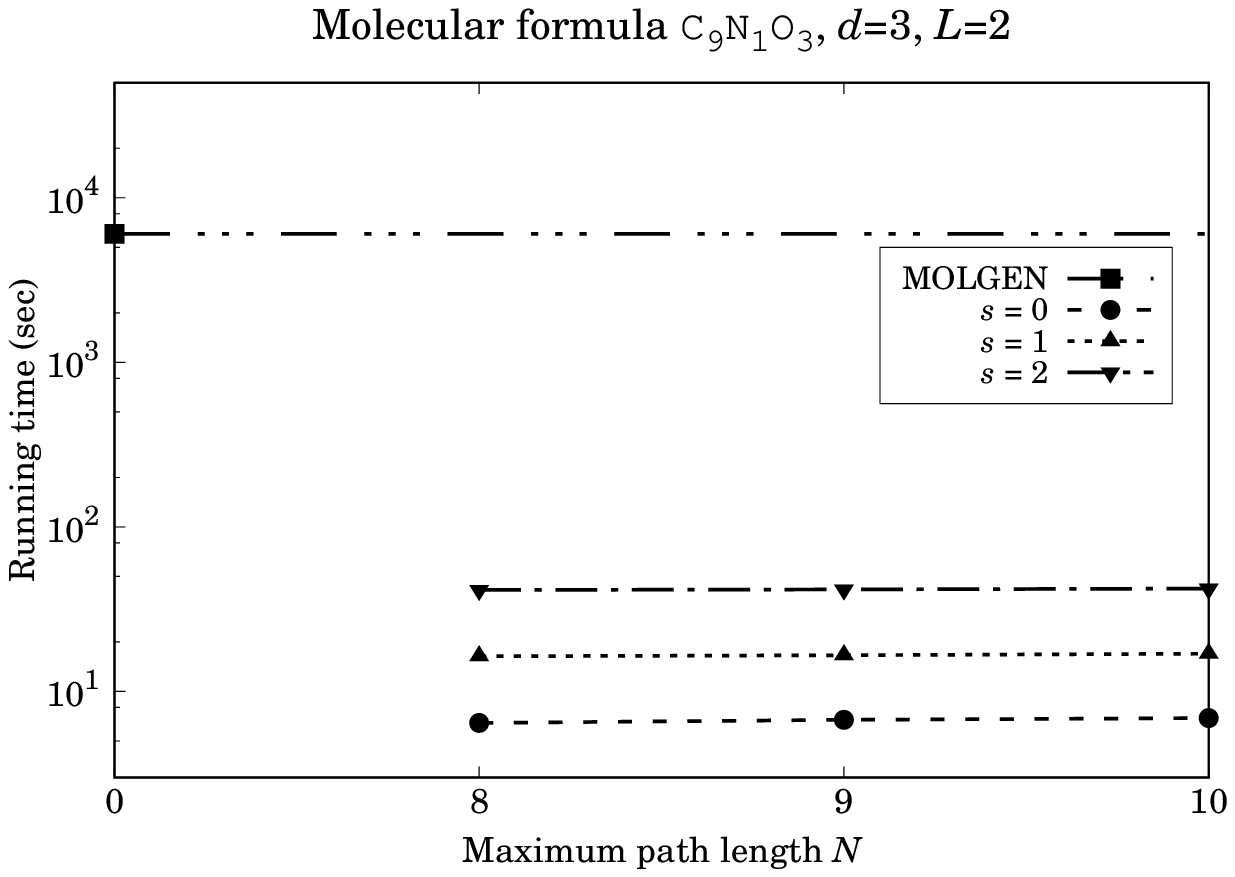}\\
   {\footnotesize (a)}\\
  \end{minipage}
\hfill
  \begin{minipage}{0.45\textwidth}
   \centering
   \includegraphics[width=1.1\textwidth]{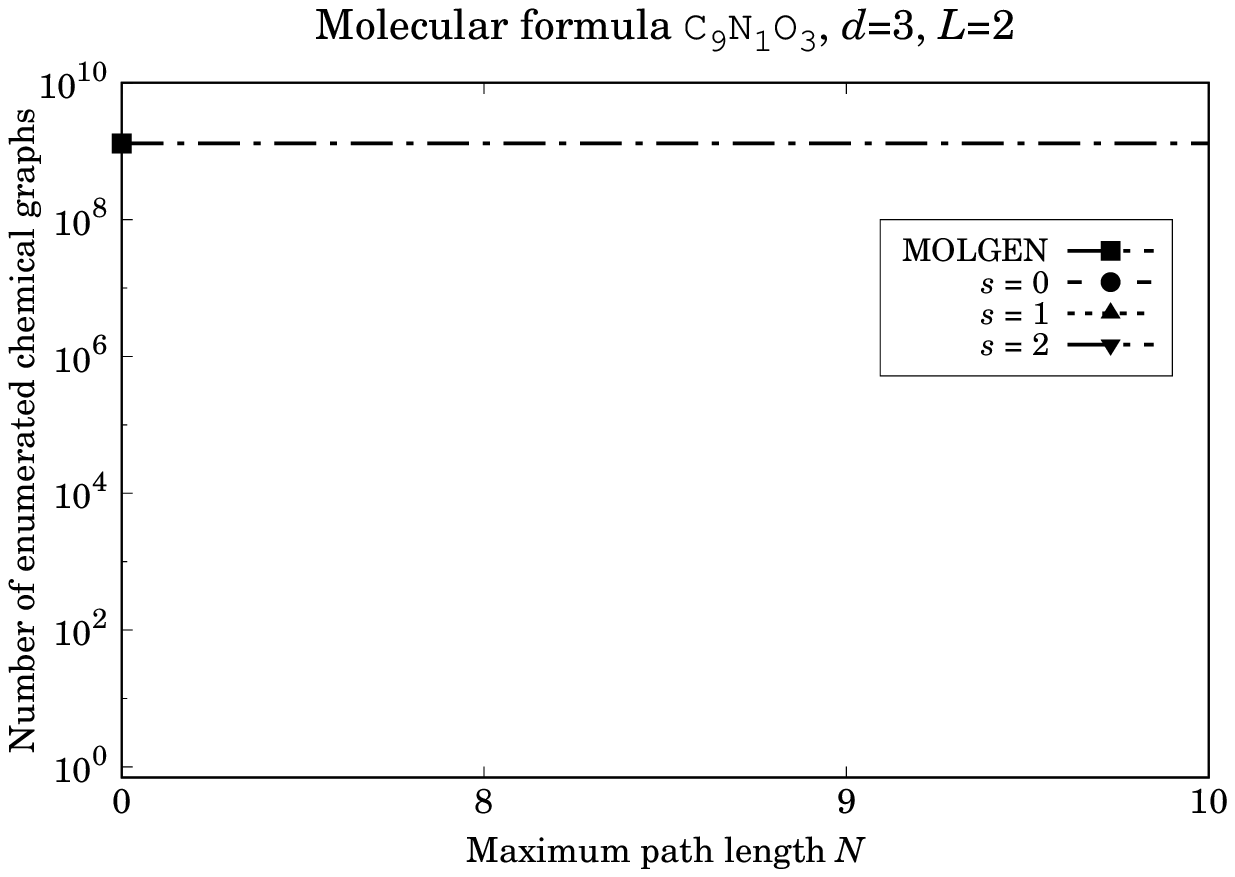}\\
   {\footnotesize (c)}\\
  \end{minipage} 
  \medskip

  \begin{minipage}{0.45\textwidth}
   \centering
      \includegraphics[width=1.1\textwidth]{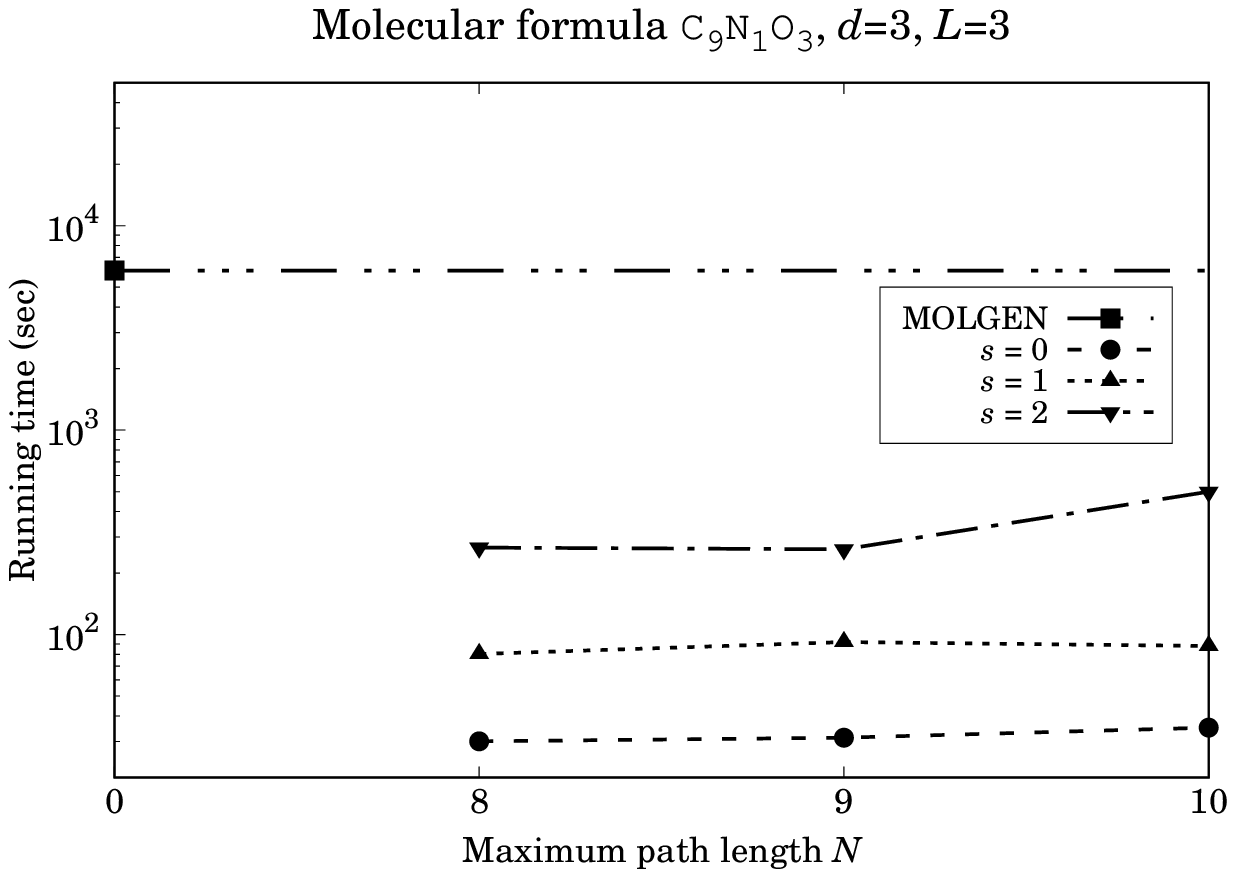}\\
      {\footnotesize (b)}\\
  \end{minipage} 
\hfill
  \begin{minipage}{0.45\textwidth}
   \centering
    \includegraphics[width=1.1\textwidth]{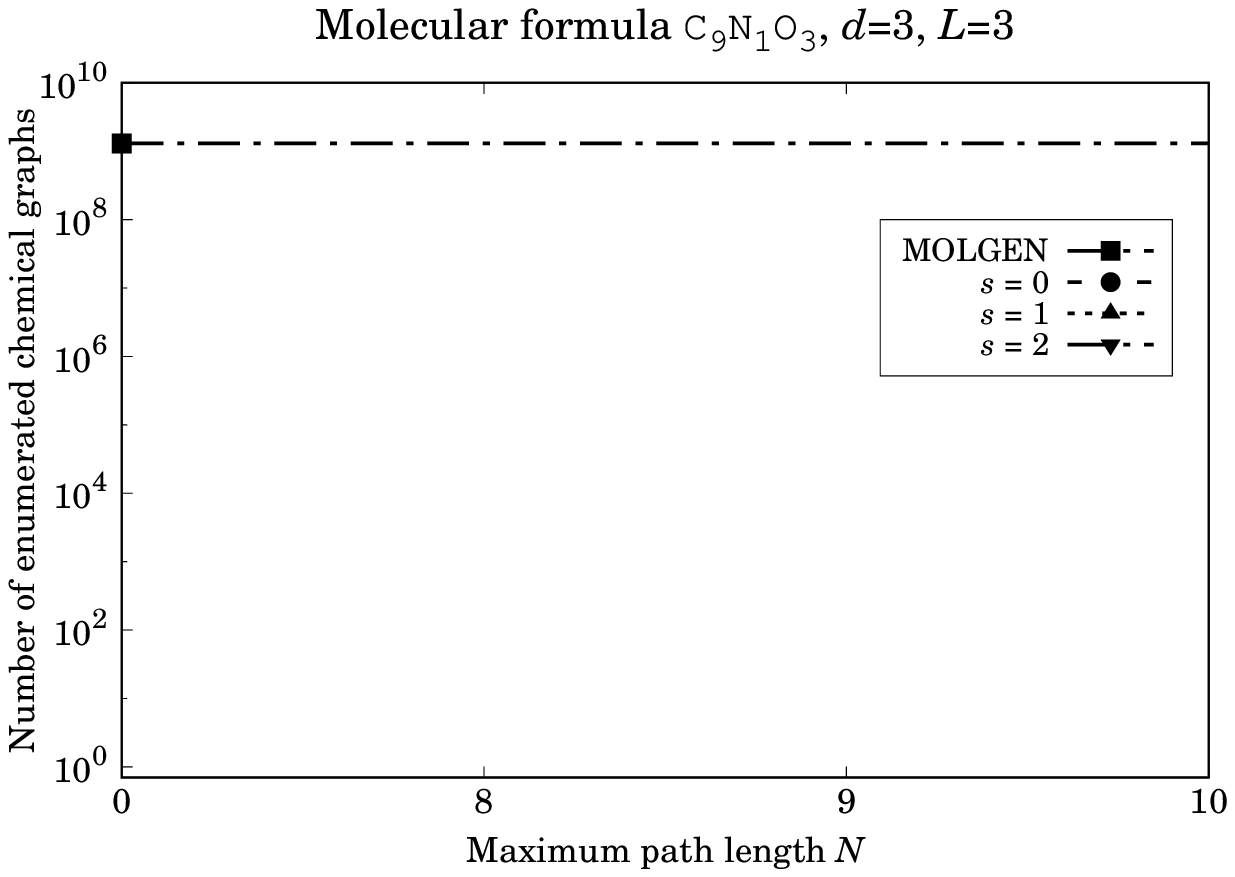}\\
    {\footnotesize (d)}\\
  \end{minipage} 
  \medskip
  
  \vspace{1cm}
  
  \caption{
    Plots showing the computation time 
    and number of chemical graphs enumerated by our algorithm
    for instance type EULF-$L$-P, as compared to MOLGEN.
    The sample structure from PubChem is with CID~57320502,
    molecular formula {\tt C$_9$N$_1$O$_3$},
    and maximum bond multiplicity~$d=3$.
    (a), (b)~Running time;
    (c), (d)~Number of enumerated chemical graphs
    (our algorithm detects that there are no chemical graphs that
    satisfy the given path frequency specification).
  }
 \label{fig:result_graphs_2.2}
 \end{figure}

  \begin{figure}[!ht]
  \begin{minipage}{0.45\textwidth}
   \centering
   \includegraphics[width=1.1\textwidth]{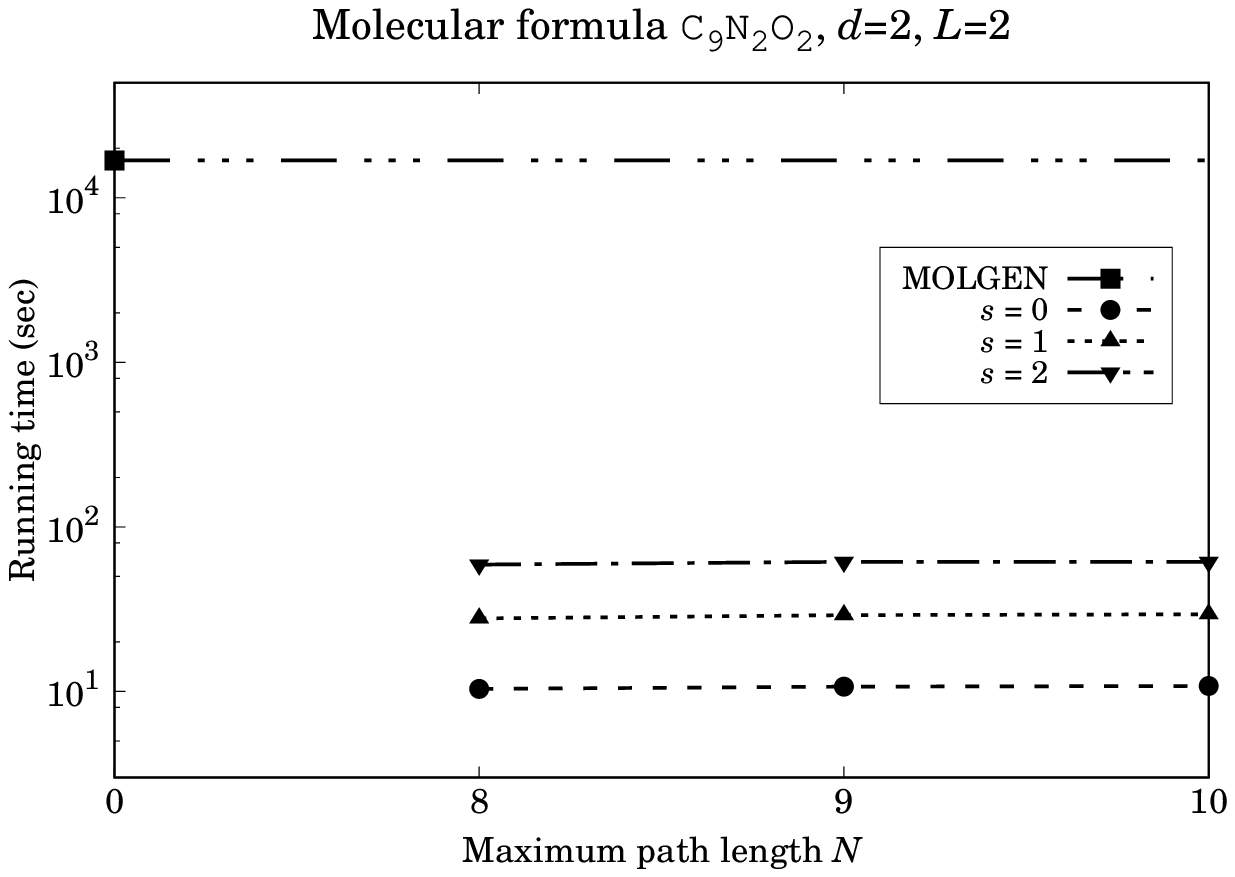}\\
   {\footnotesize (a)}\\
  \end{minipage}
\hfill
  \begin{minipage}{0.45\textwidth}
   \centering
   \includegraphics[width=1.1\textwidth]{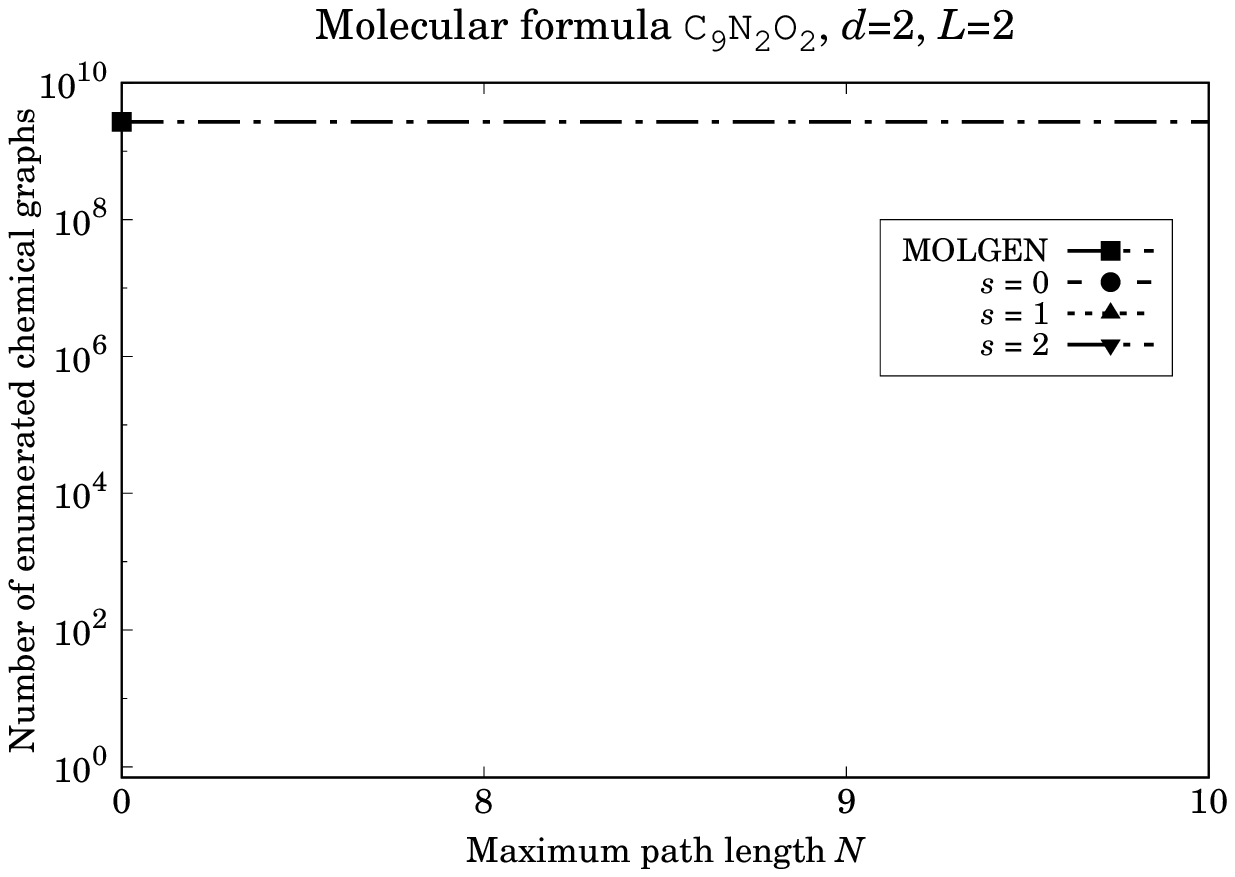}\\
   {\footnotesize (c)}\\
  \end{minipage} 
  \medskip

  \begin{minipage}{0.45\textwidth}
   \centering
      \includegraphics[width=1.1\textwidth]{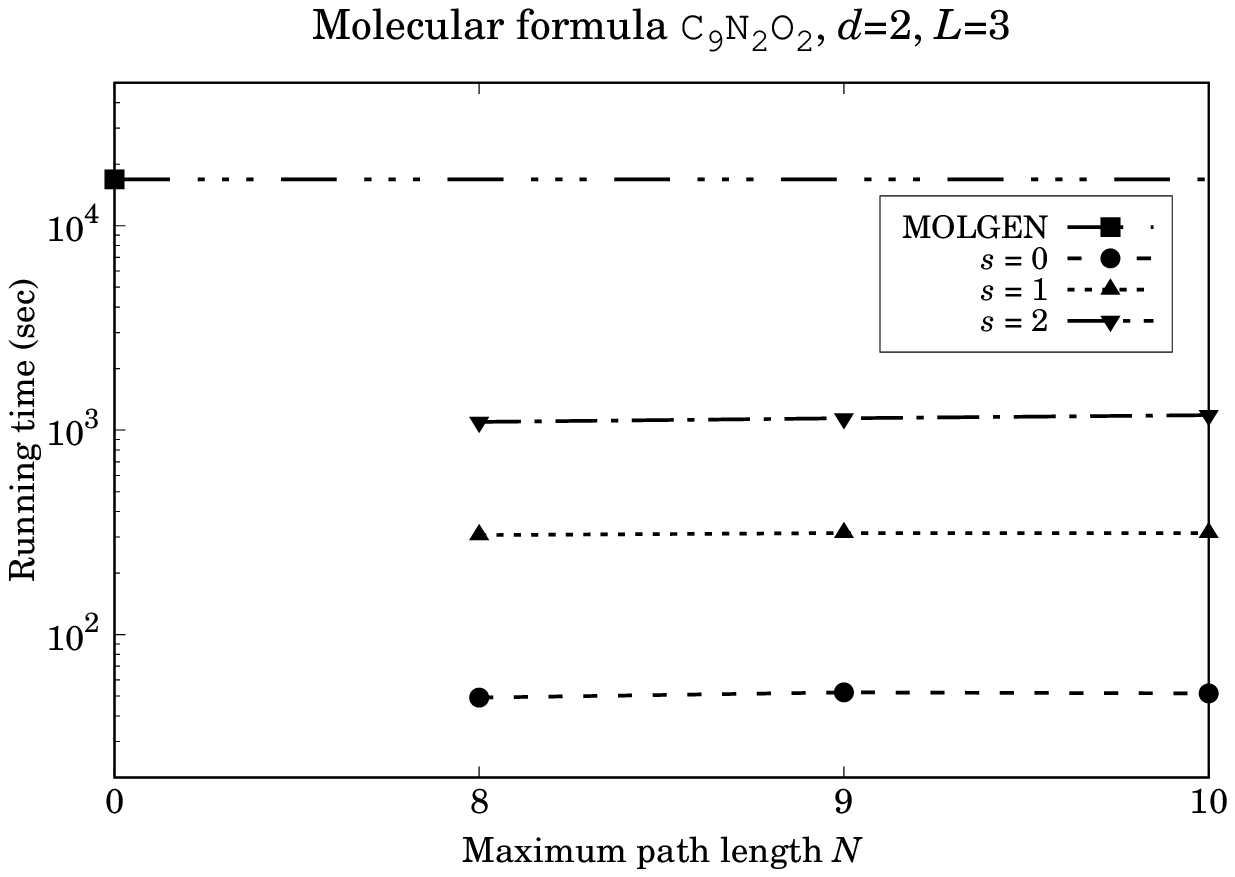}\\
      {\footnotesize (b)}\\
  \end{minipage} 
\hfill
  \begin{minipage}{0.45\textwidth}
   \centering
    \includegraphics[width=1.1\textwidth]{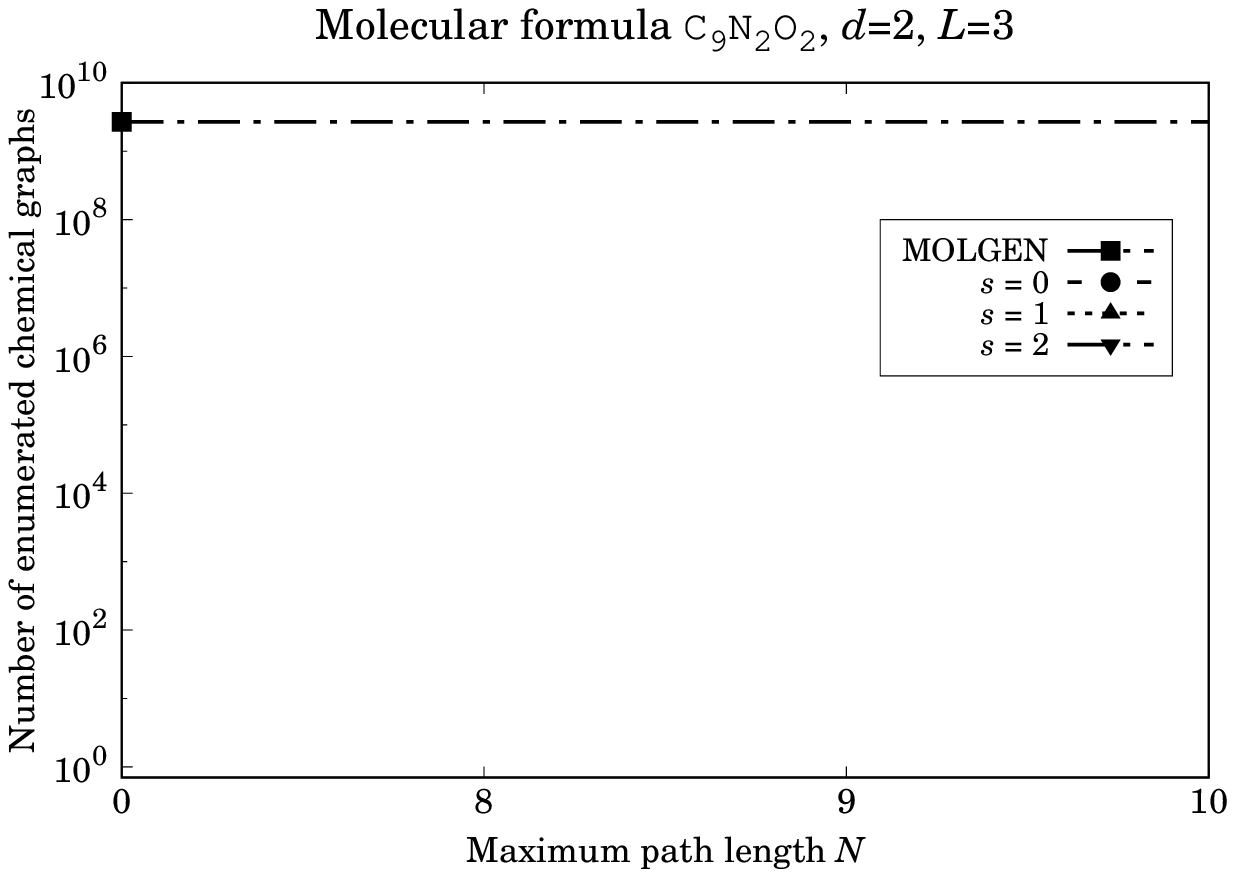}\\
    {\footnotesize (d)}\\
  \end{minipage} 
  \medskip
  
  \vspace{1cm}
  
  \caption{
    Plots showing the computation time 
    and number of chemical graphs enumerated by our algorithm
    for instance type EULF-$L$-P, as compared to MOLGEN.
    The sample structure from PubChem is with CID~6163405,
    molecular formula {\tt C$_9$N$_2$O$_2$},
    and maximum bond multiplicity~$d=2$.
    (a), (b)~Running time;
    (c), (d)~Number of enumerated chemical graphs
    (our algorithm detects that there are no chemical graphs that
    satisfy the given path frequency specification).
  }
 \label{fig:result_graphs_3.2}
 \end{figure}

  \begin{figure}[!ht]
  \begin{minipage}{0.45\textwidth}
   \centering
   \includegraphics[width=1.1\textwidth]{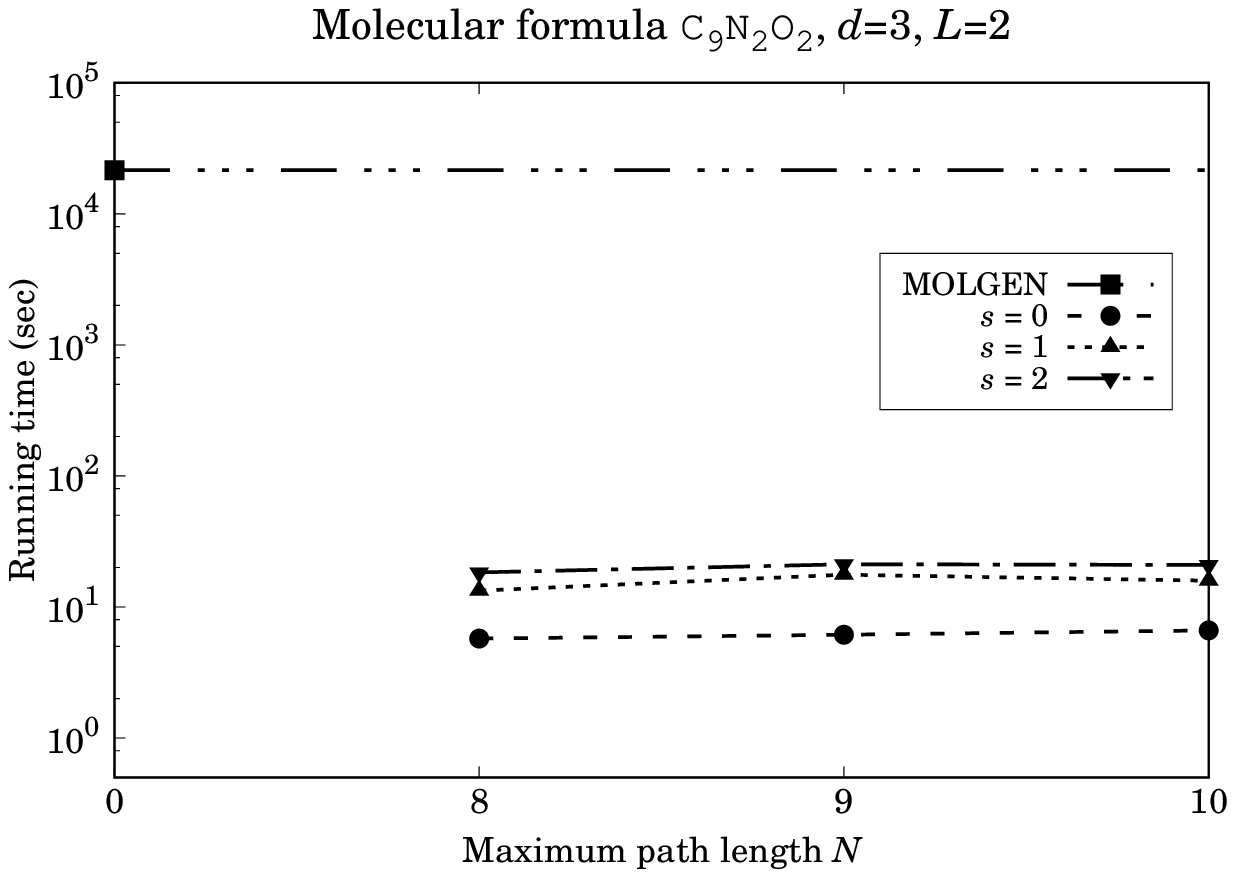}\\
   {\footnotesize (a)}\\
  \end{minipage}
\hfill
  \begin{minipage}{0.45\textwidth}
   \centering
   \includegraphics[width=1.1\textwidth]{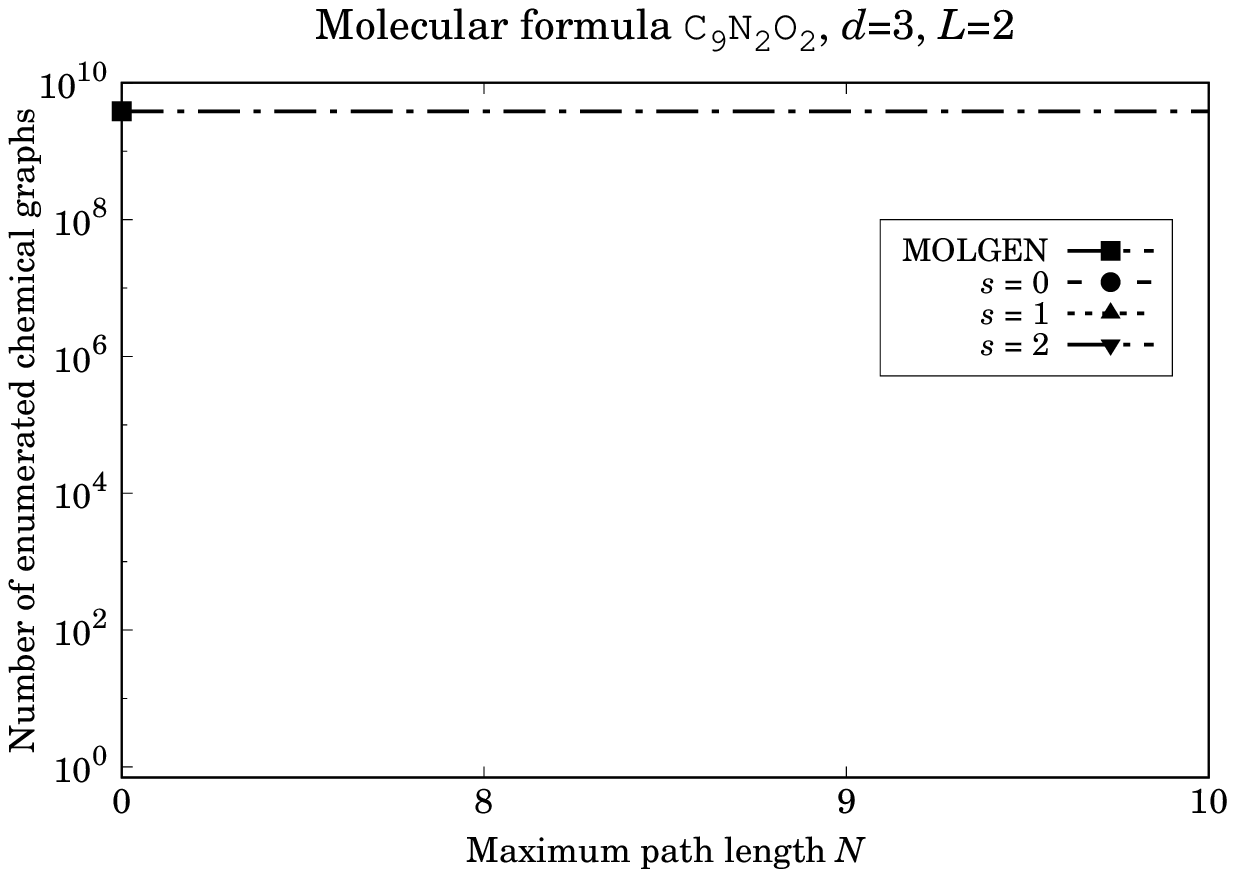}\\
   {\footnotesize (c)}\\
  \end{minipage} 
  \medskip

  \begin{minipage}{0.45\textwidth}
   \centering
      \includegraphics[width=1.1\textwidth]{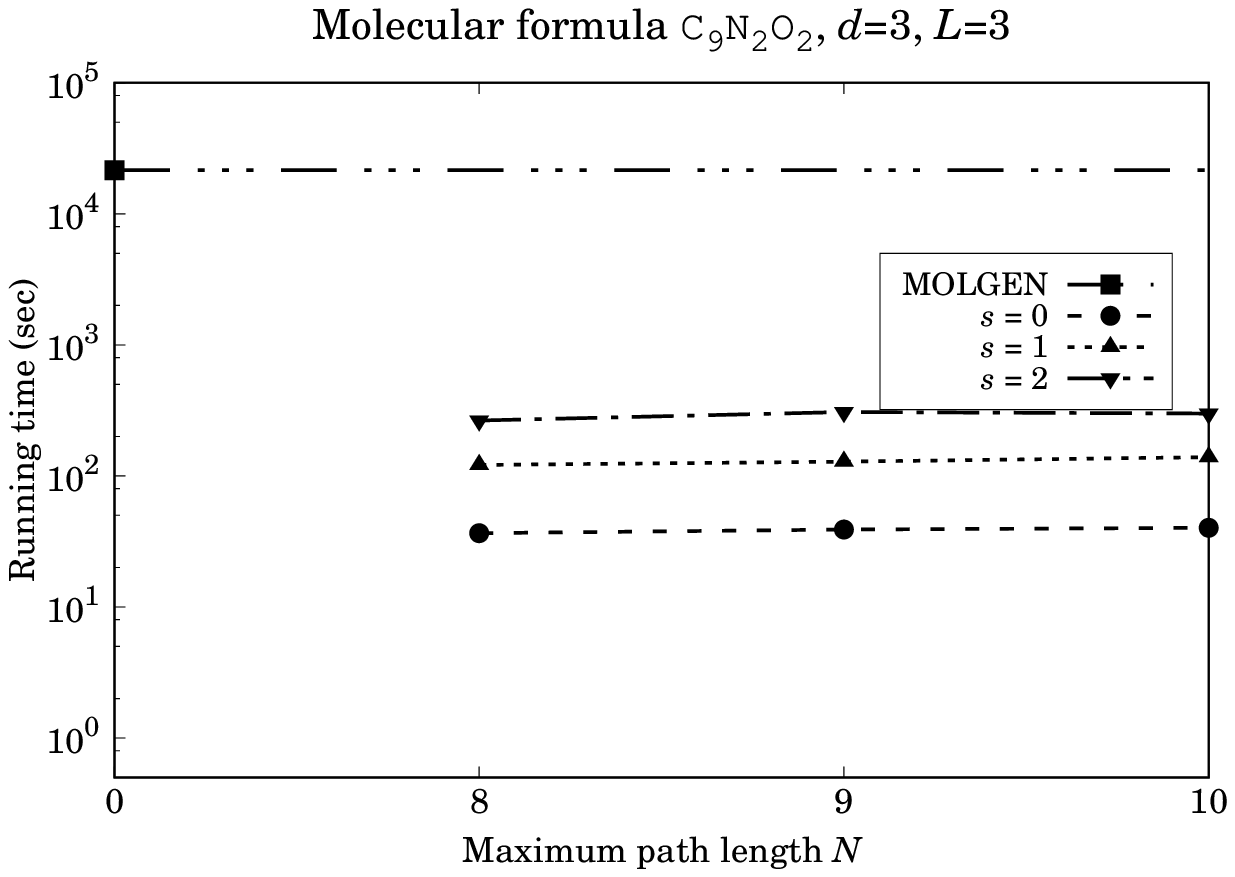}\\
      {\footnotesize (b)}\\
  \end{minipage} 
\hfill
  \begin{minipage}{0.45\textwidth}
   \centering
    \includegraphics[width=1.1\textwidth]{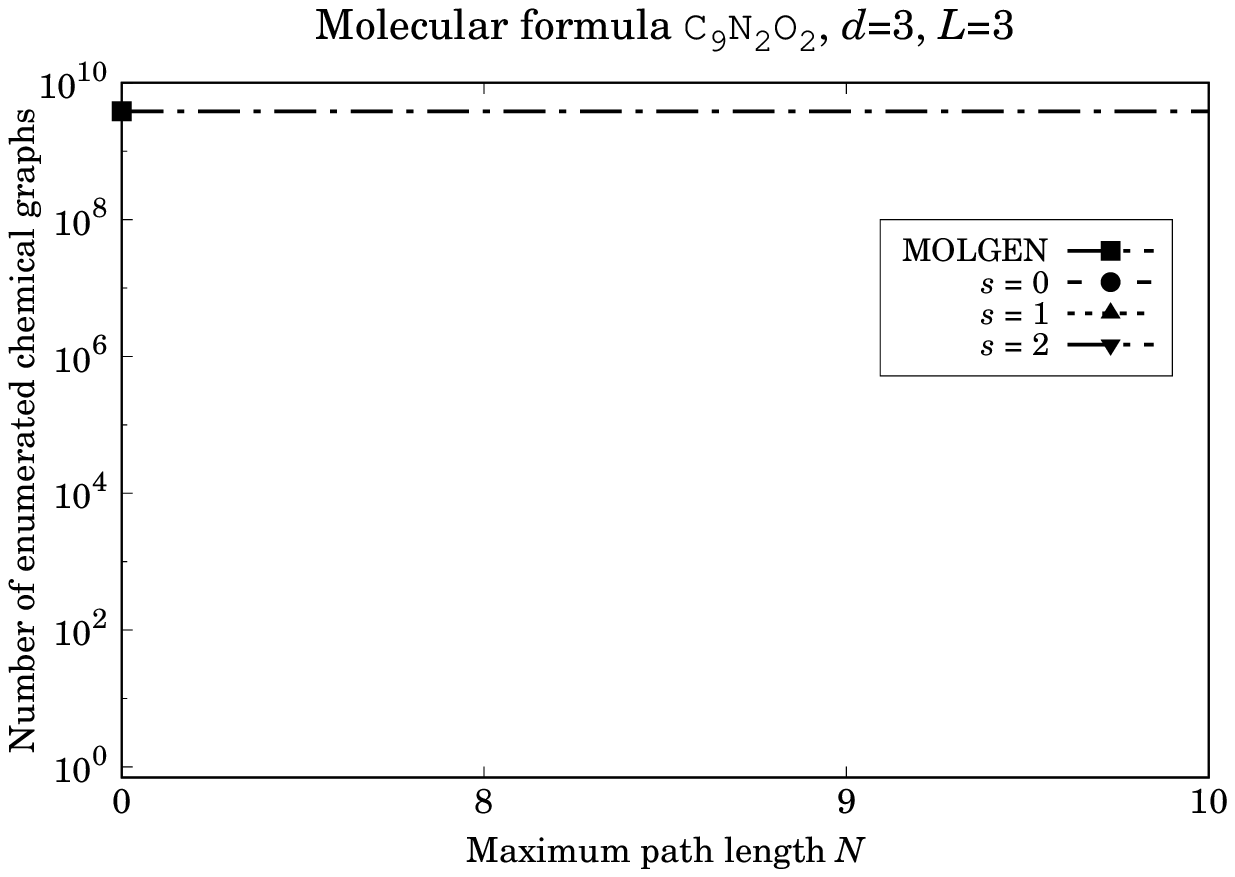}\\
    {\footnotesize (d)}\\
  \end{minipage} 
  \medskip
  
  \vspace{1cm}
  
  \caption{
    Plots showing the computation time 
    and number of chemical graphs enumerated by our algorithm
    for instance type EULF-$L$-P, as compared to MOLGEN.
    The sample structure from PubChem is with CID~131335510,
    molecular formula {\tt C$_9$N$_2$O$_2$},
    and maximum bond multiplicity~$d=3$.
    (a), (b)~Running time;
    (c), (d)~Number of enumerated chemical graphs
    (our algorithm detects that there are no chemical graphs that
    satisfy the given path frequency specification).
  }
 \label{fig:result_graphs_4.2}
 \end{figure}

  \begin{figure}[!ht]
  \begin{minipage}{0.45\textwidth}
   \centering
   \includegraphics[width=1.1\textwidth]{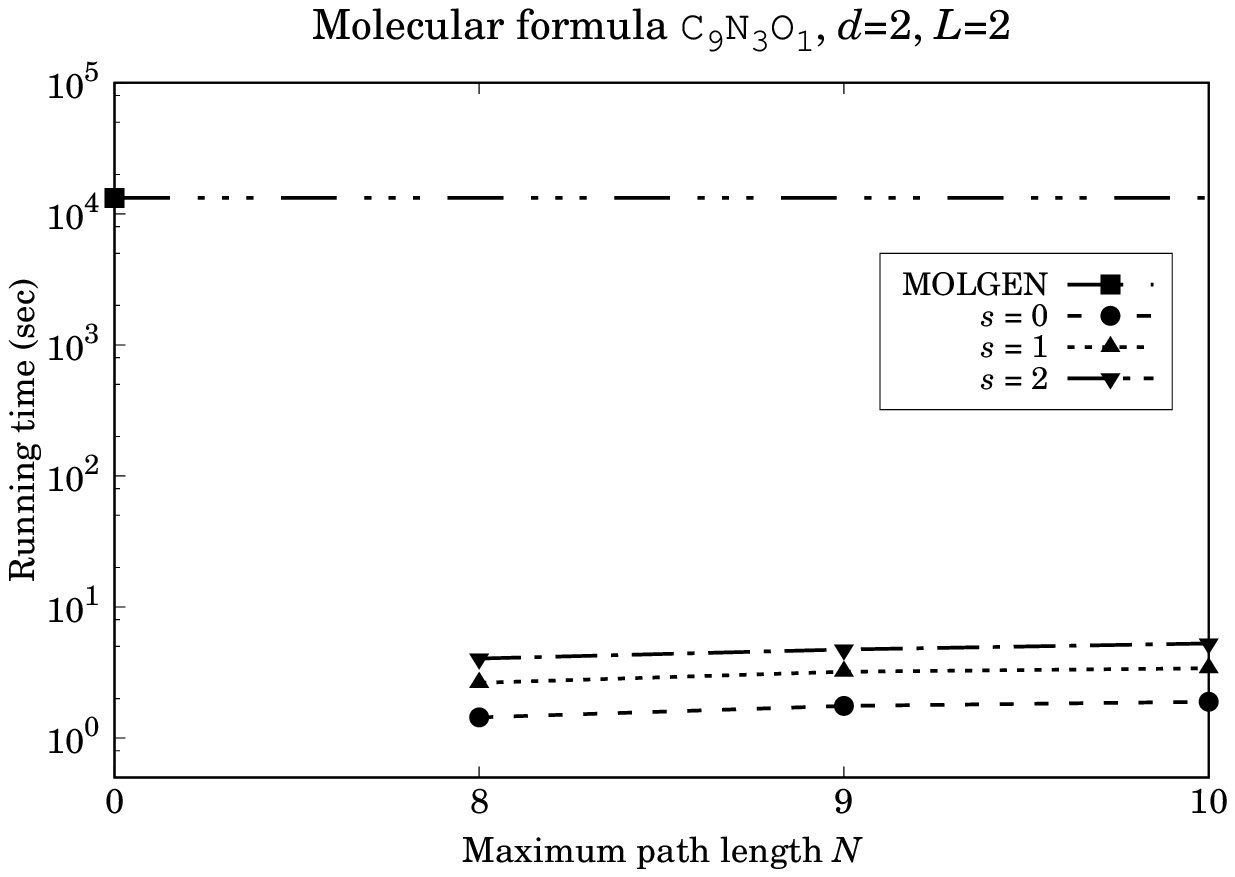}\\
   {\footnotesize (a)}\\
  \end{minipage}
\hfill
  \begin{minipage}{0.45\textwidth}
   \centering
   \includegraphics[width=1.1\textwidth]{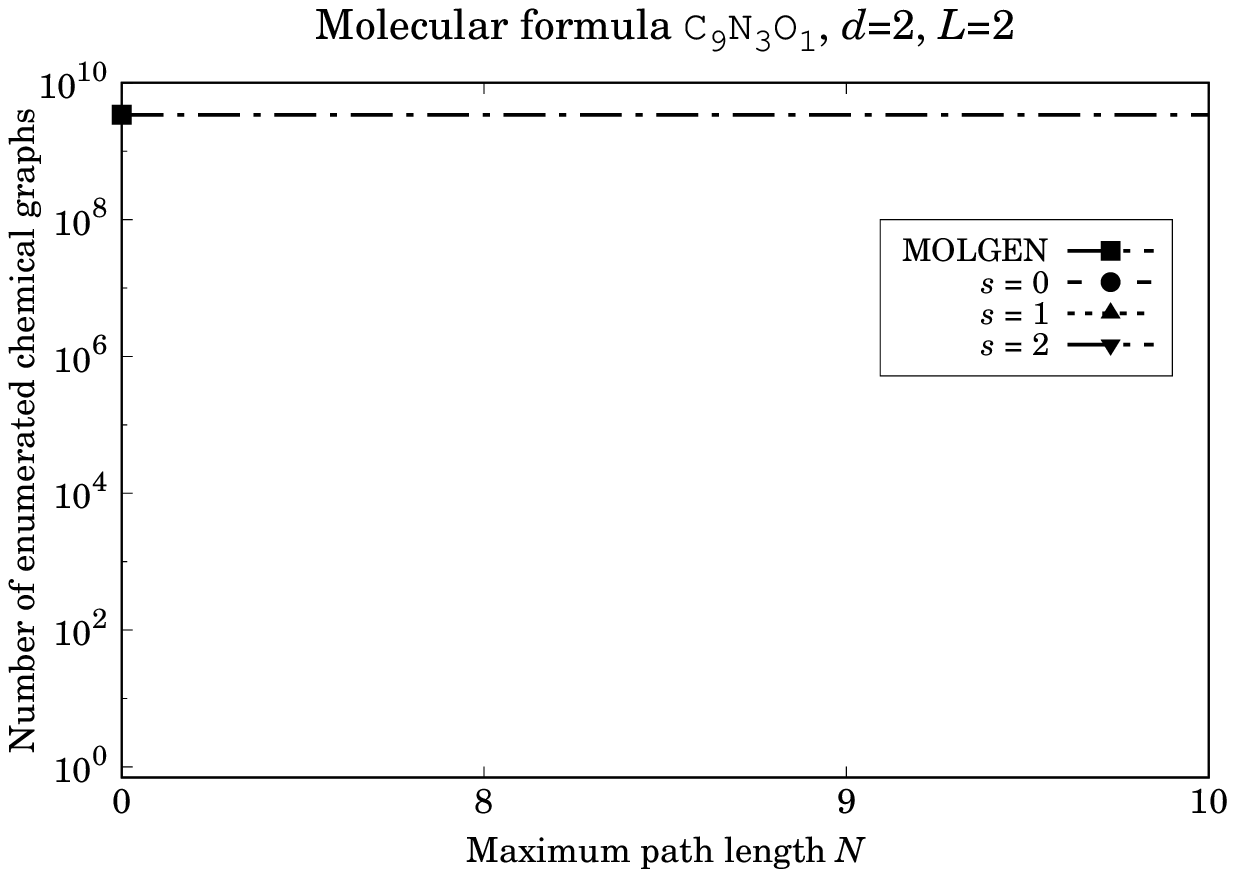}\\
   {\footnotesize (c)}\\
  \end{minipage} 
  \medskip

  \begin{minipage}{0.45\textwidth}
   \centering
      \includegraphics[width=1.1\textwidth]{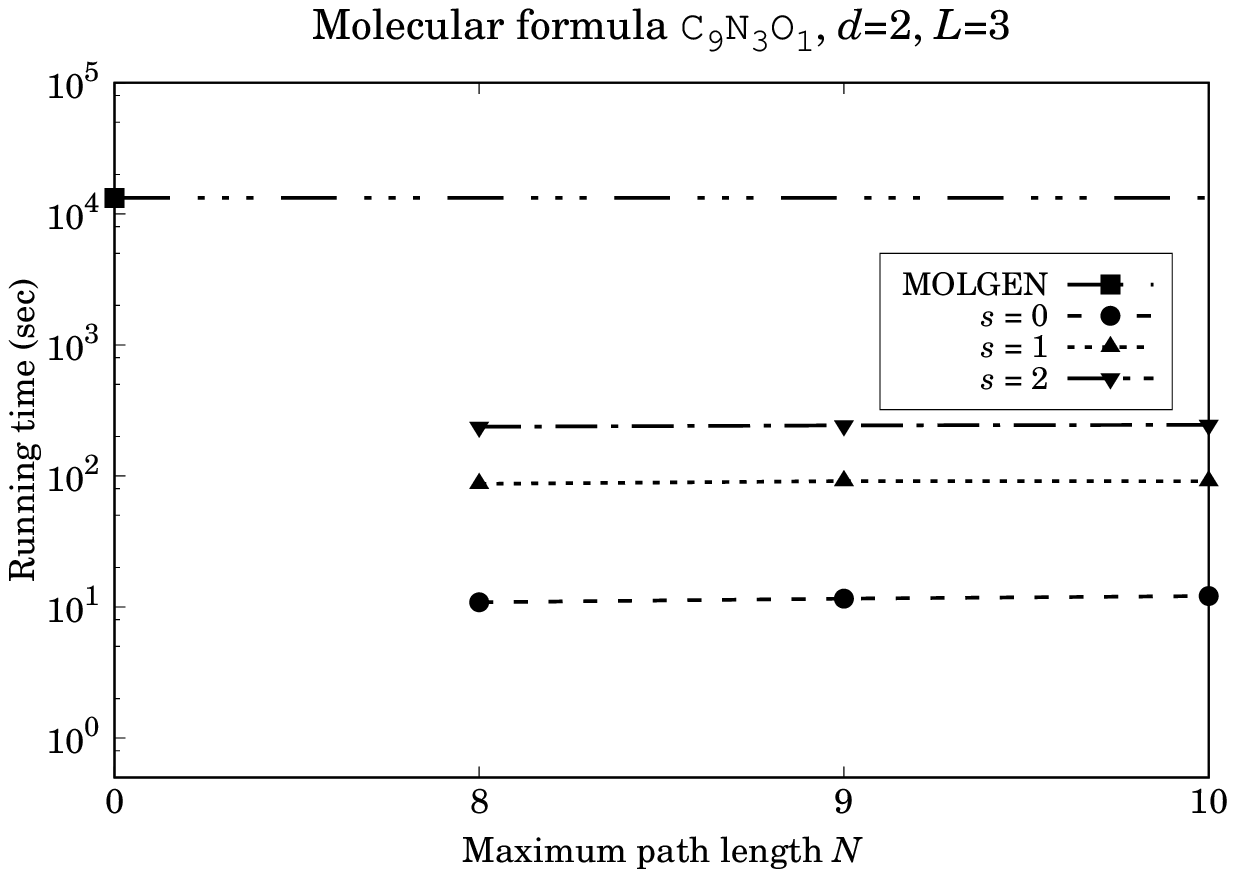}\\
      {\footnotesize (b)}\\
  \end{minipage} 
\hfill
  \begin{minipage}{0.45\textwidth}
   \centering
    \includegraphics[width=1.1\textwidth]{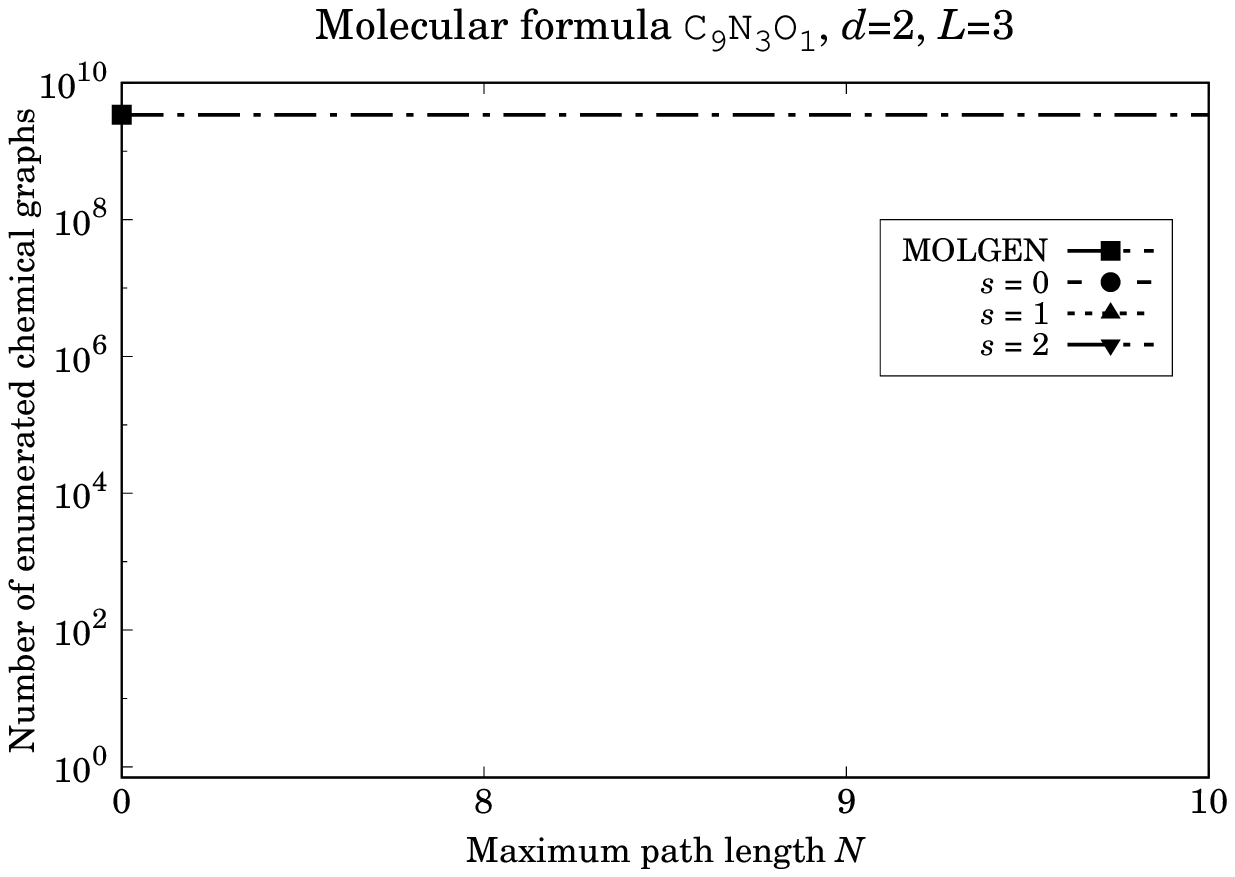}\\
    {\footnotesize (d)}\\
  \end{minipage} 
  \medskip
  
  \vspace{1cm}
  
  \caption{
    Plots showing the computation time 
    and number of chemical graphs enumerated by our algorithm
    for instance type EULF-$L$-P, as compared to MOLGEN.
    The sample structure from PubChem is with CID~9942278,
    molecular formula {\tt C$_9$N$_3$O$_1$},
    and maximum bond multiplicity~$d=2$.
    (a), (b)~Running time;
    (c), (d)~Number of enumerated chemical graphs
    (our algorithm detects that there are no chemical graphs that
    satisfy the given path frequency specification).
  }
 \label{fig:result_graphs_5.2}
 \end{figure}

  \begin{figure}[!ht]
  \begin{minipage}{0.45\textwidth}
   \centering
   \includegraphics[width=1.1\textwidth]{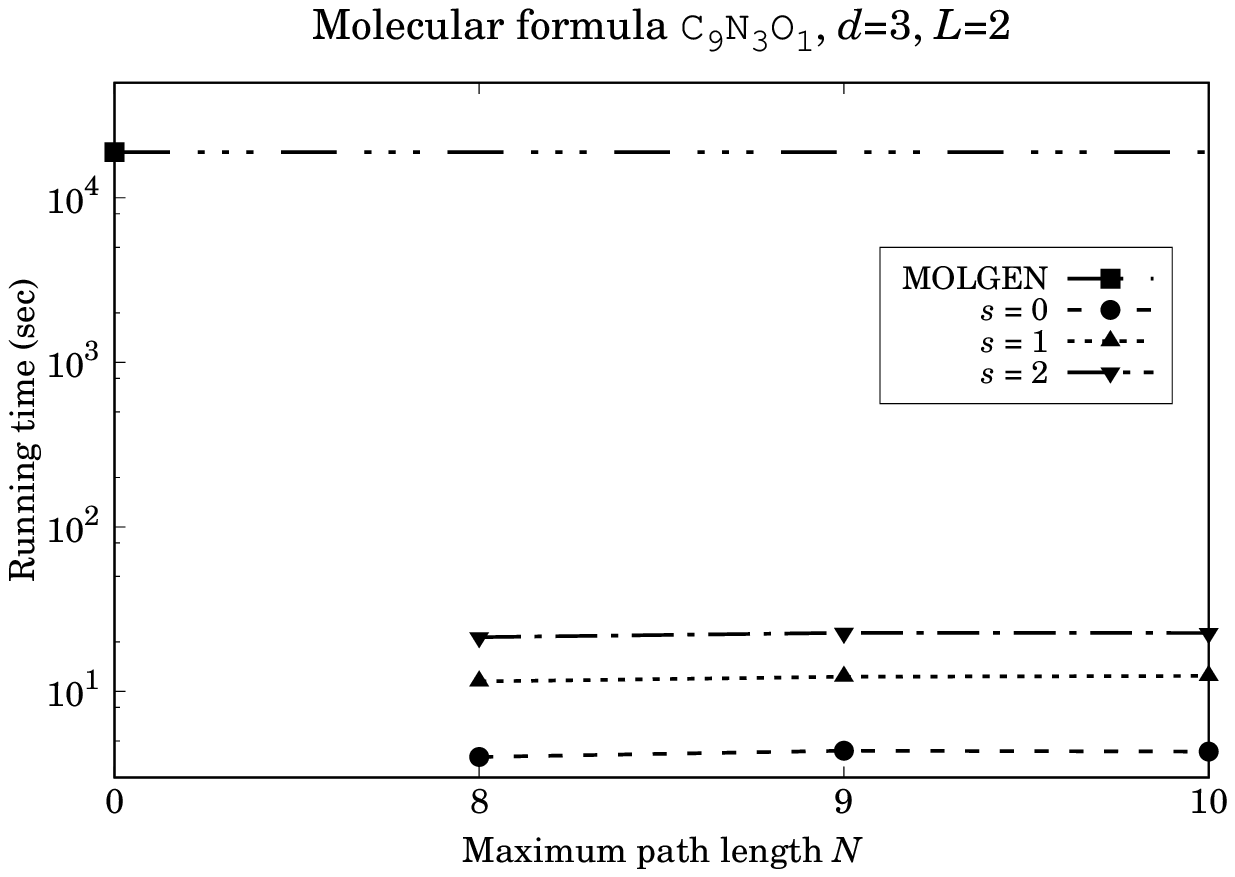}\\
   {\footnotesize (a)}\\
  \end{minipage}
\hfill
  \begin{minipage}{0.45\textwidth}
   \centering
   \includegraphics[width=1.1\textwidth]{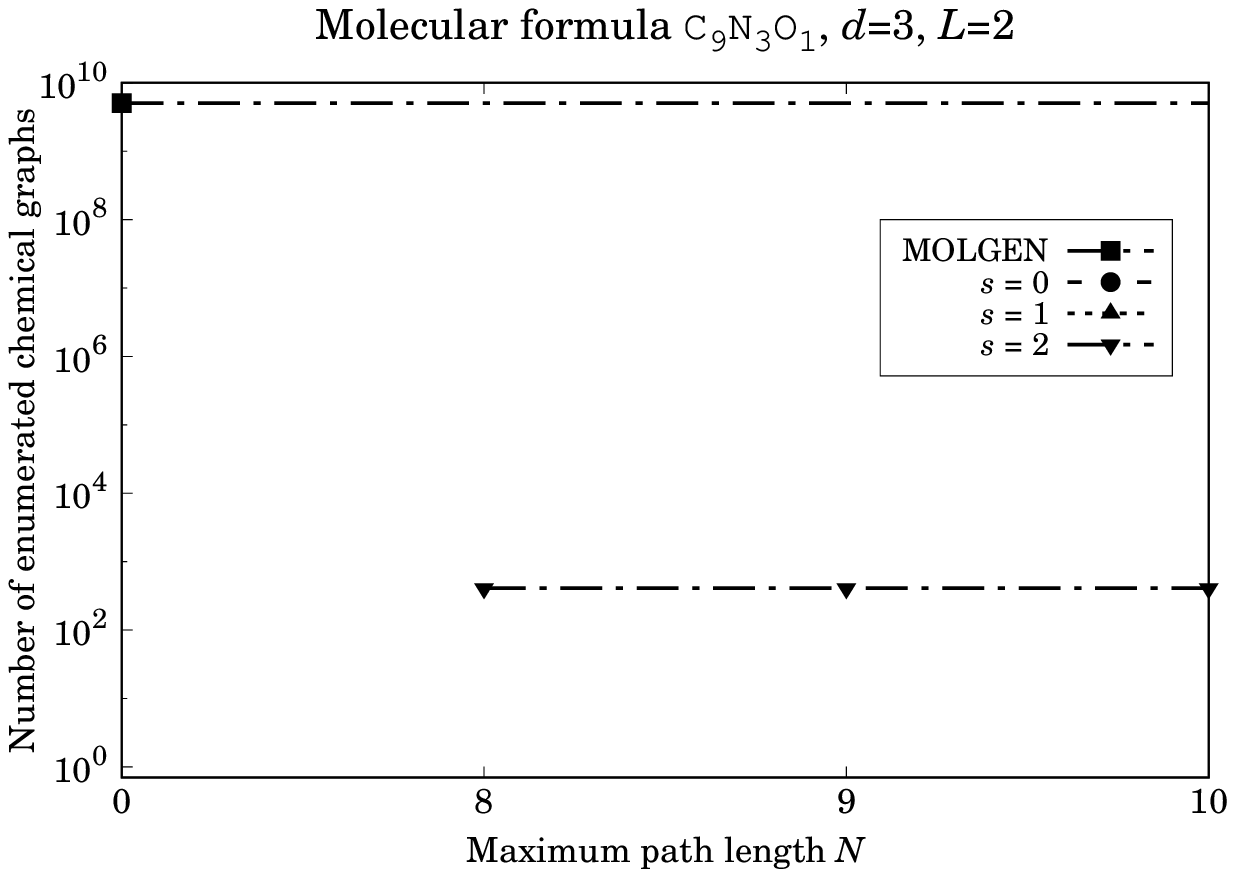}\\
   {\footnotesize (c)}\\
  \end{minipage} 
  \medskip

  \begin{minipage}{0.45\textwidth}
   \centering
      \includegraphics[width=1.1\textwidth]{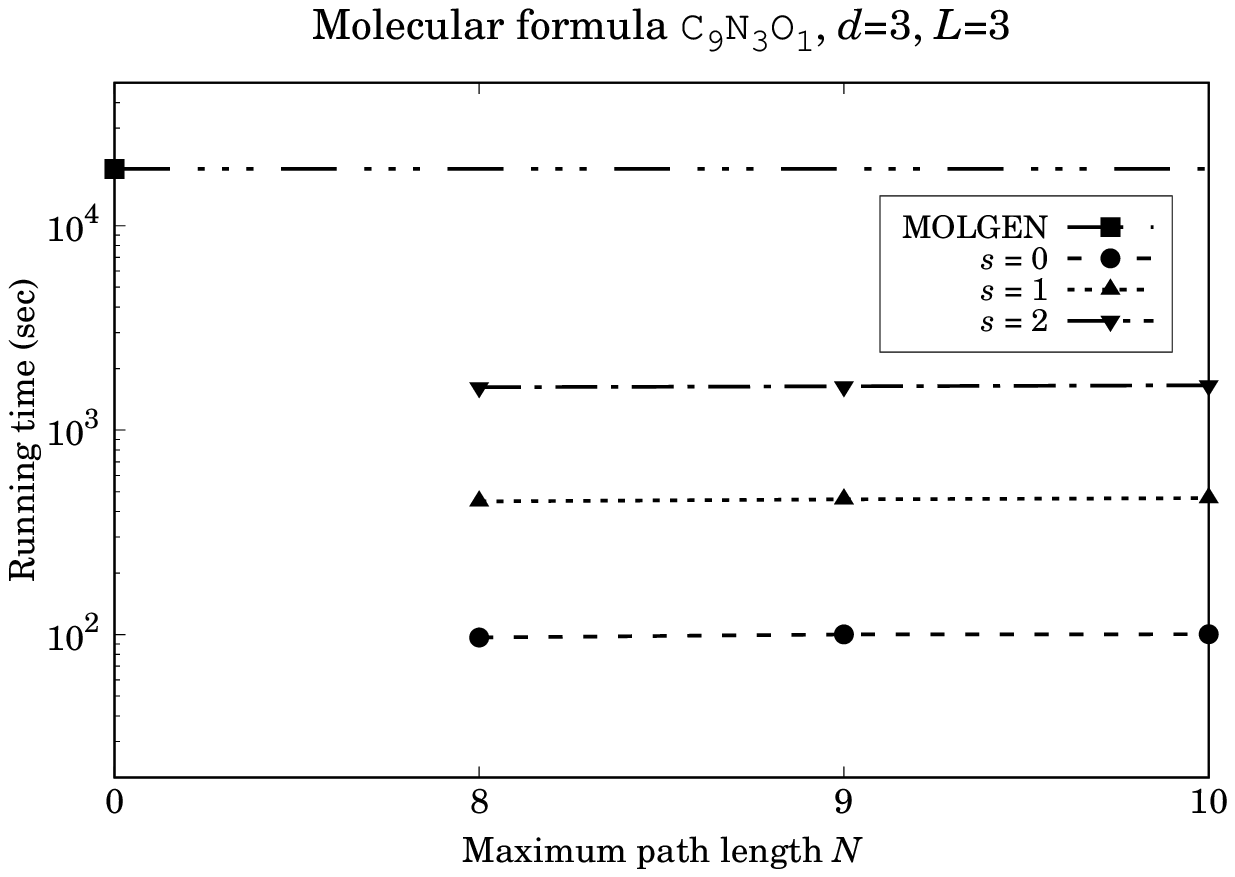}\\
      {\footnotesize (b)}\\
  \end{minipage} 
\hfill
  \begin{minipage}{0.45\textwidth}
   \centering
    \includegraphics[width=1.1\textwidth]{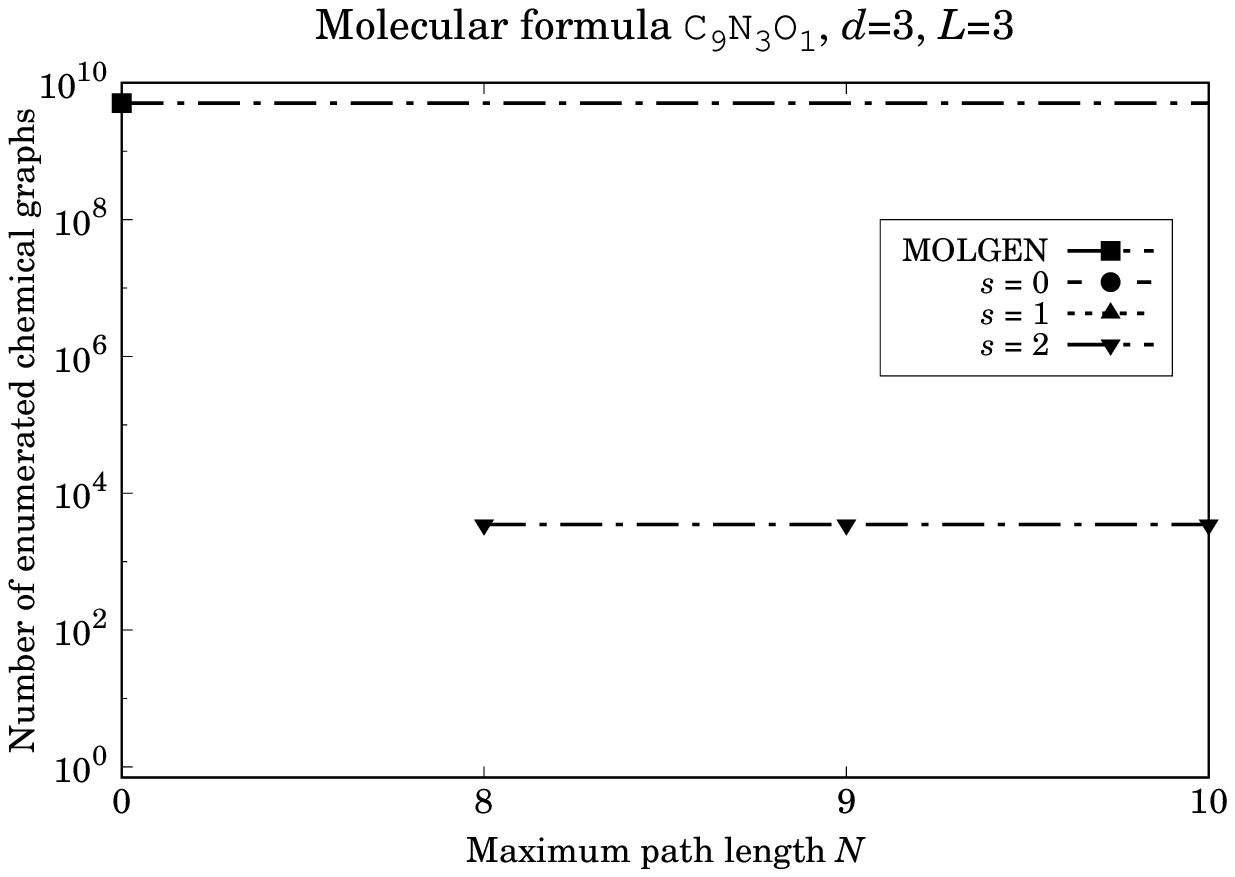}\\
    {\footnotesize (d)}\\
  \end{minipage} 
  \medskip
  
  \vspace{1cm}
  
  \caption{
    Plots showing the computation time 
    and number of chemical graphs enumerated by our algorithm
    for instance type EULF-$L$-P, as compared to MOLGEN.
    The sample structure from PubChem is with CID~10103630,
    molecular formula {\tt C$_9$N$_3$O$_1$},
    and maximum bond multiplicity~$d=3$.
    (a), (b)~Running time;
    (c), (d)~Number of enumerated chemical graphs.
  }
 \label{fig:result_graphs_6.2}
 \end{figure}

\section{Conclusion and Future Work}\label{sec:conclusion}

We formulated two problem
settings
of enumerating 
chemical graphs
that satisfy given lower and upper bounds 
on path frequencies in a given set of paths,
EULF-$L$-A, and EULF-$L$-P.
The problem of enumerating chemical graphs 
has an important practical application
in inverse QSAR/QSPR, and can be used as a part of a framework
for inferring novel chemical structures~\cite{ACZSNA20,CWZSNA20,IWSNA20}
together with a method for solving the inverse problem
on artificial neural networks based on linear programming due to Akutsu and Nagamochi~\cite{AN19}.

We focused on enumerating
chemical graphs with 
a mono-block $2$-augmented tree structure.
We designed a branch-and-bound algorithm for the problem by developing a new procedure 
to add edges between a pair of non-adjacent vertices of a monocyclic graph.
Our  procedure relies on a carefully chosen parent-child relationship between
mono-block $2$-augmented trees and monocyclic graphs to avoid inter-duplication,
and a way of choosing a proper set of non-adjacent vertex pairs in a monocyclic graph,
such that adding edges between each pair in the set will not cause intra-duplication,
nor any possible mono-block $2$-augmented trees to be omitted.

Experimental results reveal that our algorithm 
offers a big advantage in terms of running time 
and the number of generated structures
for instance type EULF-$L$-A
when we have a path frequency specification over using
MOLGEN~\cite{MOLGEN5} to generate chemical graphs 
with a particular molecular formula.
Namely, while MOLGEN may produce on the order of billions of chemical graphs
with 2-augmented tree structure with a particular chemical formula and maximum bond multiplicity,
for a given path specification our algorithm produces much fewer structures,
and we also have the advantage to generate
only mono-block structures.

However, for instance type EULF-$L$-P, the experimental results reveal that
our algorithm takes much time to finish 
even when there are no chemical graphs that satisfy a given path frequency specification.
It would be very interesting to equip our algorithm with a procedure
that detects this situation much earlier in the computation process,
or even design an algorithm based on a different idea - namely one that starts
building a chemical graph from one of the paths with a non-zero lower bound in a given set.


\end{document}